\documentclass{article}

\usepackage{etoolbox}
\newtoggle{arxivformat}

\toggletrue{arxivformat}

\nottoggle{arxivformat}{}

\iftoggle{arxivformat}{
  \usepackage{natbib}
  \bibliographystyle{plainnat}

  \newcommand{\BlackBox}{\rule{1.5ex}{1.5ex}}    \newenvironment{proof}{\par\noindent{\bf Proof\ }}{\hfill\BlackBox\\[2mm]}
  \newtheorem{theorem}{Theorem}
  \newtheorem{proposition}[theorem]{Proposition}

  \usepackage[colorlinks=false,allbordercolors={1 1 1}]{hyperref}
} {
    \usepackage{jmlr2e}
}

\usepackage[]{graphicx}
\usepackage[]{color}
\makeatletter
\def\maxwidth{   \ifdim\Gin@nat@width>\linewidth
    \linewidth
  \else
    \Gin@nat@width
  \fi
}
\makeatother

\definecolor{fgcolor}{rgb}{0.345, 0.345, 0.345}

\usepackage{framed}


\definecolor{shadecolor}{rgb}{.97, .97, .97}
\definecolor{messagecolor}{rgb}{0, 0, 0}
\definecolor{warningcolor}{rgb}{1, 0, 1}
\definecolor{errorcolor}{rgb}{1, 0, 0}
\newenvironment{knitrout}{}{} 
\usepackage{alltt}
 
\usepackage{listings}
\usepackage{prettyref}
\usepackage{amsmath}
\usepackage{amssymb}
\usepackage{xargs}[2008/03/08]
\usepackage{lastpage}
\usepackage{times}

\newtheorem{assumption}{Assumption}
\newtheorem{condition}{Condition}

\newtheorem{lrvb_lemma}{Lemma}
\newtheorem{lrvb_corollary}{Corollary}
\newtheorem{lrvb_definition}{Definition}

\usepackage{appendix}

\nottoggle{arxivformat}{    \ShortHeadings{Covariances, Robustness, and Variational Bayes}
    {Giordano, Broderick, and Jordan}
  \firstpageno{1}
  \editor{Mohammad Emtiyaz Khan}
}

\newcommand{\fig}[1]{Fig.~\ref{fig:#1}}

\newrefformat{app}{Appendix \ref{#1}}
\newrefformat{eq}{Eq.~(\ref{#1})}
\newrefformat{assu}{Assumption~\ref{#1}}
\newrefformat{subsec}{Section~\ref{#1}}
\newrefformat{cor}{Corollary~\ref{#1}}
\newrefformat{thm}{Theorem~\ref{#1}}
\newrefformat{def}{Definition~\ref{#1}}
\newrefformat{prop}{Proposition~\ref{#1}}
\newrefformat{tab}{Table~\ref{#1}}
\newrefformat{cond}{Condition~\ref{#1}}

\begin{document}

\nottoggle{arxivformat}{\jmlrheading{19}{2018}
    {1-\pageref{LastPage}}{11/17; Revised 7/18}{8/18}{17-670}
    {Ryan Giordano, Tamara Broderick, and Michael I.~Jordan}
}

\title{Covariances, Robustness, and Variational Bayes}

\iftoggle{arxivformat}{
\author{    Ryan Giordano \\
    \texttt{rgiordano@berkeley.edu} \\
    Department of Statistics, UC Berkeley\\
          367 Evans Hall, UC Berkeley\\
          Berkeley, CA 94720
    \and
    Tamara Broderick \\
    \texttt{tbroderick@csail.mit.edu}\\
    Department of EECS, MIT\\
    77 Massachusetts Ave., 38-401\\
    Cambridge, MA 02139
    \and
    Michael I.~Jordan \\
    \texttt{jordan@cs.berkeley.edu} \\
    Department of Statistics and EECS, UC Berkeley\\
    367 Evans Hall, UC Berkeley\\
    Berkeley, CA 94720
}
}{
\author{    \name Ryan Giordano
    \email rgiordano@berkeley.edu \\
    \addr Department of Statistics, UC Berkeley\\
          367 Evans Hall, UC Berkeley\\
          Berkeley, CA 94720
    \AND
    \name Tamara Broderick
    \email tbroderick@csail.mit.edu\\
    \addr Department of EECS, MIT\\
          77 Massachusetts Ave., 38-401\\
          Cambridge, MA 02139
    \AND
    \name Michael I.~Jordan \email jordan@cs.berkeley.edu \\
    \addr Department of Statistics and EECS, UC Berkeley\\
          367 Evans Hall, UC Berkeley\\
          Berkeley, CA 94720
}
}

\maketitle

\begin{abstract}Mean-field Variational Bayes (MFVB) is an approximate Bayesian posterior
inference technique that is increasingly popular due to its fast runtimes on
large-scale data sets. However, even when MFVB provides accurate posterior means
for certain parameters, it often mis-estimates variances and covariances.
Furthermore, prior robustness measures have remained undeveloped for MFVB. By
deriving a simple formula for the effect of infinitesimal model perturbations on
MFVB posterior means, we provide both improved covariance estimates and local
robustness measures for MFVB, thus greatly expanding the practical usefulness of
MFVB posterior approximations. The estimates for MFVB posterior covariances rely
on a result from the classical Bayesian robustness literature that relates
derivatives of posterior expectations to posterior covariances and includes the
Laplace approximation as a special case. Our key condition is that the MFVB
approximation provides good estimates of a select subset of posterior
means---an assumption that has been shown to hold in many practical settings. In our
experiments, we demonstrate that our methods are simple, general, and fast,
providing accurate posterior uncertainty estimates and robustness measures with
runtimes that can be an order of magnitude faster than MCMC.
 \end{abstract}

\global\long\def\constant{Constant}

\newcommandx\tconst[1][usedefault, addprefix=\global, 1=\alpha]{C_{#1}}

\global\long\def\trans{\intercal}

\global\long\def\mbe{\mathbb{E}}

\global\long\def\indep{\stackrel{indep}{\sim}}

\global\long\def\iid{\stackrel{iid}{\sim}}

\global\long\def\kl{\mathrm{KL}}

\global\long\def\qmfvb{\mathcal{Q}_{mf}}

\global\long\def\cov{\mathrm{Cov}}

\global\long\def\var{\mathrm{Var}}

\global\long\def\normal{\mathcal{N}}

\global\long\def\at#1{\vert_{#1}}

\global\long\def\argmin{\operatornamewithlimits{argmin}}

\global\long\def\argmax{\operatornamewithlimits{argmax}}

\newcommandx\pthetapost[1][usedefault, addprefix=\global, 1=\alpha]{p_{#1}}

\global\long\def\alphazeroball{\mathcal{A}_{0}}

\newcommandx\qthetapost[1][usedefault, addprefix=\global, 1=\alpha]{q_{#1}}

\global\long\def\pzeropost{\pthetapost[0]}

\global\long\def\qzeropost{\qthetapost[0]}

\global\long\def\etaopt{\eta^{*}}

\global\long\def\etaoptzero{\eta_{0}^{*}}

\global\long\def\etatopt{\eta^{*}\left(\alpha\right)}

\global\long\def\qthetapostarg{\qthetapost\left(\theta;\eta^{*}\right)}

\global\long\def\qtthetapostarg{\qthetapost\left(\theta;\eta^{*}\left(\alpha\right)\right)}

\newcommandx\gtheta[1][usedefault, addprefix=\global, 1=\theta]{g\left(#1\right)}

\global\long\def\mbeq{\mbe_{\qthetapost}}

\global\long\def\mbep{\mbe_{\pthetapost}}

\global\long\def\epgtheta{\mbe_{\pthetapost}\left[\gtheta\right]}

\global\long\def\eqgtheta{\mbe_{\qthetapost}\left[\gtheta\right]}

\global\long\def\lrvbcov{\mathrm{Cov}_{\qthetapost[0]}^{LR}}

\global\long\def\qcov{\cov_{\qzeropost}}

\global\long\def\pcov{\cov_{\pzeropost}}

\global\long\def\klhess{\mathbf{H}_{\eta\eta}}

\global\long\def\efhess{\mathbf{f}_{\alpha\eta}}

\global\long\def\eggrad{\mathbf{g}_{\eta}}

\global\long\def\covmat{\boldsymbol{\Sigma}}

\global\long\def\infomat{\boldsymbol{\Lambda}}

\global\long\def\lkjparam{\boldsymbol{\Psi}}

\global\long\def\lkjdiagmat{\mathbf{M}}

\global\long\def\lkjcorrmat{\mathbf{R}}

\global\long\def\lkjlocmat{\mathbf{V}}

\global\long\def\psens{\mathbf{S}_{\alpha_{0}}}

\global\long\def\qsens{\mathbf{S}_{\alpha_{0}}^{q}}

\global\long\def\psenshat{\hat{\mathbf{S}}_{\alpha_{0}}}

\global\long\def\contampriornoarg{u}

\newcommandx\contamprior[1][usedefault, addprefix=\global, 1=\theta]{\contampriornoarg\left(#1\right)}

\global\long\def\psensfunc{S_{\contampriornoarg}}

\global\long\def\qsensfunc{S_{\contampriornoarg}^{q}}

\global\long\def\qsenssup{S_{sup}^{q}}

\global\long\def\psenssup{S_{sup}}

\global\long\def\contamnorm{C_{u}}

\newcommandx\origprior[1][usedefault, addprefix=\global, 1=\theta]{p_{0}\left(#1\right)}

\newcommandx\influence[2][usedefault, addprefix=\global, 1=\phantom{}, 2=\phantom{}]{I_{#2}^{#1}\left(\theta\right)}

\global\long\def\influenceplus{\influence[+]}

\newcommandx\qinfluence[1][usedefault, addprefix=\global, 1=\phantom{}]{\influence[#1][q]}

\newcommandx\covdens[2][usedefault, addprefix=\global, 1=\theta, 2=\alpha]{\rho\left(#1,#2\right)}

\newcommandx\unusedcovdensnorm[2][usedefault, addprefix=\global, 1=\theta, 2=\alpha]{p\left(#1\vert#2\right)}

\newcommandx\linearoperator[2][usedefault, addprefix=\global, 1=\theta, 2={\,}]{L^{#2}\left(#1\right)}

\global\long\def\laptheta{\hat{\theta}_{Lap}}

\global\long\def\laphess{\mathbf{H}_{Lap}}

\global\long\def\lap{q_{Lap}}

\global\long\def\qlap{\mathcal{Q}_{Lap}}

\global\long\def\covlap{\mathrm{Cov}_{q_{Lap}}^{Lap}}

\global\long\def\qsimple{\qmfvb}

\global\long\def\lapgtheta{g_{\laptheta}}

\global\long\def\qadvi{\mathcal{Q}_{ad}}

\newcommand{\vbassum}{Assumptions~\ref{assu:exchange_order}--\ref{assu:e_q_g_smooth} }
 
\section{Introduction
\label{sec:intro}}
Most Bayesian posteriors cannot be calculated analytically, so in practice we
turn to approximations. Variational Bayes (VB) casts posterior approximation as
an optimization problem in which the objective to be minimized is the
divergence, among a sub-class of tractable distributions, from the
exact posterior. For example, one widely-used and relatively simple flavor of VB
is ``mean field variational Bayes'' (MFVB), which employs Kullback-Leibler (KL)
divergence and a factorizing exponential family approximation for the tractable
sub-class of posteriors \citep{wainwright2008graphical}. MFVB has been
increasingly popular as an alternative to Markov Chain Monte Carlo (MCMC) in
part due to its fast runtimes on large-scale data sets. Although MFVB does not
come with any general accuracy guarantees (except asymptotic ones in special
cases \citep{westling:2015:vbconsistency,wang:2017:vbconsistency}), MFVB
produces posterior mean estimates of certain parameters that are accurate enough
to be useful in a number of real-world applications~\citep{blei:2016:variational}.
Despite this ability to produce useful point estimates for large-scale data sets,
MFVB is limited as an inferential tool; in particular, MFVB
typically underestimates marginal variances
\citep{mackay:2003:information,wang:2005:inadequacy,turner:2011:two}.  Moreover, to
the best of our knowledge, techniques for assessing Bayesian robustness have not yet been
developed for MFVB.  It is these inferential issues that are the focus of the
current paper.

Unlike the optimization approach of VB,
an MCMC posterior estimate is an empirical distribution formed with
posterior draws.
MCMC draws lend themselves
naturally to the approximate calculation of posterior moments, such
as those required for covariances. In contrast, VB approximations
lend themselves naturally to sensitivity analysis, since we can analytically
differentiate the optima with respect to perturbations.
However, as has long been known in the Bayesian robustness literature, the
contrast between derivatives and moments is not so stark since, under mild
regularity conditions that allow the exchange of integration and
differentiation, there is a direct correspondence between derivatives and
covariance
\citep[][\prettyref{subsec:cov_and_sens} below]{gustafson:1996:localposterior,basu:1996:local,efron:2015:frequentist}.

Thus, in order to calculate local sensitivity to model hyperparameters, the
Bayesian robustness literature re-casts derivatives with respect to
hyperparameters as posterior covariances that can be calculated with MCMC. In
order to provide covariance estimates for MFVB, we turn this idea on its head
and use the sensitivity of MFVB posterior expectations to estimate their
covariances. These sensitivity-based covariance estimates are referred to as
``linear response'' estimates in the statistical mechanics literature
\citep{opper:2001:advancedmeanfield}, so we
refer to them here as \emph{linear response variational Bayes} (LRVB)
covariances. Additionally, we derive straightforward MFVB versions of
hyperparameter sensitivity measures from the Bayesian robustness literature.
Under the assumption that the posterior means of interest are well-estimated by
MFVB for all the perturbations of interest, we establish that LRVB provides a
good estimate of local sensitivities. In our experiments, we compare LRVB
estimates to MCMC, MFVB, and Laplace posterior approximations. We find that the
LRVB covariances, unlike the MFVB and Laplace approximations, match the MCMC
approximations closely while still being computed over an order of magnitude
more quickly than MCMC.

In \prettyref{sec:theory} we first discuss the general relationship
between Bayesian sensitivity and posterior covariance and then define local
robustness and sensitivity. Next, in \prettyref{sec:vb_theory},
we introduce VB and derive the linear
system for the MFVB local sensitivity estimates. In
\prettyref{sec:sensitivity_in_action}, we show how to use the MFVB local
sensitivity results to estimate covariances and calculate canonical
Bayesian hyperparameter sensitivity measures. Finally, in
\prettyref{sec:experiments}, we demonstrate the speed and effectiveness of our
methods with simple simulated data, an application of automatic differentiation
variational inference (ADVI), and a large-scale industry data set.
 
\section{Bayesian Covariances and Sensitivity
\label{sec:theory}}

\subsection{Local Sensitivity and Robustness
\label{subsec:local_sensitivity}}

Denote an unknown model parameter by the vector $\theta\in\mathbb{R}^{K}$,
assume a dominating measure for $\theta$ on $\mathbb{R}^{K}$ given by $\lambda$, and denote
observed data by $x$. Suppose that we have a vector-valued hyperparameter
$\alpha\in\mathcal{A}\subseteq\mathbb{R}^{D}$ that parameterizes some
aspects of our model.  For example, $\alpha$ might represent prior parameters,
in which case we would write the prior density with respect to $\lambda$ as $p\left(\theta\vert\alpha\right)$,
or it might parameterize a class of likelihoods, in which case we could write
the likelihood as $p\left(x\vert\theta,\alpha\right)$.
Without loss of generality, we will include $\alpha$ in the definition of
both the prior and likelihood.
For the moment, let $\pthetapost\left(\theta\right)$ denote the
posterior density of $\theta$ given $x$ and $\alpha$, as given by Bayes'
Theorem (this definition of $\pthetapost\left(\theta\right)$
will be a special case of the more general \prettyref{def:pthetapost} below):
\begin{align*}
\pthetapost\left(\theta\right) &
    := p\left(\theta\vert x,\alpha\right)=
    \frac{p\left(x\vert\theta,\alpha\right)p\left(\theta\vert\alpha\right)}
    {\int p\left(x\vert\theta',\alpha\right)p\left(\theta'\vert\alpha\right)
    \lambda\left(d\theta'\right)} =
    \frac{p\left(x\vert\theta,\alpha\right)p\left(\theta\vert\alpha\right)}
    {p\left(x\vert\alpha\right)}.
\end{align*}
We will assume that we are interested in a posterior expectation of
some function $\gtheta$ (e.g., a parameter mean, a posterior predictive
value, or squared loss): $\epgtheta$. In the current work, we will
quantify the uncertainty of $\gtheta$ by the posterior variance,
$\var_{\pthetapost}\left(\gtheta\right)$. Other measures of central
tendency (e.g., posterior medians) or uncertainty (e.g., posterior
quantiles) may also be good choices but are beyond the scope of the
current work.

Note the dependence of $\epgtheta$ on both the likelihood and prior, and hence
on $\alpha$, through Bayes' Theorem. The choice of a prior and choice of a
likelihood are made by the modeler and are almost invariably a simplified
representation of the real world. The choices are therefore to some extent
subjective, and so one hopes that the salient aspects of the posterior would not
vary under reasonable variation in either choice. Consider the prior, for
example. The process of prior elicitation may be prohibitively time-consuming;
two practitioners may have irreconcilable subjective prior beliefs, or the model
may be so complex and high-dimensional that humans cannot reasonably express
their prior beliefs as formal distributions. All of these circumstances might
give rise to a range of reasonable prior choices. A posterior quantity is
``robust'' to the prior to the extent that it does not change much when
calculated under these different prior choices.

Quantifying the sensitivity of the posterior to variation in the likelihood and
prior is one of the central concerns of the field of robust Bayes
\citep{berger:2012:robust}. (We will not discuss the other central concern,
which is the selection of priors and likelihoods that lead to robust
estimators.) Suppose that we have determined that the hyperparameter $\alpha$
belongs to some open set $\mathcal{A}$, perhaps after expert prior elicitation.
Ideally, we would calculate the extrema of $\epgtheta$ as $\alpha$ ranges over
all of $\mathcal{A}$. These extrema are a measure of \textit{global robustness},
and their calculation is intractable or difficult except in special cases
\citep[Chapter 15]{moreno:2012:globalrobustness,huber:2011:robust}. A more
practical alternative is to examine how much $\epgtheta$ changes locally in
response to small perturbations in the value of $\alpha$ near some tentative
guess, $\alpha_{0}\in\mathcal{A}$. To this end we define the \textit{local
sensitivity at $\alpha_{0}$} \citep{gustafson:2012:localrobustnessbook}.
\begin{lrvb_definition} \label{def:exact_sensitivity}
The local sensitivity of $\epgtheta$
to hyperparameter $\alpha$ at $\alpha_{0}$ is given by
\begin{eqnarray}
\psens & := &
    \left.\frac{d\epgtheta}{d\alpha}\right|_{\alpha_{0}}.
    \label{eq:local_robustness}
\end{eqnarray}
\end{lrvb_definition}
$\psens$, the local sensitivity, can be considered a measure of\emph{
local robustness} \citep{gustafson:2012:localrobustnessbook}. Throughout
the paper we will distinguish between sensitivity, which comprises
objectively defined quantities such as $\psens$, and robustness,
which we treat as a more subjective concept that may be informed
by the sensitivity as well as other considerations. For example, even
if one knows $\psens$ precisely, how much posterior change is too
much change and how much prior variation is reasonable remain decisions
to be made by the modeler. For a more in-depth discussion of how we
use the terms sensitivity and robustness, see
\prettyref{app:sens_and_robustness}.

The quantity $\psens$ can be interpreted as measuring sensitivity
to hyperparameters within a small region near $\alpha=\alpha_{0}$
where the posterior dependence on $\alpha$ is approximately linear.
Then local sensitivity provides an approximation to global sensitivity
in the sense that, to first order,
\begin{align*}
\epgtheta & \approx\mbe_{\pthetapost[\alpha_0]}\left[\gtheta\right] +
    \psens^{\trans}\left(\alpha-\alpha_{0}\right).
\end{align*}
Generally, the dependence of $\epgtheta$ on $\alpha$ is not given in any
closed form that is easy to differentiate.  However, as we will now see,
the derivative $\psens$ is equal, under mild regularity conditions, to
a particular posterior covariance that can easily be estimated with MCMC draws.

\subsection{Covariances and Sensitivity\label{subsec:cov_and_sens}}

We will first state a general result relating sensitivity and covariance and
then apply it to our specific cases of interest as they arise throughout the
paper, beginning with the calculation of $\psens$ from
\prettyref{subsec:local_sensitivity}. Consider a general base density
$\pzeropost\left(\theta\right)$ defined relative to $\lambda$ and define
$\covdens$ to be a $\lambda$-measurable log perturbation function that depends
on $\alpha\in\mathcal{A}\subseteq\mathbb{R}^{D}$.
We will require the following mild technical assumption:
\begin{assumption}
\label{assu:exchange_order}
For all $\alpha\in\mathcal{A}$, $\covdens$ is continuously differentiable
with respect to $\alpha$, and, for a given $\lambda$-measurable
$g\left(\theta\right)$ there exist $\lambda$-integrable functions
$f_{0}\left(\theta\right)$ and $f_{1}\left(\theta\right)$ such that
$\left|\pzeropost\left(\theta\right)\exp\left(\covdens\right)
g\left(\theta\right)\right|<f_{0}\left(\theta\right)$
and
$\left|\pzeropost\left(\theta\right)\exp\left(\covdens\right)\right|<f_{1}\left(\theta\right)$.
\end{assumption}
Under \prettyref{assu:exchange_order} we can normalize the log-perturbed
quantity $\pzeropost\left(\theta\right) \exp\left(\covdens\right)$ to get
a density in $\theta$ with respect to $\lambda$.
\begin{lrvb_definition}\label{def:pthetapost}
Denote by $\pthetapost\left(\theta\right)$ the normalized posterior
given $\alpha$:
\begin{align}
\pthetapost\left(\theta\right) := &
        \frac{\pzeropost\left(\theta\right)\exp\left(\covdens\right)}
            {\int \pzeropost\left(\theta'\right)\exp\left(\covdens[\theta']\right)
            \lambda\left(d\theta'\right)}.\label{eq:tilting_definition}
\end{align}
\end{lrvb_definition}
For example,
$\pthetapost\left(\theta\right)$ defined in
\prettyref{subsec:local_sensitivity} is equivalent to taking
$\pzeropost\left(\theta\right) = p\left(\theta \vert x, \alpha_0\right)$
and $\covdens[][\alpha] =\log p\left(x\vert\theta,\alpha\right) +
    \log p\left(\theta\vert\alpha\right)-\log p\left(x\vert\theta,\alpha_{0}\right) -
    \log p\left(\theta\vert\alpha_{0}\right)$.

For a $\lambda$-measurable function $\gtheta$, consider differentiating
the expectation $\mbe_{\pthetapost}\left[g\left(\theta\right)\right]$
with respect to $\alpha$:
\begin{align}
    \label{eq:dg_dalpha}
\frac{d\mbe_{\pthetapost}\left[
    g\left(\theta\right)\right]}{d\alpha^{\trans}} &
    :=\frac{d}{d\alpha}\int \pthetapost\left(\theta\right)
    \gtheta\lambda\left(d\theta\right).
\end{align}
When evaluated at some $\alpha_{0} \in \mathcal{A}$, this derivative measures
the local sensitivity of $\mbe_{\pthetapost}\left[g\left(\theta\right)\right]$
to the index $\alpha$ at $\alpha_{0}$. Define
$\alphazeroball\subseteq\mathcal{A}$ to be an open ball containing $\alpha_{0}$.
Under \prettyref{assu:exchange_order} we assume without loss of generality that
$\covdens[][\alpha_0] \equiv 0$ so that $\pzeropost\left(\theta\right) =
\pthetapost[\alpha_{0}]\left(\theta\right)$; if $\covdens[][\alpha_0]$ is
non-zero, we can simply incorporate it into the definition of
$\pzeropost\left(\theta\right)$. Then, under \prettyref{assu:exchange_order},
the derivative in \prettyref{eq:dg_dalpha} is equivalent to a particular
posterior covariance.
\begin{theorem}
\label{thm:sens_cov}Under \prettyref{assu:exchange_order} ,
\begin{align}
\left.\frac{d\mbe_{\pthetapost}
    \left[g\left(\theta\right)\right]}
    {d\alpha^{\trans}}\right|_{\alpha_{0}} &
    =\cov_{\pzeropost}\left(g\left(\theta\right),
    \left.\frac{\partial\covdens}{\partial\alpha}\right|_{\alpha_{0}}\right).
    \label{eq:covariance_sensitivity_general}
\end{align}
\end{theorem}
\prettyref{thm:sens_cov} is a straightforward consequence of the
Lebesgue dominated convergence theorem; see \prettyref{app:sens_and_cov} for a
detailed proof. Versions of \prettyref{thm:sens_cov} have appeared many times
before; e.g.,
\citet{diaconis:1986:consistency,basu:1996:local,gustafson:1996:localposterior,perez:2006:mcmc}
have contributed variants of this result to the robustness literature.

By using MCMC draws from $\pzeropost$$\left(\theta\right)$ to calculate the
covariance on the right-hand side of
\prettyref{eq:covariance_sensitivity_general}, one can form an estimate of
$d\mbe_{\pthetapost}\left[g\left(\theta\right)\right]/d\alpha^{\trans}$ at
$\alpha=\alpha_{0}$. One might also approach the problem of calculating
$d\mbe_{\pthetapost}\left[g\left(\theta\right)\right]/d\alpha^{\trans}$ using
importance sampling as follows \citep[Chapter 9]{owen:2013:mcmcbook}. First, an
importance sampling estimate of the dependence of
$\mbe_{\pthetapost}\left[\gtheta\right]$ on $\alpha$ can be constructed with
weights that depend on $\alpha$. Then, differentiating the weights with
respect to $\alpha$ provides a sample-based estimate of
$d\mbe_{\pthetapost}\left[g\left(\theta\right)\right]/d\alpha^{\trans}$. We show
in \prettyref{app:mcmc_importance_sampling} that this importance sampling
approach is equivalent to using MCMC samples to estimate the covariance in
\prettyref{thm:sens_cov}.

An immediate corollary of \prettyref{thm:sens_cov} allows us to calculate
$\psens$ as a  covariance.
\begin{lrvb_corollary}
\label{cor:sens_cov_prior}Suppose that \prettyref{assu:exchange_order}
holds for some $\alpha_{0}\in\mathcal{A}$, some $g\left(\theta\right)$, and for
\begin{align*}
\covdens[][\alpha] & =\log p\left(x\vert\theta,\alpha\right) +
    \log p\left(\theta\vert\alpha\right)-\log p\left(x\vert\theta,\alpha_{0}\right) -
    \log p\left(\theta\vert\alpha_{0}\right).
\end{align*}
Then \prettyref{thm:sens_cov} implies that
\begin{align}
\psens & =\pcov\left(\gtheta,\left.
    \frac{\partial\covdens}{\partial\alpha}\right|_{\alpha_{0}}\right).
    \label{eq:covariance_sensitivity}
\end{align}
\end{lrvb_corollary}
\prettyref{cor:sens_cov_prior} can be found in \citet{basu:1996:local},
in which a version of \prettyref{cor:sens_cov_prior} is stated in
the proof of their Theorem 1, as well as in \citet{perez:2006:mcmc}
and \citet{efron:2015:frequentist}. Note that the definition of $\covdens$
does not contain any normalizing constants and so can typically be
easily calculated. Given $N_s$ MCMC draws $\{\theta_n\}_{n=1}^{N_s}$ from a chain that we assume to have
reached equilibrium at the stationary distribution $\pzeropost\left(\theta\right)$,
one can calculate an estimate of $\psens$ using the sample covariance
version of \prettyref{eq:covariance_sensitivity_general}:
\begin{align}
\psenshat & := \frac{1}{N_{s}}\sum_{n=1}^{N_{s}}
    g\left(\theta_{n}\right)\left.
    \frac{\partial\covdens[\theta_{n}]}{\partial\alpha}
    \right|_{\alpha_{0}}-\left(\frac{1}{N_{s}}
    \sum_{n=1}^{N_{s}}g\left(\theta_{n}\right)\right)
    \left(\frac{1}{N_{s}}\sum_{n=1}^{N_{s}}
    \left.\frac{\partial\covdens[\theta_{n}]}{\partial\alpha}
    \right|_{\alpha_{0}}\right)
\label{eq:mcmc_sample_cov}\\
\textrm{for }\theta_{n} &
    \sim\pzeropost\left(\theta\right),\textrm{ where }n=1,...,N_{s}.\nonumber
\end{align}
 
\section{Variational Bayesian Covariances and Sensitivity
\label{sec:vb_theory}}

\subsection{Variational Bayes\label{subsec:variational_Bayes}}

We briefly review variational Bayes and state our key assumptions about its
accuracy. We wish to find an approximate distribution, in some class
$\mathcal{Q}$ of tractable distributions, selected to minimize the
Kullback-Leibler divergence (KL divergence) between $q\in\mathcal{Q}$ and the
exact log-perturbed posterior $\pthetapost$. We assume that distributions in
$\mathcal{Q}$ are parameterized by a finite-dimensional parameter $\eta$ in some
feasible set $\Omega_{\eta}\subseteq\mathbb{R}^{K_{\eta}}$.
\begin{lrvb_definition}
The approximating variational family is given by
\begin{eqnarray}
\mathcal{Q} & := & \left\{ q:q=q\left(\theta;\eta\right)\textrm{ for }
    \eta\in\Omega_{\eta}\right\} .\label{eq:q_approximating_family}
\end{eqnarray}
\end{lrvb_definition}
Given $\mathcal{Q}$, we define the optimal $q\in\mathcal{Q}$, which
we call $\qthetapost\left(\theta\right)$, as the distribution that minimizes the KL
divergence $KL\left(q\left(\theta;\eta\right)||\pthetapost\left(\theta\right)\right)$
from $\pthetapost\left(\theta\right)$. We denote the corresponding optimal variational
parameters as $\etaopt$.
\begin{lrvb_definition}
\label{def:kl_divergence}The variational approximation $\qthetapost\left(\theta\right)$
to $\pthetapost\left(\theta\right)$ is defined by

\begin{eqnarray}
\qthetapost\left(\theta\right) := q\left(\theta;\etaopt\right) & := &
    \textrm{argmin}_{q\in\mathcal{Q}}\left\{
        KL\left(q\left(\theta;\eta\right)||\pthetapost\left(\theta\right)\right)\right\},
\label{eq:kl_divergence}
\end{eqnarray}
where
\begin{align*}
KL\left(q\left(\theta;\eta\right)||\pthetapost\left(\theta\right)\right) &
    = \mbe_{q\left(\theta;\eta\right)}
        \left[\log q\left(\theta;\eta\right) - \log \pthetapost\left(\theta\right)
        \right].
\end{align*}
\end{lrvb_definition}
In the KL divergence, the (generally intractable) normalizing constant for
$\pthetapost\left(\theta\right)$ does not depend on $q\left(\theta\right)$
and so can be neglected when optimizing. In order for
the KL divergence to be well defined, we assume that both
$\pzeropost\left(\theta\right)$
and $q\left(\theta\right)$ are given with respect to the same base
measure, $\lambda$, and that the support of $q\left(\theta\right)$
is contained in the support of $\pthetapost\left(\theta\right)$.
We will require some additional mild regularity conditions in
\prettyref{subsec:vb_sensitivity} below.

A common choice for the approximating family
$\mathcal{Q}$ in \prettyref{eq:q_approximating_family}
is the ``mean field family'' \citep{wainwright2008graphical,blei:2016:variational},
\begin{align}
\qmfvb & :=\left\{ q\left(\theta\right):q\left(\theta\right) =
    \prod_{k}q\left(\theta_{k};\eta_{k}\right)\right\} ,
    \label{eq:q_mean_field_family}
\end{align}
where $k$ indexes a partition of the full vector $\theta$ and of the parameter
vector $\eta$. That is, $\qmfvb$ approximates the posterior
$\pthetapost\left(\theta\right)$ as a distribution that factorizes across
sub-components of $\theta$. This approximation is commonly referred to as
``MFVB,'' for ``mean field variational Bayes.'' Note that, in general, each
function $q\left(\theta_{k};\eta_{k}\right)$ in the product is different. For
notational convenience we write $q\left(\theta_{k};\eta_{k}\right)$ instead of
$q_{k}\left(\theta_{k};\eta_{k}\right)$ when the arguments make it clear which
function we are referring to, much as the same symbol $p$ is used to refer to
many different probability distributions without additional indexing.

One may additionally assume that the components
$q\left(\theta_{k};\eta_{k}\right)$ are in a convenient exponential family.
Although the exponential family assumption does not in general follow from a
factorizing assumption, for compactness we will refer to both the factorization
and the exponential family assumption as MFVB.

In an MFVB approximation, $\Omega_{\eta}$ could be a stacked vector
of the natural parameters of the exponential families, or the moment
parameterization, or perhaps a transformation of these parameters
into an unconstrained space (e.g., the entries of the log-Cholesky decomposition
of a positive definite information matrix). For more concrete examples,
see \prettyref{sec:experiments}. Although all of our experiments
and much of our motivating intuition will use MFVB, our results extend
to other choices of $\mathcal{Q}$ that satisfy the necessary assumptions.

\subsection{Variational Bayes sensitivity\label{subsec:vb_sensitivity} }

Just as MCMC approximations lend themselves to moment calculations,
the variational form of VB approximations lends itself to sensitivity
calculations. In this section we derive the sensitivity of VB posterior
means to generic perturbations---a VB analogue of \prettyref{thm:sens_cov}.
In \prettyref{sec:sensitivity_in_action} we will choose particular
perturbations to calculate VB prior sensitivity and, through \prettyref{thm:sens_cov},
posterior covariances.

In \prettyref{def:kl_divergence}, the variational approximation is
a function of $\alpha$ through the optimal parameters $\etatopt$,
i.e., $\qthetapost\left(\theta\right)=q\left(\theta,\etatopt\right)$.
In turn, the posterior expectation $\mbe_{\qthetapost}\left[\gtheta\right]$
is also a function of $\alpha$, and its derivative at
$\alpha_{0}$---the local sensitivity of the variational approximation
to $\alpha$---has a closed form under the following mild technical
conditions. As with $\pzeropost$, define $\qzeropost:=\qthetapost[\alpha_{0}]$,
and define $\etaoptzero:=\etaopt\left(\alpha_{0}\right)$.

All the following assumptions are intended to hold for a given $\pthetapost\left(\theta\right)$,
approximating class $\mathcal{Q}$, $\lambda$-measurable function
$\gtheta$, and to hold for all $\alpha\in\alphazeroball$ and all
$\eta$ in an open neighborhood of $\etaoptzero$.
\begin{assumption}
\label{assu:kl_differentiable}The KL divergence at $KL\left(q\left(\theta;\eta\right)||\pzeropost\left(\theta\right)\right)$
and expected log perturbation $\mbe_{q\left(\theta;\eta\right)}\left[\covdens\right]$
are twice continuously differentiable in $\eta$ and $\alpha$.
\end{assumption}
\begin{assumption}
\label{assu:opt_interior}There exists a strict local minimum, $\etaopt\left(\alpha\right)$,
of $KL\left(q\left(\theta;\eta\right)||\pthetapost\left(\theta\right)\right)$
in \prettyref{eq:kl_divergence} such that $\etatopt$ is interior
to $\Omega_{\eta}$.
\end{assumption}
\begin{assumption}
\label{assu:e_q_g_smooth}The expectation $\mbe_{q\left(\theta;\eta\right)}\left[\gtheta\right]$
is a continuously differentiable function of $\eta$.
\end{assumption}
We define the following quantities for notational convenience.
\begin{lrvb_definition}
\label{def:vb_derivatives}Define the following derivatives of variational
expectations evaluated at the optimal parameters:

\begin{tabular}{ccc}
$\klhess :=
    \left.\frac{\partial^{2}KL\left(q\left(\theta;\eta\right)||
        \pzeropost\left(\theta\right)\right)}{\partial\eta\partial\eta^{\trans}}
            \right|_{\eta=\etaoptzero}$ &
$\efhess:=
    \left.\frac{\partial^2\mbe_{q\left(\theta;\eta\right)}\left[\covdens\right]}
    {\partial\alpha\partial\eta^{\trans}}
    \right|_{\eta=\etaoptzero,\alpha=\alpha_{0}}$ &
$\eggrad:=
    \left.\frac{\partial\mbe_{q\left(\theta;\eta\right)}
    \left[g\left(\theta\right)\right]}{\partial\eta^{\trans}}
    \right|_{\eta=\etaoptzero}.$\tabularnewline
\end{tabular}
\end{lrvb_definition}
Since $\gtheta$, $\alpha$, and $\eta$ are all vectors, the quantities
$\klhess$, $\efhess$, and $\eggrad$ are matrices. We are now ready
to state a VB analogue of \prettyref{thm:sens_cov}.
\begin{theorem}
\label{thm:lrvb_formula}Consider a variational approximation
$\qthetapost\left(\theta\right)$
to $\pthetapost\left(\theta\right)$ as given in \prettyref{def:kl_divergence}
and a $\lambda$-measurable function $\gtheta$. Then, under \vbassum,
using the definitions given in \prettyref{def:vb_derivatives}, we
have
\begin{eqnarray}
\left.\frac{d\mbe_{\qthetapost}
    \left[\gtheta\right]}{d\alpha^{\trans}}\right|_{\alpha_{0}} & = &
    \eggrad\klhess^{-1}\efhess^{\trans}.\label{eq:lrvb_formula}
\end{eqnarray}
\end{theorem}
A proof of \prettyref{thm:lrvb_formula} is given in \prettyref{app:lrvb}.
As with \prettyref{thm:sens_cov}, by choosing the appropriate $\covdens$
and evaluating $\efhess$, we can use \prettyref{thm:lrvb_formula}
to calculate the exact sensitivity of VB solutions to any arbitrary
local perturbations that satisfy the regularity conditions. \vbassum
are typically not hard to verify. For an example, see \prettyref{app:mvn_exact},
where we establish \vbassum for a multivariate normal target
distribution and a mean-field approximation.

\prettyref{eq:lrvb_formula} is formally similar to frequentist sensitivity
estimates. For example, the pioneering paper of \citet{cook:1986:assessment}
contains a formula for assessing the curvature of a marginal likelihood surface
\citep[Equation 15]{cook:1986:assessment} that, like our
\prettyref{thm:lrvb_formula}, represents the sensitivity as a linear system
involving the Hessian of an objective function at its optimum. The geometric
interpretation of local robustness suggested by \citet{cook:1986:assessment} has
been extended to Bayesian settings (see, for example,
\citet{zhu:2007:perturbation,zhu:2011:bayesian}). In addition to generality, one
attractive aspect of their geometric approach is its invariance to
parameterization. Investigating geometric interpretations of the present work
may be an interesting avenue for future research.

\subsection{Approximating with Variational Bayes}

Recall that we are ultimately interested in $\epgtheta$. Variational
approximations and their sensitivity measures will be useful to the
extent that both the variational means and sensitivities are close
to the exact means and sensitivities. We formalize these desiderata as follows.
\begin{condition}
\label{cond:vb_accurate}
Under \vbassum and the quantities defined therein, we additionally have,
for all $\alpha \in \mathcal{A}$,
\begin{align}
\mbe_{\qthetapost}\left[\gtheta\right] &
    \approx\mbe_{\pthetapost}\left[\gtheta\right] \quad \textrm{and}
    \label{eq:vb_accurate}\\
    \left.\frac{d\mbe_{\qthetapost[\alpha]}\left[\gtheta\right]}
        {d\alpha^{\trans}}\right|_{\alpha_{0}} &
        \approx\left.\frac{d\mbe_{\pthetapost}\left[\gtheta\right]}
            {d\alpha^{\trans}}
        \right|_{\alpha_{0}}
    \label{eq:vb_slope_accurate}
\end{align}
\end{condition}
We will not attempt to be precise about what we mean by the ``approximately
equal'' sign, since we are not aware of any practical tools for evaluating
quantitatively whether \prettyref{cond:vb_accurate} holds other than
running both VB and MCMC (or some other slow but accurate posterior
approximation) and comparing the results. However, VB has been useful
in practice to the extent that \prettyref{cond:vb_accurate} holds
true for at least some parameters of interest. We provide some intuition
for when \prettyref{cond:vb_accurate} might hold in
\prettyref{subsec:laplace_experiments},
and will evaluate \prettyref{cond:vb_accurate} in each of our experiments
below by comparing the VB and MCMC posterior approximate means and
sensitivities.

Since \prettyref{cond:vb_accurate} holds only for a particular choice
of $\gtheta$, it is weaker than the assumption that $\qthetapost$
is close to $\pthetapost$ in KL divergence, or even that all the
posterior means are accurately estimated. For example, as discussed
in Appendix B of \citet{giordano:2015:lrvb} and in Section 10.1.2
of \citet{bishop:2006:pattern}, a mean-field approximation to a multivariate
normal posterior produces inaccurate covariances and may have an arbitrarily
bad KL divergence from $\pthetapost$, but \prettyref{cond:vb_accurate}
holds exactly for the location parameters. We discuss the multivariate
normal example further in \prettyref{subsec:lrvb_cov} and
\prettyref{subsec:laplace_experiments} below.
 
\section{Calculation and Uses of Sensitivity
\label{sec:sensitivity_in_action}}
In this section, we discuss two applications of \prettyref{thm:sens_cov}
and \prettyref{thm:lrvb_formula}: calculating improved
covariance estimates and prior sensitivity measures for MFVB. Throughout
this section, we will assume that we can apply \prettyref{thm:sens_cov}
and \prettyref{thm:lrvb_formula} unless stated otherwise.

\subsection{Covariances for Variational Bayes\label{subsec:lrvb_cov}}

Consider the mean field approximating family, $\mathcal{Q}_{mf}$,
from \prettyref{subsec:variational_Bayes} and a fixed exact posterior
$\pzeropost\left(\theta\right)$. It is well known that the resulting
marginal variances also tend to be under-estimated even when
parameters means are well-estimated
(see, e.g., \citep[Chapter 10]{mackay:2003:information,wang:2005:inadequacy,turner:2011:two,bishop:2006:pattern}).
Even more obviously, any $q\in\mathcal{Q}_{mf}$ yields zero as its estimate of
the covariance between sub-components of $\theta$ that are in different
factors of the mean field approximating family. It is therefore unreasonable
to expect that $\qcov\left(\gtheta\right)\approx\pcov\left(\gtheta\right)$.
However, if \prettyref{cond:vb_accurate} holds, we may expect the
sensitivity of MFVB means to certain perturbations to be accurate by
\prettyref{cond:vb_accurate}, and, by \prettyref{thm:sens_cov},
we expect the corresponding covariances to be accurately estimated
by the MFVB sensitivity. In particular, by taking
$\covdens=\alpha^{\trans}g\left(\theta\right)$
and $\alpha_{0}=0$, we
have by \prettyref{cond:vb_accurate} that
\begin{align}
\left.\frac{d\mbe_{\qthetapost}\left[\gtheta\right]}{d\alpha^{\trans}}\right|_{\alpha=0} &
    \approx\left.\frac{d\mbe_{\pthetapost}\left[\gtheta\right]}
    {d\alpha^{\trans}}\right|_{\alpha=0}=\pcov\left(g\left(\theta\right)\right).
    \label{eq:sensitivity_is_cov}
\end{align}
We can consequently use \prettyref{thm:lrvb_formula} to provide an
estimate of $\pcov\left(g\left(\theta\right)\right)$ that may be
superior to $\qcov\left(\gtheta\right)$. With this motivation in
mind, we make the following definition.
\begin{lrvb_definition}
\label{def:lrvb_covariance}The \textit{linear response variational
Bayes (LRVB) approximation}, $\lrvbcov\left(\gtheta\right)$, is given
by
\begin{align}
\lrvbcov\left(\gtheta\right) & :=\eggrad\klhess^{-1}\eggrad^{\trans}.
\label{eq:lrvb_for_covariance}
\end{align}
\end{lrvb_definition}
\begin{lrvb_corollary}
\label{cor:lrvb_accurate}For a given $\pzeropost\left(\theta\right)$,
class $\mathcal{Q}$, and function $\gtheta$, when \vbassum and
\prettyref{cond:vb_accurate} hold for $\covdens=\alpha^{\trans}g\left(\theta\right)$
and $\alpha_{0}=0$, then
\begin{align*}
\lrvbcov\left(\gtheta\right) & \approx\pcov\left(\gtheta\right).
\end{align*}
\end{lrvb_corollary}
The strict optimality of $\etaoptzero$ in \prettyref{assu:opt_interior}
guarantees that $\klhess$ will be positive definite and symmetric,
and, as desired, the covariance estimate $\lrvbcov\left(\gtheta\right)$
will be positive semidefinite and symmetric. Since the optimal value
of every component of $\mbe_{\qthetapost}\left[\gtheta\right]$ may
be affected by the log perturbation $\alpha^{\trans}\gtheta$, $\lrvbcov\left(\gtheta\right)$
can estimate non-zero covariances between elements of $\gtheta$ even
when they have been partitioned into separate factors of the mean
field approximation.

Note that $\lrvbcov\left(\gtheta\right)$ and $\qcov\left(\gtheta\right)$
differ only when there are at least some moments of $\pthetapost[0]$
that $\qthetapost[0]$ fails to accurately estimate. In particular,
if $\qthetapost$ provided a good approximation to $\pthetapost$
for both the first and second moments of $\gtheta$, then we would
have $\lrvbcov\left(\gtheta\right)\approx\qcov\left(g\left(\theta\right)\right)$
since, for $\qzeropost$ and $\pzeropost$,
\begin{align*}
\mbe_{\qzeropost}\left[\gtheta\right] &\approx \mbe_{\pzeropost}\left[\gtheta\right]\textrm{ and }\\
\mbe_{\qzeropost}\left[\gtheta\gtheta^{\trans}\right] &\approx \mbe_{\pzeropost}\left[\gtheta\gtheta^{\trans}\right]\Rightarrow\\
\qcov\left(g\left(\theta\right)\right) &\approx \pcov\left(\gtheta\right),
\end{align*}
and, for $\qthetapost$ and $\pthetapost$,
\begin{align*}
\mbeq\left[\gtheta\right]& \approx\epgtheta \Rightarrow\\
\lrvbcov\left(\gtheta\right) & \approx\pcov\left(\gtheta\right).
\end{align*}
Putting these two approximate equalities together, we see that, when the first and second moments
of $\qthetapost$ approximately match those of $\pthetapost$,
\begin{align*}
\qcov\left(g\left(\theta\right)\right) & \approx\lrvbcov\left(\gtheta\right).
\end{align*}
However, in general, $\lrvbcov\left(\gtheta\right)\ne\qcov\left(g\left(\theta\right)\right)$.
In this sense, any discrepancy between $\lrvbcov\left(\gtheta\right)$
and $\qcov\left(g\left(\theta\right)\right)$ indicates an inadequacy
of the variational approximation for at least the second moments of
$\gtheta$.

Let us consider a simple concrete illustrative example which will
demonstrate both how $\qcov\left(\gtheta\right)$ can be a poor approximation
to $\pcov\left(\gtheta\right)$ and how $\lrvbcov\left(\gtheta\right)$
can improve the approximation for some moments but not others. Suppose
that the exact posterior is a bivariate normal,
\begin{align}
\pzeropost\left(\theta\right) & =\mathcal{N}\left(\theta \vert \mu,\covmat\right),\label{eq:normal_example}
\end{align}
where $\theta=\left(\theta_{1},\theta_{2}\right)^{\trans}$, $\mu=\left(\mu_{1},\mu_{2}\right)^{\trans}$,
$\covmat$ is invertible, and $\infomat:=\covmat^{-1}$. One may think
of $\mu$ and $\covmat$ as known functions of $x$
via Bayes' theorem, for example, as given by a normal-normal conjugate
model. Suppose we use the MFVB approximating family
\begin{align*}
\mathcal{Q}_{mf} & =\left\{ q\left(\theta\right):q\left(\theta\right)=q\left(\theta_{1}\right)q\left(\theta_{2}\right)\right\} .
\end{align*}
One can show (see \prettyref{app:mvn_exact}) that the optimal
MFVB approximation to $\pthetapost$ in the family $\mathcal{Q}_{mf}$
is given by
\begin{align*}
\qzeropost\left(\theta_{1}\right) & =\mathcal{N}\left(\theta_{1} \vert \mu_{1},\infomat_{11}^{-1}\right)\\
\qzeropost\left(\theta_{2}\right) & =\mathcal{N}\left(\theta_{2} \vert \mu_{2},\infomat_{22}^{-1}\right).
\end{align*}
Note that the posterior mean of $\theta_{1}$ is exactly estimated
by the MFVB procedure:
\begin{align*}
\mbe_{\qzeropost}\left[\theta_{1}\right] & =\mu_{1}=\mbe_{\pzeropost}\left[\theta_{1}\right].
\end{align*}
However, if $\covmat_{12}\ne0$, then $\infomat_{11}^{-1}<\covmat_{11}$,
and the variance of $\theta_{1}$ is underestimated. It follows that
the expectation of $\theta_{1}^{2}$ is \textit{not} correctly estimated
by the MFVB procedure:
\begin{align*}
\mbeq\left[\theta_{1}^{2}\right] & =\mu_{1}^{2}+\infomat_{11}^{-1}<\mu_{1}^{2}+\covmat_{11}=\mbep\left[\theta_{1}^{2}\right].
\end{align*}
An analogous statement holds for $\theta_{2}$. Of course, the covariance
is also mis-estimated if $\covmat_{12}\ne0$ since, by construction
of the MFVB approximation,
\begin{align*}
\qcov\left(\theta_{1},\theta_{2}\right) & =0\ne\covmat_{12}=\pcov\left(\theta_{1},\theta_{2}\right).
\end{align*}

Now let us take the log perturbation
$\covdens=\theta_{1}\alpha_{1}+\theta_{2}\alpha_{2}$. For all $\alpha$ in a
neighborhood of zero, the log-perturbed posterior given by
\prettyref{eq:tilting_definition} remains multivariate normal, so it remains the
case that, as a function of $\alpha$,
$\mbe_{\qthetapost}\left[\theta_{1}\right]=\mbe_{\pthetapost}\left[\theta_{1}\right]$
and
$\mbe_{\qthetapost}\left[\theta_{2}\right]=\mbe_{\pthetapost}\left[\theta_{2}\right]$.
Again, see \prettyref{app:mvn_exact} for a detailed proof. Consequently,
\prettyref{cond:vb_accurate} holds with equality (not approximate equality) when
$\gtheta=\theta$. However, since the second moments are not accurate
(irrespective of $\alpha$), \prettyref{cond:vb_accurate} does not hold exactly when
$\gtheta=\left(\theta_{1}^{2},\theta_{2}^{2}\right)^{\trans}$, nor when
$\gtheta=\theta_{1}\theta_{2}$. (\prettyref{cond:vb_accurate} may still hold
approximately for second moments when $\covmat_{12}$ is small.) The fact that
\prettyref{cond:vb_accurate} holds with equality for $\gtheta=\theta$ allows us
to use \prettyref{thm:sens_cov} and \prettyref{thm:lrvb_formula} to calculate
$\lrvbcov\left(\gtheta\right)=\pcov\left(g\left(\theta\right)\right)$, even
though $\mbe_{\pzeropost}\left[\theta_{1}\theta_{2}\right]$ and
$\mbe_{\pzeropost}\left[\left(\theta_{1}^{2},\theta_{2}^{2}\right)^{\trans}\right]$
are mis-estimated.

In fact, when \prettyref{cond:vb_accurate} holds with equality for some
$\theta_{i}$, then the estimated covariance in
\prettyref{eq:lrvb_for_covariance} for all terms involving $\theta_{i}$ will be
exact as well. \prettyref{cond:vb_accurate} holds with equality for the means of
$\theta_{i}$ in the bivariate normal model above, and in fact holds for the
general multivariate normal case, as described in \prettyref{app:mvn_exact}.
Below, in \prettyref{sec:experiments}, in addition to robustness measures, we
will also report the accuracy of \prettyref{eq:lrvb_for_covariance} for
estimating posterior covariances. We find that, for most parameters of interest,
particularly location parameters, $\lrvbcov\left(g\left(\theta\right)\right)$
provides a good approximation to $\pcov\left(g\left(\theta\right)\right)$.

\subsection{Linear Response Covariances in Previous Literature}

The application of sensitivity measures to VB problems for the purpose
of improving covariance estimates has a long history under the name
``linear response methods.'' These methods originated in the statistical
physics literature
\citep{tanaka:2000:information,opper:2001:advancedmeanfield}
and have been applied to various statistical and machine learning
problems
\citep{kappen:1998:efficient,tanaka:1998:mean,welling:2004:linear,opper:2003:variational}.
The current paper, which builds on this line of work and on our earlier work
\citep{giordano:2015:lrvb}, represents a simplification and generalization
of classical linear response methods and serves to elucidate the
relationship between these methods and the local robustness literature.
In particular, while \citet{giordano:2015:lrvb} focused on moment-parameterized
exponential families, we derive linear-response covariances for generic
variational approximations and connect the linear-response methodology
to the Bayesian robustness literature.

A very reasonable approach to address the inadequacy of MFVB covariances is
simply to increase the expressiveness of the model class $\mathcal{Q}$---although, as noted by \citet{turner:2011:two}, increased expressiveness
does not necessarily lead to better posterior moment estimates. This approach is
taken by much of the recent VB literature
\citep[e.g.,][]{tran:2015:copula,tran:2015:gp,ranganath:2015:hierarchical,rezende:2015:flows,liu:2016:stein}.
Though this research direction remains lively and promising, the
use of a more complex class $\mathcal{Q}$ sometimes sacrifices the
speed and simplicity that made VB attractive in the first place, and
often without the relatively well-understood convergence guarantees
of MCMC. We also stress that the current work is not necessarily
at odds with the approach of increasing expressiveness. Sensitivity
methods can be a supplement to any VB approximation for which our
estimators, which require solving a linear system involving the Hessian
of the KL divergence, are tractable.

\subsection{The Laplace Approximation and Linear Response Covariances
\label{subsec:laplace_approx}}

In this section, we briefly compare linear response covariances to
the Laplace approximation \citep[Chapter 13]{gelman:2014:bayesian}.
The Laplace approximation to $\pzeropost\left(\theta\right)$ is formed
by first finding the ``maximum \textit{a posteriori}'' (MAP) estimate,
\begin{align}
\laptheta & :=\argmax_{\theta}\pzeropost\left(\theta\right),\label{eq:laplace_opt}
\end{align}
and then forming the multivariate normal posterior approximation
\begin{align}
\laphess & :=-\left.\frac{\partial^{2}\pzeropost\left(\theta\right)}
    {\partial\theta\partial\theta^{\trans}}\right|_{\laptheta}
    \label{eq:laplace_hess}\\
\covlap\left(\theta\right) & :=\laphess^{-1}\nonumber \\
\lap\left(\theta\right) &
    :=\mathcal{N}\left(\theta \vert \laptheta,\covlap\left(\theta\right)\right).
    \label{eq:laplace_approx}
\end{align}
Since both LRVB and the Laplace approximation require the solution
of an optimization problem (\prettyref{eq:kl_divergence} and
\prettyref{eq:laplace_opt}
respectively) and the estimation of covariances via an inverse Hessian
of the optimization objective (\prettyref{eq:lrvb_for_covariance}
and \prettyref{eq:laplace_hess} respectively), it will be instructive
to compare the two approaches.

Following \citet{neal:1998:variationalEM}, we can, in fact, view the
MAP estimator as a special variational approximation, where
we define
\begin{align*}
\qlap:= & \Big\{ q\left(\theta;\theta_{0}\right):
    \int q\left(\theta;\theta_{0}\right)\log\pzeropost\left(\theta\right)
    \lambda\left(d\theta\right)=\log\pzeropost\left(\theta_{0}\right)\textrm{ and }\\
 & \quad\quad\int q\left(\theta;\theta_{0}\right)\log q\left(\theta;\theta_{0}\right)
 \lambda\left(d\theta\right)=\constant\Big\},
\end{align*}
where the $\constant$ term is constant in $\theta_{0}$. That is,
$\qlap$ consists of ``point masses'' at $\theta_{0}$ with constant
entropy. Generally such point masses may not be defined as densities
with respect to $\lambda$, and the $KL$ divergence in \prettyref{eq:kl_divergence}
may not be formally defined for $q\in\qlap$. However, if $\qlap$
can be approximated arbitrarily well by well-defined densities (e.g.,
normal distributions with variance fixed at an arbitrarily small number),
then we can use $\qlap$ as a heuristic tool for understanding the
MAP estimator.

Since $\qlap$ contains only point masses, the covariance of the variational
approximation is the zero matrix: $\cov_{\lap}\left(\theta\right)=0$.
Thus, as when one uses the mean field assumption, $\cov_{\lap}\left(\theta\right)$
underestimates the marginal variances and magnitudes of the covariances
of $\pcov\left(\theta\right)$. Of course, the standard Laplace approximation
uses $\covlap\left(\theta\right)$, not $\cov_{\lap}\left(\theta\right)$,
to approximate $\pcov\left(\theta\right)$. In fact,
$\covlap\left(\theta\right)$
is equivalent to a linear response covariance matrix calculated for
the approximating family $\qlap$:
\begin{align*}
KL\left(q\left(\theta;\theta_{0}\right)||\pzeropost\left(\theta\right)\right) &
    =-\log\pzeropost\left(\theta_{0}\right)-\constant\Rightarrow\\
\laptheta & =\argmax_{\theta}\pzeropost\left(\theta\right) =
    \argmin_{\theta_{0}}KL\left(q\left(\theta;\theta_{0}\right)||\pzeropost
    \left(\theta\right)\right)=\theta_{0}^{*}\\
\laphess & =-\left.\frac{\partial^{2}\pzeropost\left(\theta\right)}
    {\partial\theta\partial\theta^{\trans}}\right|_{\laptheta} =
    -\left.\frac{\partial^{2}
        KL\left(q\left(\theta;\theta_{0}\right)||
            \pzeropost\left(\theta\right)\right)}{\partial\theta_{0}
            \partial\theta_{0}^{\trans}}\right|_{\theta_{0}^{*}}=\klhess.
\end{align*}
So $\laptheta=\theta_{0}^{*}$, $\laphess=\klhess$,
and $\covlap\left(\theta\right)=\lrvbcov\left(\theta\right)$
for the approximating family $\qlap$.

From this perspective, the accuracy of the Laplace approximation depends
precisely on the extent to which \prettyref{cond:vb_accurate} holds
for the family of point masses $\qlap$. Typically, VB approximations
use a $\mathcal{Q}$ that is more expressive than $\qlap$,
and we might expect \prettyref{cond:vb_accurate} to be more likely
to apply for a more expressive family. It follows that we might expect
the LRVB covariance estimate $\lrvbcov$ for general $\mathcal{Q}$
to be more accurate than the Laplace covariance approximation $\covlap$.
We demonstrate the validity of this intuition in the experiments of
\prettyref{sec:experiments}.

\subsection{Local Prior Sensitivity for MFVB \label{subsec:lrvb_robustness} }

We now turn to estimating prior sensitivities for MFVB estimates---the
variational analogues of $\psens$ in \prettyref{def:exact_sensitivity}.
First, we define the variational local sensitivity.
\begin{lrvb_definition}
The \textit{local sensitivity} of $\eqgtheta$ to prior parameter
$\alpha$ at $\alpha_{0}$ is given by
\begin{align*}
\qsens & :=\left.\frac{d\eqgtheta}{d\alpha}\right|_{\alpha_{0}}.
\end{align*}
\end{lrvb_definition}
\begin{lrvb_corollary}
\label{cor:sens_vb_prior}Suppose that \vbassum and \prettyref{cond:vb_accurate}
hold for some $\alpha_{0}\in\mathcal{A}$ and for
\begin{align*}
\covdens[][\alpha] & =
    \log p\left(x\vert\theta,\alpha\right) +
    \log p\left(\theta\vert\alpha\right)-
    \log p\left(x\vert\theta,\alpha_{0}\right)-
    \log p\left(\theta\vert\alpha\right).
\end{align*}
Then $\qsens\approx\psens$.
\end{lrvb_corollary}
\prettyref{cor:sens_vb_prior} states that, as with the covariance approximations
in \prettyref{subsec:lrvb_cov}, $\qsens$ is a useful approximation to $\psens$
to the extent that \prettyref{cond:vb_accurate} holds---that is, to
the extent that the MFVB means are good approximations to the exact means for
the prior perturbations $\alpha\in\alphazeroball$.

Under the $\covdens$ given in \prettyref{cor:sens_vb_prior},
\prettyref{thm:lrvb_formula}
gives the following formula for the variational local sensitivity:
\begin{align}
\qsens & =\eggrad\klhess^{-1}\left.\frac{\partial}{\partial\eta^{\trans}}
    \mbe_{q\left(\theta;\eta\right)}
    \left[\left.\frac{\partial\covdens[][\alpha]}{\partial\alpha}
    \right|_{\alpha_{0}}\right]\right|_{\etaoptzero}.
    \label{eq:variational_local_robustness}
\end{align}

We now use \prettyref{eq:variational_local_robustness} to reproduce
MFVB versions of some standard robustness measures found in the existing
literature. A simple case is when the prior $p\left(\theta\vert\alpha\right)$
is believed to be in a given parametric family, and we are simply
interested in the effect of varying the parametric family's parameters
\citep{basu:1996:local,giordano:2016:robust}. For illustration, we
first consider a simple example where $p\left(\theta\vert\alpha\right)$
is in the exponential family, with natural sufficient statistic $\theta$
and log normalizer $A\left(\alpha\right)$, and we take $g\left(\theta\right)=\theta$.
In this case,
\begin{eqnarray*}
\log p\left(\theta\vert\alpha\right) & = &
    \alpha^{\trans}\theta-A\left(\alpha\right)\\
\efhess & = & \left.\frac{\partial}{\partial\eta^{\trans}}
    \mbe_{q\left(\theta;\eta\right)}\left[\left.\frac{\partial}
    {\partial\alpha}\left(\alpha^{\trans}\theta -
    A\left(\alpha\right)\right)\right|_{\alpha_{0}}\right]\right|_{\etaoptzero}\\
 & = & \left.\left(\frac{\partial}{\partial\eta^{\trans}}
    \mbe_{q\left(\theta;\eta\right)}\left[\theta\right]-\left.
    \frac{\partial}{\partial\eta^{\trans}}
    \frac{\partial A\left(\alpha\right)}
        {\partial\alpha}\right|_{\alpha_{0}}\right)\right|_{\etaoptzero}\\
 & = & \left.\frac{\partial}{\partial\eta^{\trans}}
    \mbe_{q\left(\theta;\eta\right)}\left[\theta\right]\right|_{\etaoptzero}\\
 & = & \eggrad.
\end{eqnarray*}
Note that when $\efhess=\eggrad$,
\prettyref{eq:variational_local_robustness}
is equivalent to \prettyref{eq:lrvb_for_covariance}. So we see that
\begin{align*}
\qsens & =\lrvbcov\left(\theta\right).
\end{align*}
In this case, the sensitivity is simply the linear response covariance
estimate of the covariance, $\lrvbcov\left(\theta\right)$. By
the same reasoning, the exact posterior sensitivity is given by
\begin{align*}
\psens & =\pcov\left(\theta\right).
\end{align*}
Thus, $\qsens\approx\psens$ to the extent that
$\lrvbcov\left(\theta\right)\approx\pcov\left(\theta\right)$, which again holds
to the extent that \prettyref{cond:vb_accurate} holds. Note that if we had used
a mean field assumption and had tried to use the direct, uncorrected response
covariance $\qcov\left(\theta\right)$ to try to evaluate $\qsens$, we would have
erroneously concluded that the prior on one component, $\theta_{k_{1}}$, would
not affect the posterior mean of some other component, $\theta_{k_{2}}$, for
$k_{2}\ne k_{1}$.

Sometimes it is easy to evaluate the derivative of the log prior even when it is
not easy to normalize it. As an example, we will show how to calculate the local
sensitivity to the concentration parameter of an LKJ prior
\citep{lewandowski:2009:lkj} under an inverse Wishart variational approximation.
The LKJ prior is defined as follows. Let $\covmat$ (as part of $\theta$) be an unknown $K\times K$
covariance matrix. Define the $K\times K$ scale matrix
$\lkjdiagmat$ such that
\begin{align*}
\lkjdiagmat_{ij} & =\begin{cases}
\sqrt{\covmat_{ij}} & \textrm{if }i=j\\
0 & \textrm{otherwise}.
\end{cases}
\end{align*}
Using $\lkjdiagmat$, define the correlation matrix $\lkjcorrmat$ as
\begin{align*}
\lkjcorrmat & =\lkjdiagmat^{-1}\covmat\lkjdiagmat^{-1}.
\end{align*}
The LKJ prior on the covariance matrix $\lkjcorrmat$ with concentration
parameter $\alpha>0$ is given by:
\begin{align*}
p_{\mathrm{LKJ}}\left(\lkjcorrmat\vert\alpha\right) & \propto\left|\lkjcorrmat\right|^{\alpha-1}.
\end{align*}
The Stan manual recommends the use of $p_{\mathrm{LKJ}}$, together
with an independent prior on the diagonal entries of the scaling matrix
$\lkjdiagmat$, for the prior on a covariance matrix that appears in a hierarchical
model \citep[Chapter 9.13]{stan-manual:2015}.

Suppose that we have chosen the variational approximation
\begin{align*}
q\left(\covmat\right) & :=\textrm{InverseWishart}\left(\covmat\vert\lkjparam,\nu\right),
\end{align*}
where $\lkjparam$ is a positive definite scale matrix and $\nu$
is the number of degrees of freedom. In this case, the variational parameters
are $\eta=\left(\lkjparam,\nu\right)$. We write $\eta$ with the understanding that
we have stacked only the upper-diagonal elements of $\lkjparam$
since $\lkjparam$ is constrained to be symmetric and $\etaopt$ must
be interior. As we show in \prettyref{app:lkj},
\begin{align*}
\mbe_{q}\left[\log p_{\mathrm{LKJ}}\left(\lkjcorrmat\vert\alpha\right)\right] &
    =\left(\alpha-1\right)\left(\log\left|\lkjparam\right| -
    \psi_{K}\left(\frac{\nu}{2}\right) -
    \sum_{k=1}^{K}\log\left(\frac{1}{2}\lkjparam_{kk}\right)+
    K\psi\left(\frac{\nu-K+1}{2}\right)\right)+\constant,
\end{align*}
where $\constant$ contains terms that do not depend on $\alpha$,
and where $\psi_{K}$ denotes the multivariate digamma function. Consequently,
we can evaluate
\begin{align}
\efhess & =\left.\frac{\partial}{\partial\eta^{\trans}}
\mbe_{q\left(\theta;\eta\right)}\left[\frac{\partial}{\partial\alpha}
\log p\left(\covmat\vert\alpha\right)\right]\right|_{\eta=\etaoptzero,\alpha =
    \alpha_{0}}\nonumber \\
 & =\left.\frac{\partial}{\partial\eta^{\trans}}
    \left(\log\left|\lkjparam\right|-\psi_{K}\left(
        \frac{n}{2}\right)-\sum_{k=1}^{K}
        \log\left(\frac{1}{2}\lkjparam_{kk}\right)+
        K\psi\left(\frac{n-K+1}{2}\right)\right)\right|_{\etaoptzero}.
        \label{eq:lkj_prior_vb}
\end{align}
This derivative has a closed form, but the bookkeeping required
to represent an unconstrained parameterization of the matrix $\lkjparam$
within $\eta$ would be tedious. In practice, we evaluate terms like
$\efhess$ using automatic differentiation tools \citep{baydin:2015:autodiff}.

Finally, in cases where we cannot evaluate $\mbe_{q\left(\theta;\eta\right)}\left[\log p\left(\theta\vert\alpha\right)\right]$
in closed form as a function of $\eta$, we can use numerical techniques
as described in \prettyref{subsec:lrvb_implementation}. We thus view
$\qsens$ as the exact sensitivity to an approximate KL divergence.

\subsection{Practical Considerations when Computing the Sensitivity
of Variational Approximations\label{subsec:lrvb_implementation}}

We briefly discuss practical issues in the computation of
\prettyref{eq:lrvb_formula}, which requires calculating the product
$\eggrad\klhess^{-1}$ (or, equivalently, $\klhess^{-1}\eggrad^{\trans}$
since $\klhess$ is symmetric). Calculating $\klhess$ and solving
this linear system can be the most computationally intensive part
of computing \prettyref{eq:lrvb_formula}.

We first note that it can be difficult and time consuming in practice
to manually derive and implement second-order derivatives. Even a
small programming error can lead to large errors in \prettyref{thm:lrvb_formula}.
To ensure accuracy and save analyst time, we evaluated all the requisite
derivatives using the Python \texttt{autograd} automatic differentiation
library \citep{maclaurin:2015:autograd} and the Stan math automatic
differentiation library \citep{carpenter:2015:stan}.

Note that the dimension of $\klhess$ is as large as that of $\eta$,
the parameters that specify the variational distribution $q\left(\theta;\eta\right)$.
Many applications of MFVB employ many latent variables, the number of
which may even scale with the amount of data---including several of the
cases that we examine in \prettyref{sec:experiments}. However, these applications
typically have special structure that render $\klhess$ sparse, allowing the
practitioner to calculate $\eggrad\klhess^{-1}$ quickly. Consider,
for example, a model with ``global'' parameters, $\theta_{glob}$,
that are shared by all the individual datapoint likelihoods, and ``local''
parameters, $\theta_{loc,n}$, associated with likelihood of a single
datapoint indexed by $n$. By ``global'' and ``local'' we mean
the likelihood and assumed variational distribution factorize as
\begin{align}
p\left(x,\theta_{glob},\theta_{loc,1},...,\theta_{loc,N}\right) &
    =p\left(\theta_{glob}\right)\prod_{n=1}^{N}p\left(x\vert\theta_{loc,n},
    \theta_{glob}\right)p\left(\theta_{loc,n}\vert\theta_{glob}\right)
    \label{eq:global_local}\\
q\left(\theta;\eta\right) &
    =q\left(\theta_{glob};\eta_{glob}\right)\prod_{n=1}^{N}
    q\left(\theta_{loc,n};\eta_{n}\right)
    \textrm{ for all }q\left(\theta;\eta\right)\in\mathcal{Q}.\nonumber
\end{align}
In this case, the second derivatives of the variational objective
between the parameters for local variables vanish:
\begin{align*}
\textrm{ for all }n\ne m\textrm{, }
\frac{\partial^{2}KL\left(q\left(\theta;\eta\right)||
    \pzeropost\left(\theta\right)\right)}{\partial\eta_{loc,n}
    \partial\eta_{loc,m}^{\trans}} & =0.
\end{align*}
The model in \prettyref{subsec:glmm_model} has such a global / local structure;
see \prettyref{subsec:glmm_inference} for more details. Additional discussion,
including the use of Schur complements to take advantage of sparsity in the log
likelihood, can be found in \citet{giordano:2015:lrvb}.

When even calculating or instantiating $\klhess$ is prohibitively
time-consuming, one can use conjugate gradient algorithms to approximately
compute $\klhess^{-1}\eggrad^{\trans}$ \citep[Chapter
5]{nocedalwright:1999:numerical}. The advantage of conjugate gradient algorithms
is that they approximate $\klhess^{-1}\eggrad^{\trans}$ using only the
Hessian-vector product $\klhess\eggrad^{\trans}$,
which can be computed efficiently using automatic differentiation without ever
forming the full Hessian $\klhess$. See, for example, the
\texttt{hessian\_vector\_product} method of the Python \texttt{autograd} package
\citep{maclaurin:2015:autograd}. Note that a separate conjugate gradient problem
must be solved for each column of $\eggrad^{\trans}$, so if the parameter of
interest $\gtheta$ is high-dimensional it may be faster to pay the price for
computing and inverting the entire matrix $\klhess$. See
\ref{subsec:glmm_inference} for more discussion of a specific example.

In \prettyref{thm:lrvb_formula}, we require $\etaoptzero$ to be at a true local
optimum. Otherwise the estimated sensitivities may not be reliable (e.g., the
covariance implied by \prettyref{eq:lrvb_for_covariance} may not be positive
definite). We find that the classical MFVB coordinate ascent algorithms
(\citet[Section 2.4]{blei:2016:variational}) and even quasi-second order
methods, such as BFGS \citep[e.g.,][]{regier:2015:celeste}, may not actually find
a local optimum unless run for a long time with very stringent convergence
criteria. Consequently, we recommend fitting models using second-order Newton
trust region methods. When the Hessian is slow to compute directly, as in
\prettyref{sec:experiments}, one can use the conjugate gradient trust region
method of \citet[ Chapter 7]{nocedalwright:1999:numerical}, which takes
advantage of fast automatic differentiation Hessian-vector products without
forming or inverting the full Hessian.

\section{Experiments
\label{sec:experiments} }

We now demonstrate the speed and effectiveness of linear response methods on a
number of simulated and real data sets.  We begin with simple simulated data to
provide intuition for how linear response methods can improve estimates of
covariance relative to MFVB and the Laplace approximation.  We then develop
linear response covariance estimates for ADVI and apply them to four real-world
models and data sets taken from the Stan examples library \citep{stan-examples:2017}. Finally, we calculate both
linear response covariances and prior sensitivity measures for a large-scale
industry data set.  In each case, we compare linear response methods with
ordinary MFVB, the Laplace approximation, and MCMC. We show that linear response
methods provide the best approximation to MCMC while still retaining the
speed of approximate methods. Code and instructions to reproduce the results of
this section can be found in the git repository
\href{https://github.com/rgiordan/CovariancesRobustnessVBPaper}{\texttt{rgiordan/CovariancesRobustnessVBPaper}}.
 
\subsection{Simple Expository Examples
\label{subsec:laplace_experiments}}

\newcommand{\simpleNumVBSamples}{10000}

In this section we provide a sequence of simple examples comparing MFVB and LRVB
with Laplace approximations. These examples provide intuition for the
covariance estimate $\lrvbcov\left(\gtheta\right)$ and illustrate how the
sensitivity analysis motivating $\lrvbcov\left(\gtheta\right)$ differs from the
local posterior approximation motivating $\covlap\left(\gtheta\right)$.

For each example, we will explicitly specify the target posterior
$\pzeropost\left(\theta\right)$ using a mixture of normals. This
will allow us to define known target distributions with varying degrees
of skewness, over-dispersion, or correlation and compare the truth
with a variational approximation. Formally, for some fixed $K_{z}$,
component indicators $z_{k}$, $k=1,...,K_{z}$, component probabilities
$\pi_{k}$, locations $\mu_{k}$, and covariances $\Sigma_{k}$, we
set
\begin{align*}
p\left(z\right) & =\prod_{k=1}^{K_{z}}\pi_{k}^{z_{k}}\\
\pzeropost\left(\theta\right) & =\sum_{z}p\left(z\right)
    p\left(\theta\vert z\right)=\sum_{z}p\left(z\right)
    \prod_{k=1}^{K_{z}}\mathcal{N}\left(\theta;m_{k},\covmat_{k}\right)^{z_{k}}.
\end{align*}
The values $\pi$, $m$ and $\Sigma$ will be chosen to achieve the desired shape
for each example using up to $K_{z}=3$ components. There will be no need to
state the precise values of $\pi$, $m$, and $\Sigma$; rather, we will
show plots of the target density and report the marginal means and variances,
calculated by Monte Carlo.\footnote{MFVB is often
used to approximate the posterior when the Bayesian generative model for data $x$ is a mixture model
(e.g., \citet{blei:2003:lda}). By contrast, we note for clarity that we
are \emph{not} using the mixture model as a generative model for $x$ here.
E.g., $z$ is not one of the parameters composing $\theta$, and we are not approximating the
distribution of $z$ in the variational distribution $q\left(\theta\right)$.
Rather, we are using mixtures as a way of flexibly defining skewed and
over-dispersed targets, $p\left(\theta\right)$.}

We will be interested in estimating the mean and variance of the first
component, so we take $\gtheta=\theta_{1}$. Consequently, in order
to calculate $\lrvbcov$$\left(\theta_{1}\right)$, we will be considering
the perturbation $\covdens=\alpha\theta_{1}$ with scalar $\alpha$
and $\alpha_{0}=0$.

For the variational approximations, we will use a factorizing normal
approximation:
\begin{align*}
\qsimple & =\left\{ q\left(\theta\right):q\left(\theta\right) =
    \prod_{k=1}^{K}\mathcal{N}\left(\theta_{k};
        \mu_{k},\sigma_{k}^{2}\right)\right\} .
\end{align*}
In terms of \prettyref{eq:q_approximating_family}, we take
$\eta=\left(\mu_{1},...,\mu_{K},\log\sigma_{1},...,\log\sigma_{K}\right)^{\trans}$.
Thus
$\mbe_{q\left(\theta;\eta\right)}\left[\gtheta\right]=\mbe_{q\left(\theta;\eta\right)}\left[\theta_{1}\right]=\mu_{1}$.
In the examples below, we will use multiple distinct components in the
definition of $\pzeropost\left(\theta\right)$, so that
$\pzeropost\left(\theta\right)$ is non-normal and
$\pzeropost\left(\theta\right)\notin\qsimple$.

Since the expectation $\mbe_{q\left(\theta;\eta\right)}\left[\log
p\left(\theta\right)\right]$ is intractable, we replace the exact KL divergence
with a Monte Carlo approximation using the ``re-parameterization trick''
\citep{kingma:2013:auto, rezende:2014:stochastic, titsias:2014:doubly}.
Let $\circ$ denote the Hadamard
(component-wise) product. Let $\xi_{m}\iid\mathcal{N}\left(0,I_{K}\right)$ for
$m=1,...,M$. We define
\begin{align*}
\theta_{m} & :=\sigma\circ\xi_{m}+\mu\\
KL_{approx}\left(q\left(\theta;\eta\right)||\pzeropost\left(\theta\right)\right) & :=-\frac{1}{M}\sum_{m=1}^{M}\log\pzeropost\left(\theta_{m}\right)-\sum_{k=1}^{K}\log\sigma_{k},
\end{align*}
which is a Monte Carlo estimate of
$KL\left(q\left(\theta;\eta\right)||\pzeropost\left(\theta\right)\right)$.
We found $M=\simpleNumVBSamples$ to be more than adequate for our present
purposes of illustration. Note that we used the same draws $\xi_{m}$ for both
optimization and for the calculation of $\klhess$ in order to ensure that the
$\etaoptzero$ at which $\klhess$ was evaluated was in fact an optimum. This
approach is similar to our treatment of ADVI;
see \prettyref{subsec:advi} for a more detailed discussion.

\subsubsection{Multivariate Normal Targets}

If we take only a single component in the definition of
$\pzeropost\left(\theta\right)$ ($K_{z}=1$), then
$\pthetapost\left(\theta\right)$ is a multivariate normal distribution for all
$\alpha$, and the Laplace approximation $\lap\left(\theta\right)$ is equal to
$\pthetapost\left(\theta\right)$ for all $\alpha$. Furthermore, as discussed in
\prettyref{subsec:lrvb_cov} and \prettyref{app:mvn_exact}, the variational means
$\mbe_{\qthetapost}\left[\theta\right]=\mu$ are exactly equal to the exact
posterior mean $\mbe_{\pthetapost}\left[\theta\right]=m_{1}$ for all $\alpha$
(even though in general $\qcov\left(\theta\right)\ne\Sigma_{1}$). Consequently,
for all $\alpha$, the variational approximation, the Laplace approximation, and
the exact $\pzeropost\left(\theta\right)$ all coincide in their estimates of
$\mbe\left[\theta\right]$, and by, \prettyref{cor:lrvb_accurate},
$\Sigma=\pcov\left(\theta\right)=\lrvbcov\left(\theta\right)=\covlap\left(\theta\right)$.
Of course, if $\Sigma$ is not diagonal, $\qcov\left(\theta\right)\ne\Sigma$
because of the mean field assumption. Since this argument holds for the whole
vector $\theta$, it holds \textit{a fortiori} for our quantity of interest, the
first component $g\left(\theta\right)=\theta_{1}$.

In other words, the Laplace approximation will differ only from the
LRVB approximation when $\pzeropost\left(\theta\right)$ is not multivariate
normal, a situation that we will now bring about by adding new components
to the mixture; i.e., by increasing $K_{z}$.

\subsubsection{A Univariate Skewed Distribution}

\begin{knitrout}
\definecolor{shadecolor}{rgb}{0.969, 0.969, 0.969}\color{fgcolor}\begin{figure}[t]

{\centering \includegraphics[width=0.98\linewidth,height=0.343\linewidth]{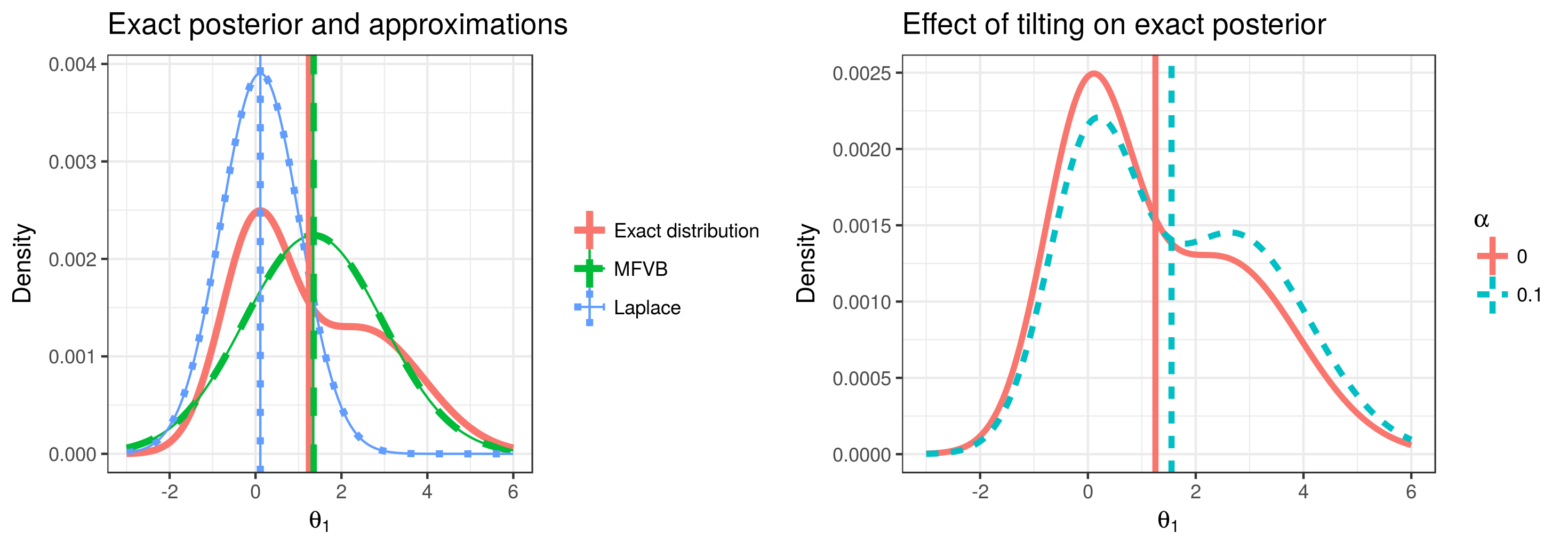} 

}

\caption[A univariate skewed distribution]{A univariate skewed distribution. Vertical lines show the location of the means.}\label{fig:SimpleExampleSkew1d}
\end{figure}

\end{knitrout}

\begin{figure}
\noindent\begin{minipage}[b]{1\columnwidth}
\begin{minipage}[c][1\totalheight][b]{0.49\columnwidth}\begin{knitrout}
\definecolor{shadecolor}{rgb}{0.969, 0.969, 0.969}\color{fgcolor}

{\centering \includegraphics[width=0.98\linewidth,height=0.686\linewidth]{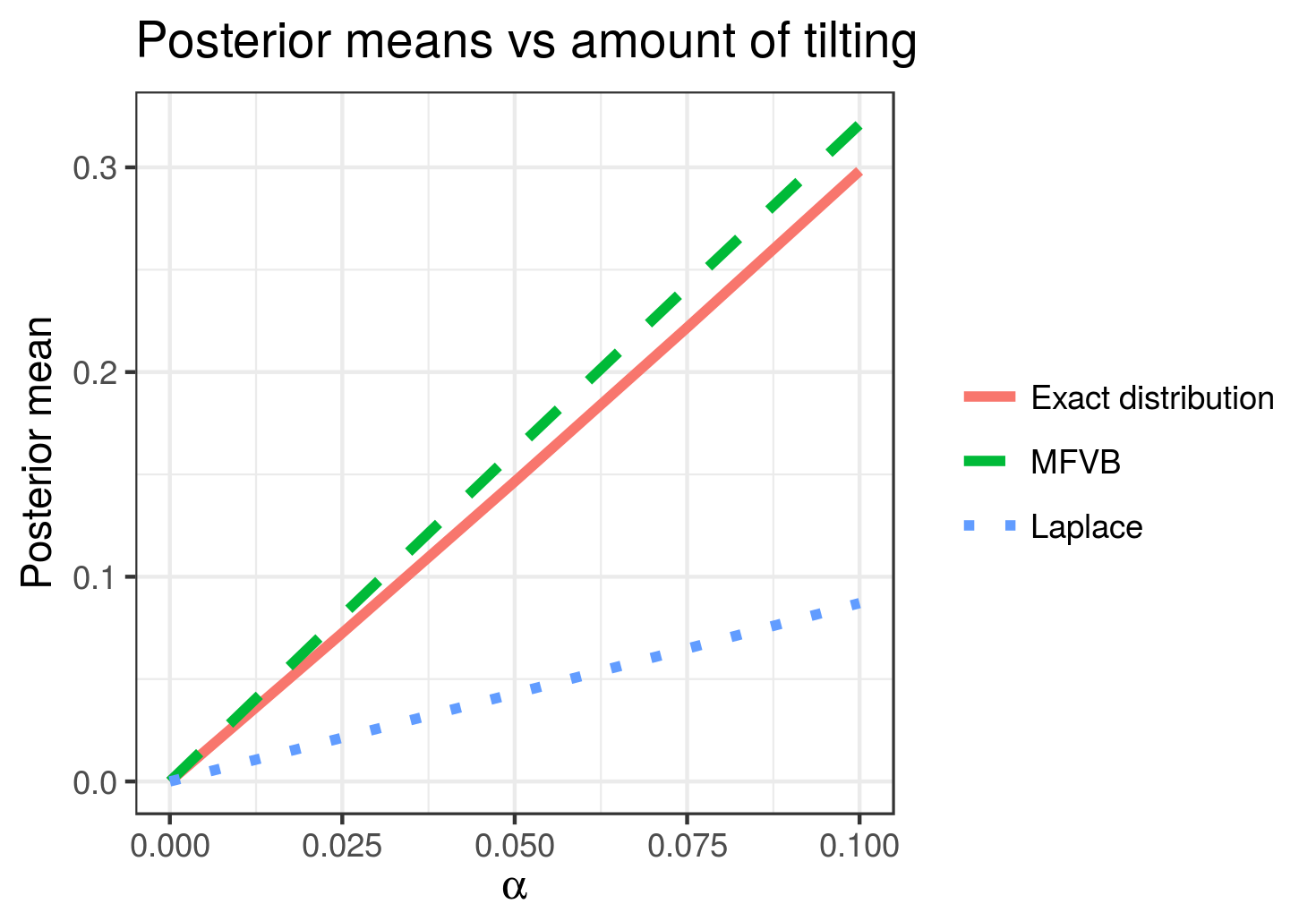} 

}

\end{knitrout}
\end{minipage}
\begin{minipage}[b]{0.49\columnwidth}\begin{center}\begin{tabular}{lrrrr}
  \hline
Metric & Exact & LRVB & MFVB & Laplace \\ 
  \hline
mean & 1.253 &  & 1.345 & 0.111 \\ 
  var & 2.872 & 3.245 & 2.599 & 0.849 \\ 
   \hline
\end{tabular}
\end{center}
\end{minipage}

\end{minipage}
\caption{Effect of tilting on a univariate skew distribution. \label{fig:SimpleExampleSkew1dResult}}
\end{figure}

If we add a second component ($K_{z}=2$), then we can make
$\pzeropost\left(\theta\right)$
skewed, as shown (with the approximations) in \fig{SimpleExampleSkew1d}.
In this case, we expect $\mbe_{\qthetapost}\left[\theta_{1}\right]$
to be more accurate than the Laplace approximation
$\mbe_{\lap}\left[\theta_{1}\right]$
because $\qsimple$ is more expressive than $\qlap$. This intuition
is born out in the left panel of \fig{SimpleExampleSkew1d}. Since
$\laptheta$ uses only information at the mode, it fails to take into
account the mass to the right of the mode, and the Laplace approximation's
mean is too far to the left. The MFVB approximation, in contrast, is
quite accurate for the posterior mean of $\theta_{1}$, even though
it gets the overall shape of the distribution wrong.

This example also shows why, in general, one cannot naively form a
``Laplace approximation'' to the posterior centered at the variational
mean rather than at the MAP. As shown in the left panel of \fig{SimpleExampleSkew1d},
in this case the posterior distribution is actually convex at the
MFVB mean. Consequently, a naive second-order approximation to the log
posterior centered at the MFVB mean would imply a negative variance.

The perturbation $\covdens=\alpha\theta_{1}$ is sometimes also described
as a ``tilting,'' and the right panel of \fig{SimpleExampleSkew1d}
shows the effect of tilting on this posterior approximation. Tilting
increases skew, but the MFVB approximation remains accurate, as shown
in \fig{SimpleExampleSkew1dResult}. Since local sensitivity of
the expectation of $\theta_{1}$ to $\alpha$ is the variance of $\theta_{1}$
(see \prettyref{eq:sensitivity_is_cov}), we have in
\fig{SimpleExampleSkew1dResult} that:
\begin{itemize}
\item The slope of the exact distribution's line is
    $\pcov\left(\theta_{1}\right)$;
\item The slope of the MFVB line is the LRVB variance
    $\lrvbcov\left(\theta_{1}\right)$; and
\item The slope of the Laplace line is $\covlap$$\left(\theta_{1}\right)$.
\end{itemize}
Since the MFVB and exact lines nearly coincide, we expect the LRVB
variance estimate to be quite accurate for this example. Similarly, since the
slope of the Laplace approximation line is lower, we expect the Laplace variance
to underestimate the exact variance. This outcome, which can be seen visually
in the left-hand panel of \fig{SimpleExampleSkew1dResult}, is shown
quantitatively in the corresponding table in the right-hand panel. The
columns of the table contain information for the exact distribution and the
three approximations. The first row, labeled ``mean,'' shows
$\mbe\left[\theta_1\right]$ and the second row, labeled ``var,'' shows
$\cov\left(\theta_1\right)$. (The ``LRVB'' entry for the mean is blank because
LRVB differs from MFVB only in covariance estimates.) We conclude that, in this
case, \prettyref{cond:vb_accurate} holds for $\qsimple$ but not for $\qlap$.

\subsubsection{A Univariate Over-dispersed Distribution}

\begin{knitrout}
\definecolor{shadecolor}{rgb}{0.969, 0.969, 0.969}\color{fgcolor}\begin{figure}[t]

{\centering \includegraphics[width=0.98\linewidth,height=0.343\linewidth]{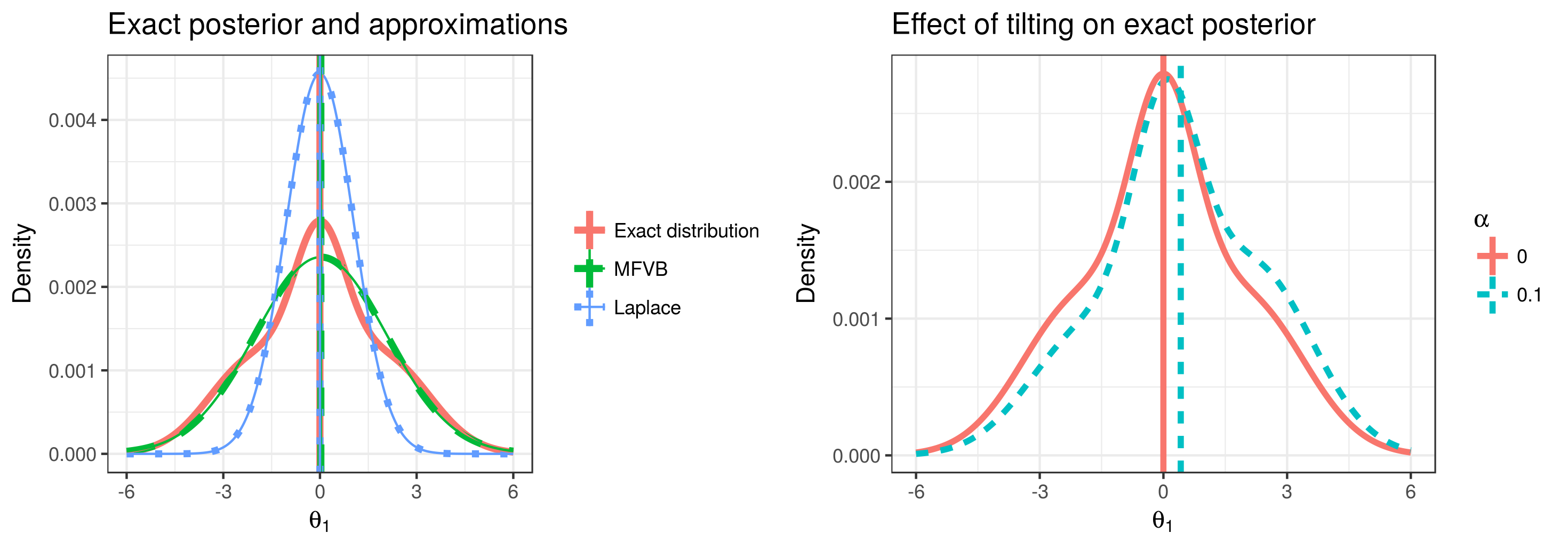} 

}

\caption[A univariate over-dispersed distribution]{A univariate over-dispersed distribution. Vertical lines show the location of the means.}\label{fig:SimpleExampleNoSkew1d}
\end{figure}

\end{knitrout}

\begin{figure}
\noindent\begin{minipage}[b]{1\columnwidth}
\begin{minipage}[c][1\totalheight][b]{0.49\columnwidth}\begin{knitrout}
\definecolor{shadecolor}{rgb}{0.969, 0.969, 0.969}\color{fgcolor}

{\centering \includegraphics[width=0.98\linewidth,height=0.686\linewidth]{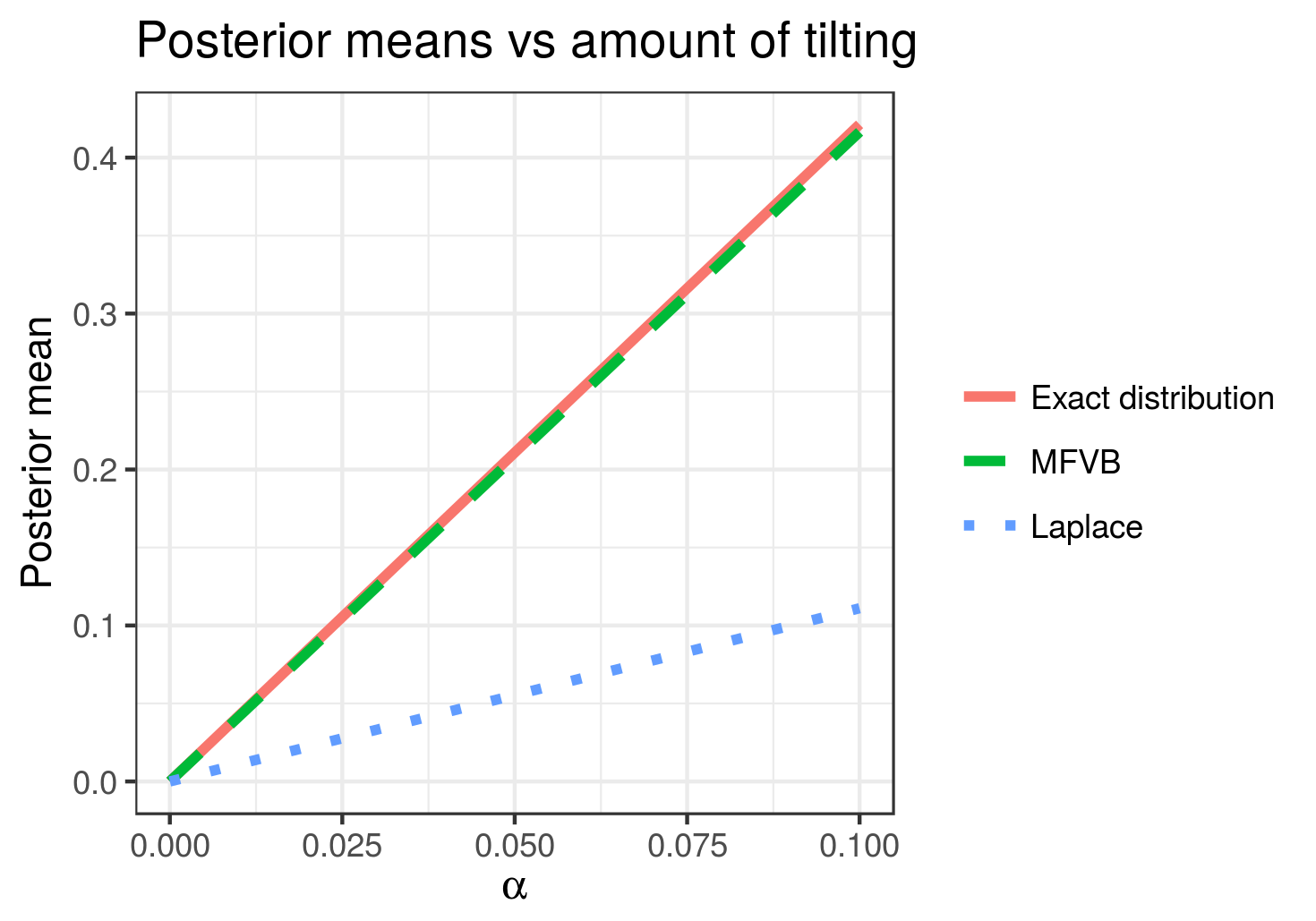} 

}

\end{knitrout}
\end{minipage}
\begin{minipage}[b]{0.49\columnwidth}\begin{center}\begin{tabular}{lrrrr}
  \hline
Metric & Exact & LRVB & MFVB & Laplace \\ 
  \hline
mean & -0.001 &  & 0.027 & -0.000 \\ 
  var & 4.218 & 4.153 & 4.161 & 1.107 \\ 
   \hline
\end{tabular}
\end{center}
\end{minipage}

\end{minipage}
\caption{Effect of tilting on a univariate over-dispersed distribution. \label{fig:SimpleExampleNoSkew1dResult}}
\end{figure}

Having seen how MFVB can outperform the Laplace approximation for a univariate
skewed distribution, we now apply that intuition to see why the linear response
covariance can be superior to the Laplace approximation covariance for
over-dispersed but symmetric distributions. Such a symmetric but over-dispersed
distribution, formed with $K_{z}=3$ components, is shown in
\fig{SimpleExampleNoSkew1d} together with its approximations. By symmetry, both
the MFVB and Laplace means are exactly correct (up to Monte Carlo error), as can
be seen in the left panel of \fig{SimpleExampleNoSkew1d}.

However, the right panel of \fig{SimpleExampleNoSkew1d} shows that
symmetry is not maintained as the distribution is tilted. For $\alpha>0$,
the distribution becomes skewed to the right. Thus, by the intuition
from the previous section, we expect the MFVB mean to be more accurate
as the distribution is tilted and $\alpha$ increases from zero. In
particular, we expect that the Laplace approximation's mean will not
shift enough as $\alpha$ varies, i.e., that the Laplace approximation
variance will be underestimated. \fig{SimpleExampleNoSkew1dResult}
shows that this is indeed the case. The slopes in the left panel once
again correspond to the estimated variances shown in the table, and,
as expected the LRVB variance estimate is superior to the Laplace
approximation variance.

In this case, \prettyref{cond:vb_accurate} holds for $\qsimple$.
For the Laplace approximation, $\mbe_{\lap}\left[\gtheta\right] =
\mbe_{\pzeropost}\left[\gtheta\right]$
for $\alpha=0$, so $\qlap$ satisfies \prettyref{eq:vb_accurate} of \prettyref{cond:vb_accurate} for
$\alpha$ near zero,
the derivatives of the two expectations with respect to
$\alpha$ are quite different,
so \prettyref{eq:vb_slope_accurate}
of \prettyref{cond:vb_accurate} does not hold for $\qlap$.

\subsubsection{A Bivariate Over-dispersed Distribution}

\begin{knitrout}
\definecolor{shadecolor}{rgb}{0.969, 0.969, 0.969}\color{fgcolor}\begin{figure}[t]

{\centering \includegraphics[width=0.98\linewidth,height=0.343\linewidth]{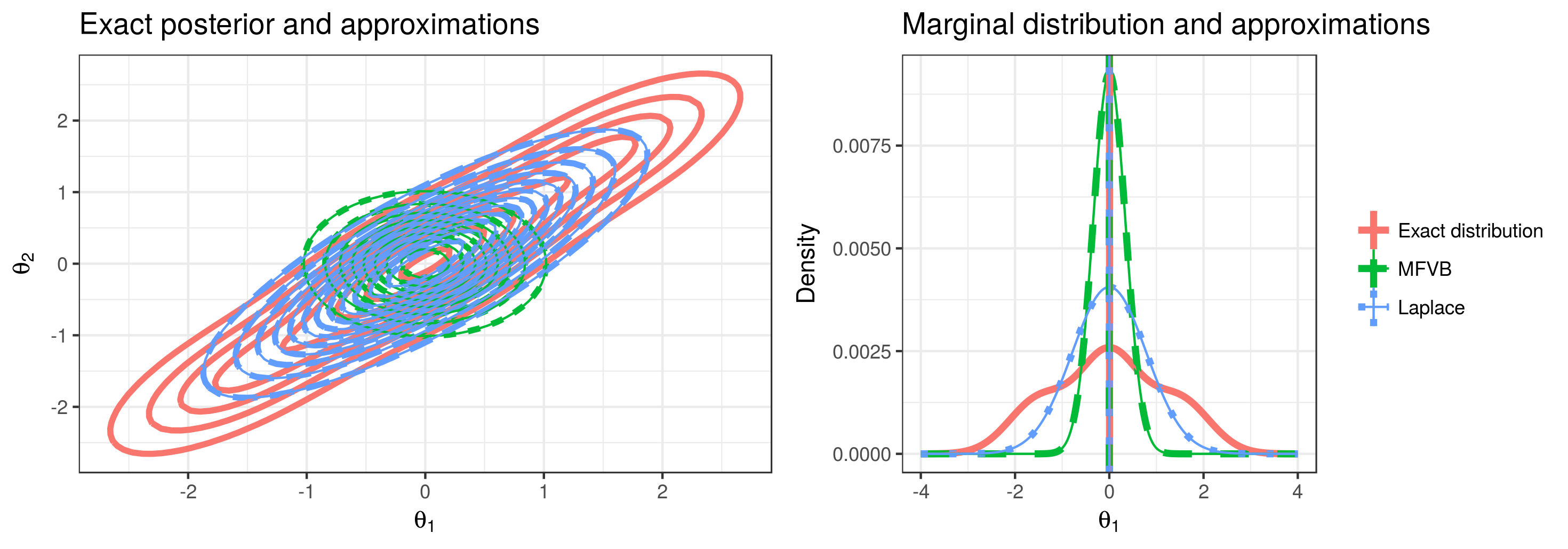} 

}

\caption[A bivariate over-dispersed distribution]{A bivariate over-dispersed distribution.}\label{fig:SimpleExampleNoSkew2d}
\end{figure}

\end{knitrout}

\begin{figure}
\noindent\begin{minipage}[b]{1\columnwidth}
\begin{minipage}[c][1\totalheight][b]{0.49\columnwidth}\begin{knitrout}
\definecolor{shadecolor}{rgb}{0.969, 0.969, 0.969}\color{fgcolor}

{\centering \includegraphics[width=0.98\linewidth,height=0.686\linewidth]{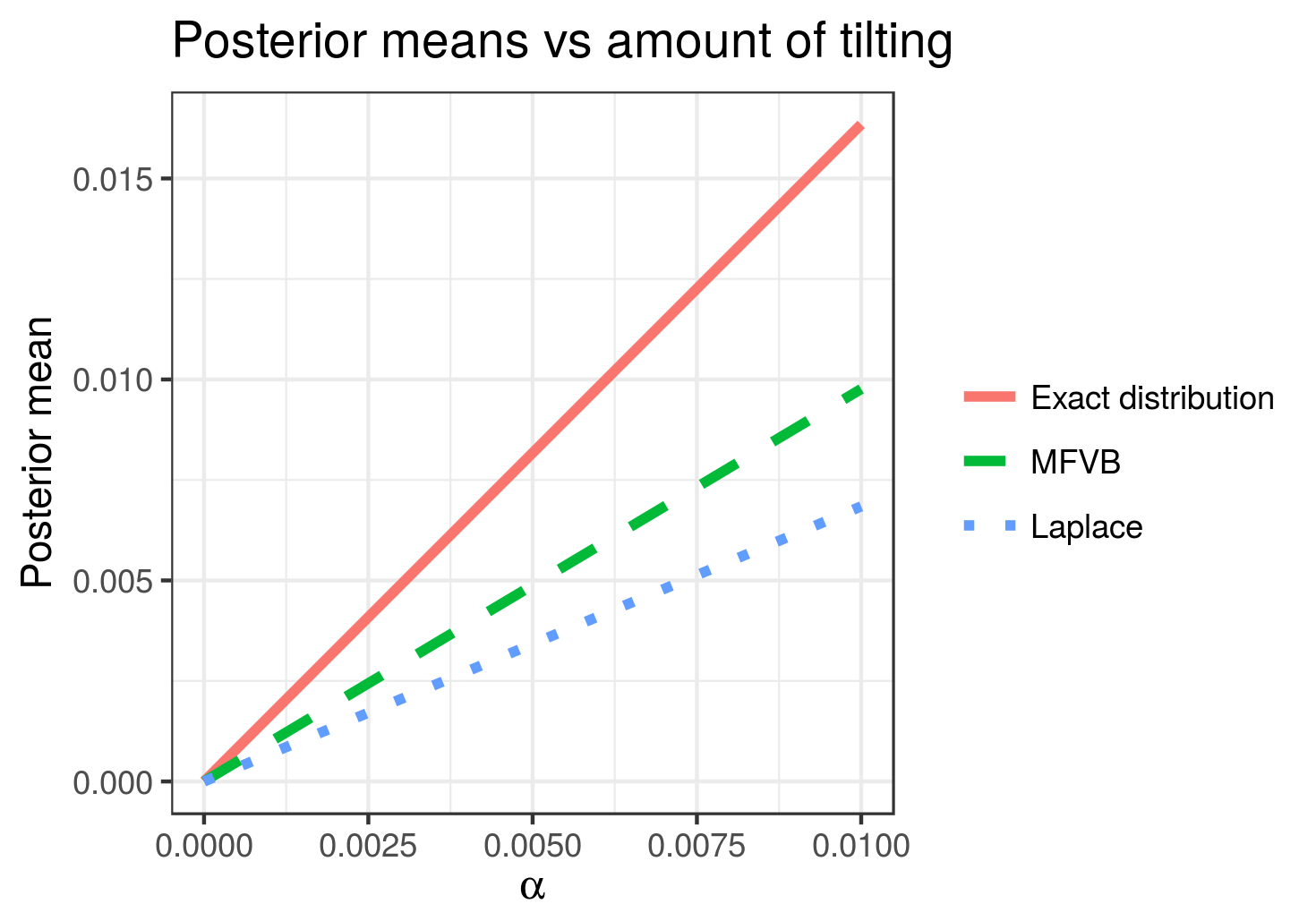} 

}

\end{knitrout}
\end{minipage} \begin{minipage}[b]{0.49\columnwidth}\begin{center}\begin{tabular}{lrrrr}
  \hline
Metric & Exact & LRVB & MFVB & Laplace \\ 
  \hline
mean & 0.005 &  & -0.002 & -0.000 \\ 
  var & 1.635 & 0.976 & 0.241 & 0.684 \\ 
   \hline
\end{tabular}
\end{center}
\end{minipage}

\end{minipage}

\caption{Effect of tilting on a bivariate over-dispersed distribution. \label{fig:SimpleExampleNoSkew2dResult}}
\end{figure}

In the previous two examples the mean field approximation in $\mathcal{Q}$ did
not matter, since the examples were one-dimensional. The only reason that the
variational approximation was different from the exact
$\pzeropost\left(\theta\right)$ was the normal assumption in $\qsimple$. Indeed,
the tables in \fig{SimpleExampleSkew1dResult} and
\fig{SimpleExampleNoSkew1dResult} show that the MFVB variance estimate is also
reasonably close to the exact variance. In order to demonstrate why the LRVB
variance can be better than both the Laplace approximation and the MFVB
approximation, we turn to a bivariate, correlated, over-dispersed
$\pzeropost\left(\theta\right)$. For this we use $K_{z}=3$ correlated normal
distributions, shown in the left panel of \fig{SimpleExampleNoSkew2d}. The right
panel of \fig{SimpleExampleNoSkew2d} shows the marginal distribution of
$\theta_{1}$, in which the over-dispersion can be seen clearly.
As \fig{SimpleExampleNoSkew2d} shows, unlike in the previous two
examples, the mean field approximation causes $\qzeropost\left(\theta\right)$
to dramatically underestimate the marginal variance of $\theta_{1}$.
Consequently, the MFVB means will also be under-responsive to the
skew introduced by tilting with $\alpha$. Though the Laplace approximation
has a larger marginal variance, it remains unable to take skewness
into account. Consequently, as seen in \fig{SimpleExampleNoSkew2dResult},
the LRVB variance, while not exactly equal to the correct variance,
is still an improvement over the Laplace covariance, and a marked
improvement on the badly under-estimated MFVB variance.

One might say, in this case, that \prettyref{cond:vb_accurate} does
not hold for either $\qsimple$ or $\qlap$, or, if it does, it is
with a liberal interpretation of the ``approximately equals'' sign.
However, the expressiveness of $\qsimple$ allows LRVB to improve
on the Laplace approximation, and the linear response allows it to
improve over the MFVB approximation, and so LRVB gives the best of
both worlds.

Thinking about problems in terms of these three simple models can
provide intuition about when and whether \prettyref{cond:vb_accurate}
might be expected to hold in a sense that is practically useful.
 
\subsection{Automatic Differentiation Variational Inference (ADVI)
\label{subsec:advi}}

\newcommand{\adviNumVBSamples}{10}

In this section we apply our methods to automatic differentiation
variational inference (ADVI) \citep{kucukelbir:2017:advi}. ADVI is
a ``black-box'' variational approximation and optimization procedure
that requires only that the user provide the log posterior,
$\log \pzeropost\left(\theta\right)$, up to a constant that does
not depend on $\theta$. To achieve this generality,
ADVI employs:
\begin{itemize}
\item A factorizing normal variational approximation,\footnote{\citet{kucukelbir:2017:advi}
    describe a non-factorizing version of ADVI, which is called
    ``fullrank'' ADVI in Stan.  The factorizing version
    that we describe here is called ``meanfield'' ADVI in Stan.
    On the examples we describe, in the current Stan implementation,
    we found that fullrank ADVI provided much worse approximations to
    the MCMC posterior
    means than the meanfield version, and so we do not consider it further.}
\item An unconstraining parameterization,
\item The ``re-parameterization trick,'' and
\item Stochastic gradient descent.
\end{itemize}
ADVI uses a family employing the factorizing normal approximation
\begin{align*}
\qadvi & :=\left\{ q\left(\theta\right):q\left(\theta\right) =
    \prod_{k=1}^{K}\mathcal{N}\left(\theta_{k}\vert\mu_{k},
        \exp\left(2\zeta_{k}\right)\right)\right\} .
\end{align*}
That is, $\qadvi$ is a fully factorizing normal family with means
$\mu_{k}$ and log standard deviations $\zeta_{k}$. Because we are
including exponential family assumptions in the definition of MFVB
(as described in \prettyref{subsec:variational_Bayes}), $\qadvi$
is an instance of a mean-field family $\qmfvb$. In the notation of
\prettyref{eq:q_approximating_family},
\begin{align}
\eta & =\left(\mu_{1},...,\mu_{K},\zeta_{1},...,\zeta_{K}\right)^{\trans},
\label{eq:advi_parameter_def}
\end{align}
 $\Omega_{\eta}=\mathbb{R}^{2K}$, $\lambda$ is the Lebesgue measure,
and the objective function \prettyref{eq:kl_divergence} is
\begin{align*}
KL\left(q\left(\theta;\eta\right)||\pzeropost\left(\theta\right)\right) &
    =-\int\mathcal{N}\left(\theta_{k}\vert\mu_{k},
        \exp\left(2\zeta_{k}\right)\right)\log\pzeropost\left(\theta\right)
        \lambda\left(d\theta\right)-\sum_{k=1}^{K}\zeta_{k},
\end{align*}
where we have used the form of the univariate normal entropy up to a
constant.

The unconstraining parameterization is required because the use of
a normal variational approximation dictates that the base measure
on the parameters $\theta\in\mathbb{R}^{K}$ be supported on all of
$\mathbb{R}^{K}$. Although many parameters of interest, such as covariance
matrices, are not supported on $\mathbb{R}^{K}$, there typically
exist differentiable maps from an unconstrained parameterization supported
on $\mathbb{R}^{K}$ to the parameter of interest. Software packages
such as Stan automatically provide such transforms for a broad set
of parameter types.
In our notation, we will take these constraining maps to be the function of
interest, $\gtheta$, and take $\theta$ to be unconstrained.  Note that, under
this convention, the prior $p\left(\theta\vert\alpha\right)$ must be a density
in the unconstrained space.  In practice (e.g., in the Stan software package),
one usually specifies the prior density in the constrained space and converts it
to a density $p\left(\theta\vert\alpha\right)$ in the unconstrained space using
the determinant of the Jacobian of the constraining transform
$g\left(\cdot\right)$.

The re-parameterization trick allows easy
approximation of derivatives of the (generally intractable) objective
$KL\left(q\left(\theta;\eta\right)||\pzeropost\left(\theta\right)\right)$.
By defining $z_{k}$ using the change of variable
\begin{align}
z_k &:= (\theta_k - \mu_k) / \exp\left(\zeta_{k}\right), \label{eq:advi_theta}
\end{align}
$KL\left(q\left(\theta;\eta\right)||\pzeropost\left(\theta\right)\right)$
can be re-written as an expectation with respect to a standard normal
distribution. We write $\theta=\exp\left(\zeta\right)\circ z+\mu$
by using the component-wise Hadamard product $\circ$. Then
\begin{align*}
KL\left(q\left(\theta;\eta\right)||\pzeropost\left(\theta\right)\right) & =
    -\mbe_{z}\left[\log\pzeropost\left(\exp\left(\zeta\right)\circ z +
    \mu\right)\right]-\sum_{k=1}^{K}\zeta_{k}+\constant.
\end{align*}
The expectation is still typically intractable, but it can be approximated
using Monte Carlo and draws from a $K$-dimensional standard normal
distribution. For a fixed number $M$ of draws $z_{1},...,z_{M}$
from a standard $K$-dimensional normal, we can define the approximate
KL divergence
\begin{align}
\widehat{KL}\left(\eta\right) &
    :=-\frac{1}{M}\sum_{m=1}^{M}\log\pzeropost\left(
        \exp\left(\zeta\right)\circ z_{m}+\mu\right) -
        \sum_{k=1}^{K}\zeta_{k} +\constant.\label{eq:monte_carlo_approx_kl}
\end{align}
For any fixed $M$,
\begin{align*}
\mbe\left[\frac{\partial}{\partial\eta}\widehat{KL}\left(\eta\right)\right] &
    =\frac{\partial}{\partial\eta}KL\left(q\left(\theta;\eta\right)||
        \pzeropost\left(\theta\right)\right),
\end{align*}
so gradients of $\widehat{KL}\left(\eta\right)$ are unbiased for
gradients of the exact KL divergence. Furthermore, for fixed draws
$z_{1},...,z_{M}$, $\widehat{KL}\left(\eta\right)$ can be easily
differentiated (using, again, the re-parameterization trick). Standard
ADVI uses this fact to optimize
$KL\left(q\left(\theta;\eta\right)||\pzeropost\left(\theta\right)\right)$
using the unbiased gradient draws
$\frac{\partial}{\partial\eta}\widehat{KL}\left(\eta\right)$
and a stochastic gradient optimization method, where the stochasticity
comes from draws of the standard normal random variable $z$. Note
that stochastic gradient methods typically use a new draw of $z$
at every gradient step.

\subsubsection{Linear Response for ADVI (LR-ADVI)}

Since ADVI uses a factorizing normal approximation, the intuition
from \prettyref{subsec:laplace_experiments} may be expected to apply.
In particular, we might expect that the ADVI means $\hat{\mu}$ might
be a good approximation to $\mbe_{\pzeropost}\left[\theta\right]$,
that the ADVI variances $\exp\left(2\hat{\zeta}\right)$ would be
under-estimates of the posterior variance $\pcov\left(\theta\right)$,
so that using $\lrvbcov\left(\theta\right)$ could improve the approximations
to the posterior variance.  We refer to LRVB covariances calculated using
an ADVI approximation as LR-ADVI.

To apply linear response to an ADVI approximation, we need to be able to
approximate the Hessian of
$KL\left(q\left(\theta;\eta\right)||\pzeropost\left(\theta\right)\right)$ and to
be assured that we have found an optimal $\etaoptzero$. But, by using a
stochastic gradient method, ADVI avoids ever actually calculating the
 expectation in
$KL\left(q\left(\theta;\eta\right)||\pzeropost\left(\theta\right)\right)$.
Furthermore even if a stochastic gradient method finds an point that is close
to the optimal value of
$KL\left(q\left(\theta;\eta\right)||\pzeropost\left(\theta\right)\right)$ it may
not be close to an optimum of $\widehat{KL}\left(\eta\right)$ for a particular
finite $M$. Indeed, we found that, even for very large $M$, the
optimum found by ADVI's stochastic gradient method is typically not close enough
to an optimum of the approximate $\widehat{KL}\left(\eta\right)$ for sensitivity
calculations to be useful. Sensitivity calculations are based on differentiating
the fixed point equation given by the gradient being zero (see the proof in
\prettyref{app:lrvb}), and do not apply at points for which the gradient is not
zero either in theory nor in practice.

Consequently, in order to calculate the local sensitivity, we simply eschew the
stochastic gradient method and directly optimize $\widehat{KL}\left(\eta\right)$
for a particular choice of $M$. (We will discuss shortly how to choose $M$.) We
can then use $\widehat{KL}\left(\eta\right)$ in \prettyref{eq:lrvb_formula}
rather than the exact KL divergence. Directly optimizing
$\widehat{KL}\left(\eta\right)$ both frees us to use
second-order optimization methods, which we found to converge more quickly to a
high-quality optimum than first-order methods, and guarantees that we are
evaluating the Hessian $\klhess$ at an optimum of the objective function used to
calculate \prettyref{eq:lrvb_formula}.

As $M$ approaches infinity, we expect the optimum of
$\widehat{KL}\left(\eta\right)$ to approach the optimum of
$KL\left(q\left(\theta;\eta\right)||\pzeropost\left(\theta\right)\right)$ by the
standard frequentist theory of estimating equations \citep[Chapter
9]{keener:2011:theoretical}.  In practice we must fix a particular finite $M$,
with larger $M$ providing better approximations of the true KL divergence but at
increased computational cost.  We can inform this tradeoff between accuracy and
computation by considering the frequentist variability of $\etaoptzero$ when
randomly sampling $M$ draws of the random variable $z$ used to approximate the
intractable integral in $\widehat{KL}\left(\eta\right)$. Denoting this
frequentist variability by $\cov_{z}\left(\etaoptzero\right)$, standard results
\citep[Chapter 9]{keener:2011:theoretical} give that
\begin{align}
\cov_{z}\left(\etaoptzero\right) &
    \approx\klhess^{-1}\cov_{z}\left(\left.\frac{\partial}{\partial\eta}
    \widehat{KL}\left(\eta\right)\right|_{\etaoptzero}\right)\klhess^{-1}.
    \label{eq:approx_kl_variance_of_eta}
\end{align}
A sufficiently large $M$ will be one for which
$\cov_{z}\left(\etaoptzero\right)$ is adequately small. One notion of
``adequately small'' might be that the ADVI means found with
$\widehat{KL}\left(\eta\right)$ are within some fraction of a posterior standard
deviation of the optimum of
$KL\left(q\left(\theta;\eta\right)||\pzeropost\left(\theta\right)\right)$.
Having chosen a particular $M$, we can calculate the frequentist variability of
$\mu^{*}$ using $\lrvbcov\left(\gtheta\right)$ and estimate the posterior
standard deviation using \prettyref{eq:lrvb_for_covariance}. If we find that
each $\mu^{*}$ is probably within 0.5 standard deviations of the optimum of
$KL\left(q\left(\theta;\eta\right)||\pzeropost\left(\theta\right)\right)$, we
can keep the results; otherwise, we increase $M$ and try again. In the examples
we consider here, we found that the relatively modest $M=\adviNumVBSamples$
satisfies this condition and provides sufficiently accurate results.

Finally, we note a minor departure from
\prettyref{eq:lrvb_for_covariance} when calculating $\lrvbcov\left(\gtheta\right)$
from $\klhess$.  Recall that, in this case, we are taking $g\left(\cdot\right)$
to be ADVI's constraining transform, and that \prettyref{eq:lrvb_for_covariance}
requires the Jacobian, $\eggrad$, of this transform.  At the time of writing,
the design of the Stan software package did not readily support automatic
calculation of $\eggrad$, though it did support rapid evaluation of $\gtheta$ at
particular values of $\theta$.  Consequently, we used linear response to
estimate $\lrvbcov\left(\theta\right)$, drew a large number $N_{s}$ of
Monte Carlo draws from
$\theta_{n}\sim\mathcal{N}\left(\mu,\lrvbcov\left(\theta\right)\right)$ for
$n=1,...,N_{s}$, and then used these draws to form a Monte Carlo estimate of the
sample covariance of $\gtheta$.
Noting that $\mbeq\left[\theta\right]=\mu$, and recalling the definition
of $\eta$ for ADVI in \prettyref{eq:advi_parameter_def},
by \prettyref{eq:lrvb_for_covariance} we have
\begin{align*}
\lrvbcov\left(\theta\right) & =
    \frac{\partial\mbeq\left[\theta\right]}{\partial\eta^{\trans}}
    \klhess^{-1}\frac{\partial\mbeq\left[\theta^{\trans}\right]}{\partial\eta} =
    \left(\begin{array}{cc}
I_{K} & 0\\
0 & 0
\end{array}\right)\klhess^{-1}\left(\begin{array}{cc}
I_{K} & 0\\
0 & 0
\end{array}\right),
\end{align*}
which is the upper-left quarter of the matrix $\klhess^{-1}$.
In addition to obviating the need for $\eggrad$,
this approach also allowed us to take into account possible nonlinearities in
$g\left(\cdot\right)$ at little additional computational cost.

\subsubsection{Results}

We present results from four models taken from the Stan example set, namely the models
\texttt{election88} (``Election model''), \texttt{sesame\_street1} (``Sesame
Street model''), \texttt{radon\_vary\_intercept\_floor}
(``Radon model''), and
\texttt{cjs\_cov\_randeff} (``Ecology model''). We experimented with many models
from the Stan examples and selected these four as representative of the type
of model where LR-ADVI can be expected to provide a benefit---specifically,
they are models of a moderate size. For very small models, MCMC
runs quickly enough in Stan that fast approximations are not necessary, and for
very large models (with thousands of parameters) the relative advantages of
LR-ADVI and the Laplace approximation diminish due to the need to calculate
$\klhess$ or $\laphess$ using automatic differentiation.\footnote{We calculated
$\klhess$ using a custom branch of Stan's automatic differentiation software
\citep{carpenter:2015:stan} that exposes Hessians and Hessian-vector products in
the \texttt{Rstan} \texttt{modelfit} class. When this custom branch is merged
with the main branch of Stan, it will be possible to implement LR-ADVI for
generic Stan models.} The size of the data and size of the parameter space for our four
chosen models are shown in \fig{ADVITimingGraph}. We also eliminated from
consideration models where Stan's MCMC algorithm reported divergent transitions
or where Stan's ADVI algorithm returned wildly inaccurate posterior mean
estimates.

For brevity, we do not attempt to describe the models or data in any
detail here; rather, we point to the relevant literature in their
respective sections. The data and Stan implementations themselves
can be found on the Stan website \citep{stan-examples:2017} as well as in
\prettyref{app:advi}.

To assess the accuracy of each model, we report means and standard deviations
for each of Stan's model parameters as calculated by Stan's MCMC and ADVI
algorithms and a Laplace approximation, and we report the standard deviations as
calculated by $\lrvbcov\left(\gtheta\right)$. Recall that, in our notation,
$g\left(\cdot\right)$ is the (generally nonlinear) map from the unconstrained
latent ADVI parameters to the constrained space of the parameters of interest.
The performance of ADVI and Laplace vary, and only LR-ADVI provides a
consistently good approximation to the MCMC standard deviations. LR-ADVI was
somewhat slower than a Laplace approximation or ADVI alone, but it was typically
about five times faster than MCMC; see \prettyref{subsec:advi_timing} for
detailed timing results.

\subsubsection{Election Model Accuracy}
\begin{knitrout}
\definecolor{shadecolor}{rgb}{0.969, 0.969, 0.969}\color{fgcolor}\begin{figure}[t]

{\centering \includegraphics[width=0.98\linewidth,height=0.514\linewidth]{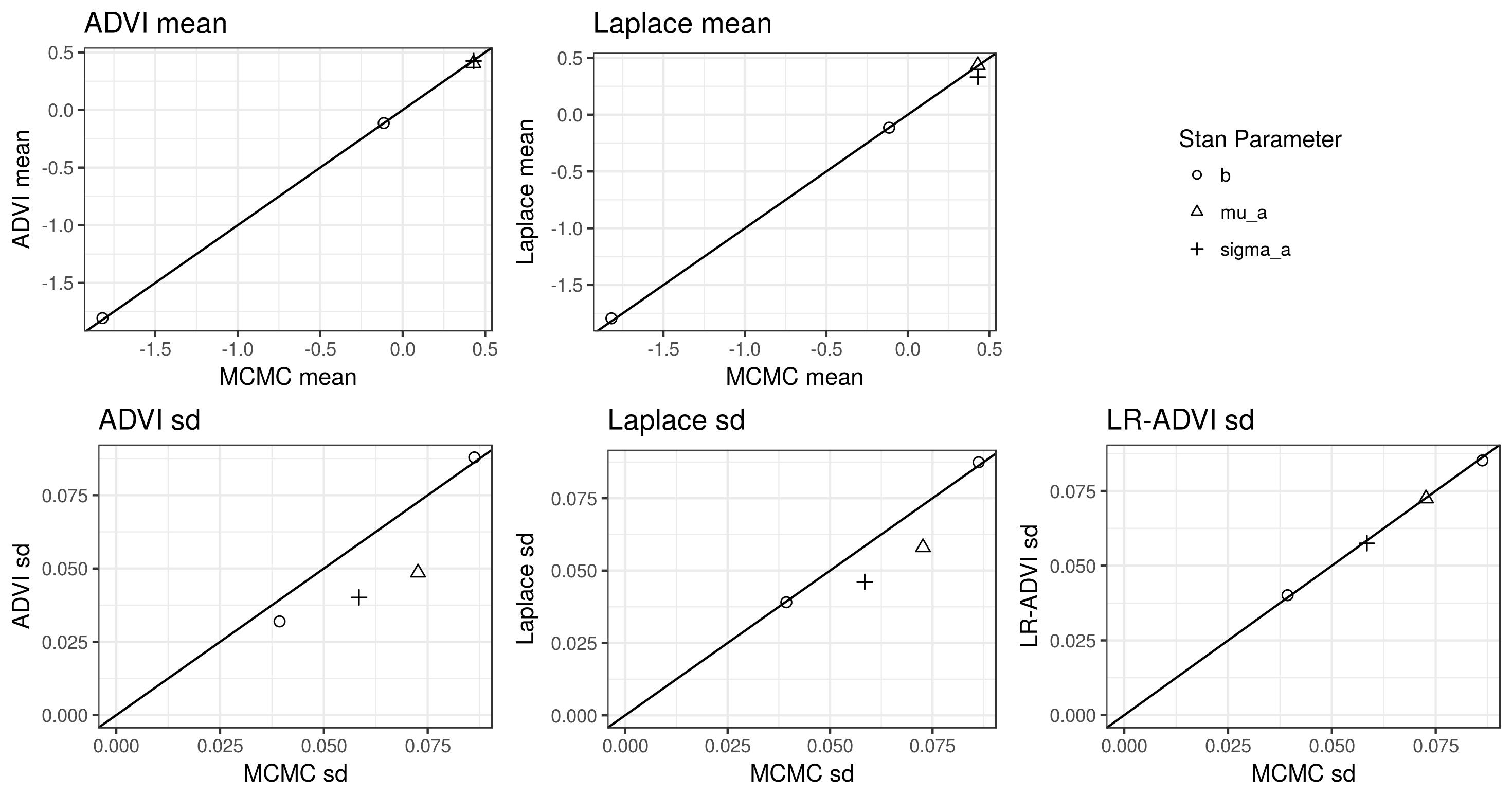} 

}

\caption[Election model]{Election model}\label{fig:ADVIElectionModel}
\end{figure}

\end{knitrout}

We begin with \texttt{election88}, which models binary responses in
a 1988 poll using a Bernoulli hierarchical model with normally distributed
random effects for state, ethnicity, and gender and a logit link.
The model and data are described in detail in \citet[Chapter 14]{gelman:2006:arm}.
\fig{ADVIElectionModel} shows that both the Laplace approximation
and ADVI do a reasonable job of matching to MCMC, though LR-ADVI is
slightly more accurate for standard deviations.

\subsubsection{Sesame Street Model Accuracy}

\begin{knitrout}
\definecolor{shadecolor}{rgb}{0.969, 0.969, 0.969}\color{fgcolor}\begin{figure}[t]

{\centering \includegraphics[width=0.98\linewidth,height=0.514\linewidth]{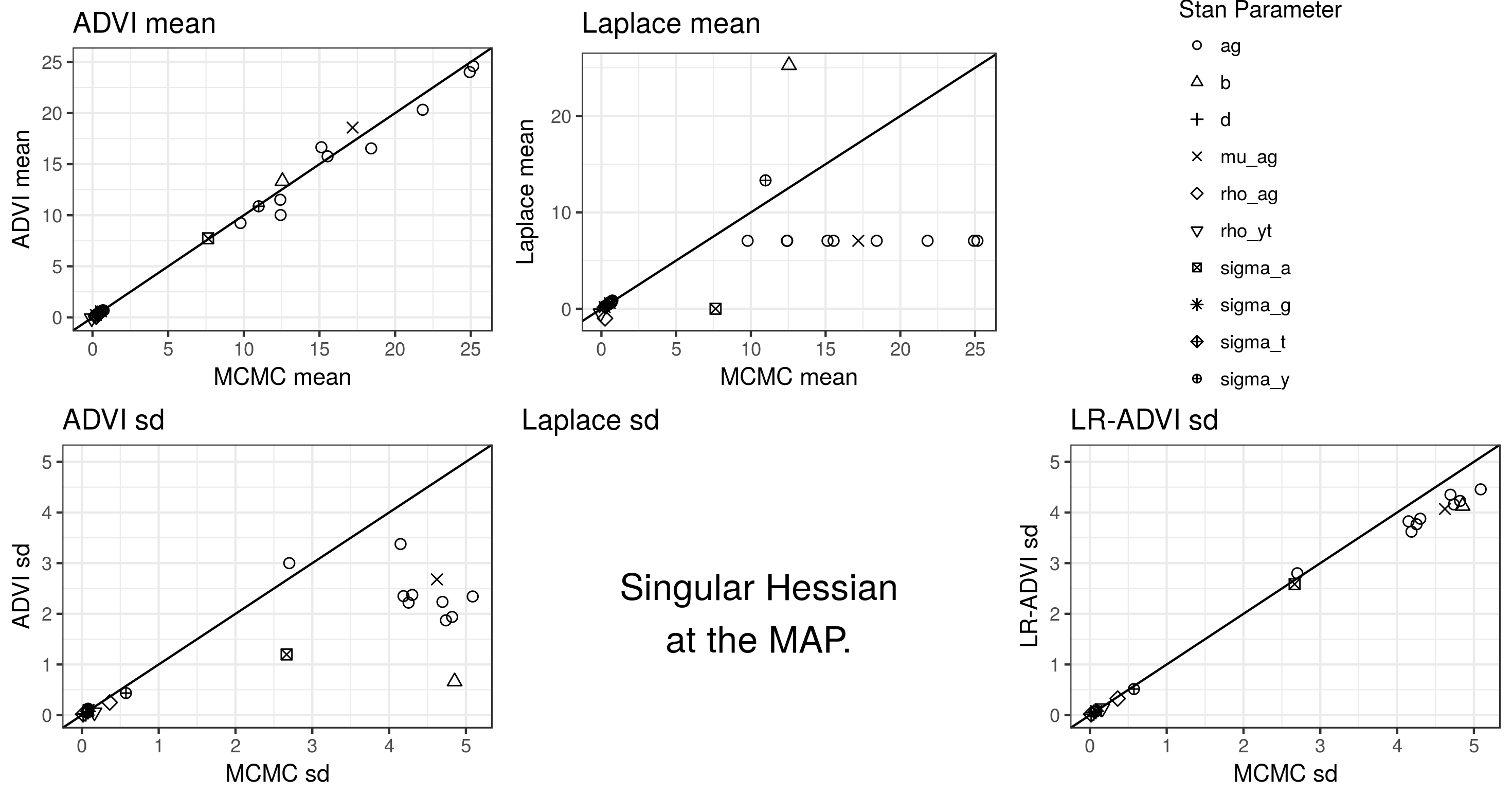} 

}

\caption[Sesame Street model]{Sesame Street model}\label{fig:ADVISesameStreetModel}
\end{figure}

\end{knitrout}

Next, we show results for \texttt{sesame\_street1}, an analysis of
a randomized controlled trial designed to estimate the causal effect
of watching the television show Sesame Street on a letter-recognition
test. To control for different conditions in the trials, a hierarchical
model is used with correlated multivariate outcomes and unknown covariance
structure. The model and data are described in detail in
\citet[Chapter 23]{gelman:2006:arm}.

As can be seen in \fig{ADVISesameStreetModel}, the MAP under-estimates the
variability of the random effects $ag$, and, in turn, under-estimates the
variance parameter $sigma\_a$. Because the MAP estimate of $sigma\_a$ is close
to zero, the log posterior has a very high curvature with respect to the
parameter $ag$ at the MAP, and the Hessian used for the Laplace approximation is
numerically singular. ADVI, which integrates out the uncertainty in the random
effects, provides reasonably good estimates of the posterior means but
underestimates the posterior standard deviations due to the mean-field
assumption.  Only LR-ADVI provides accurate estimates of posterior uncertainty.

\subsubsection{Radon Model Accuracy}

\begin{knitrout}
\definecolor{shadecolor}{rgb}{0.969, 0.969, 0.969}\color{fgcolor}\begin{figure}[t]

{\centering \includegraphics[width=0.98\linewidth,height=0.514\linewidth]{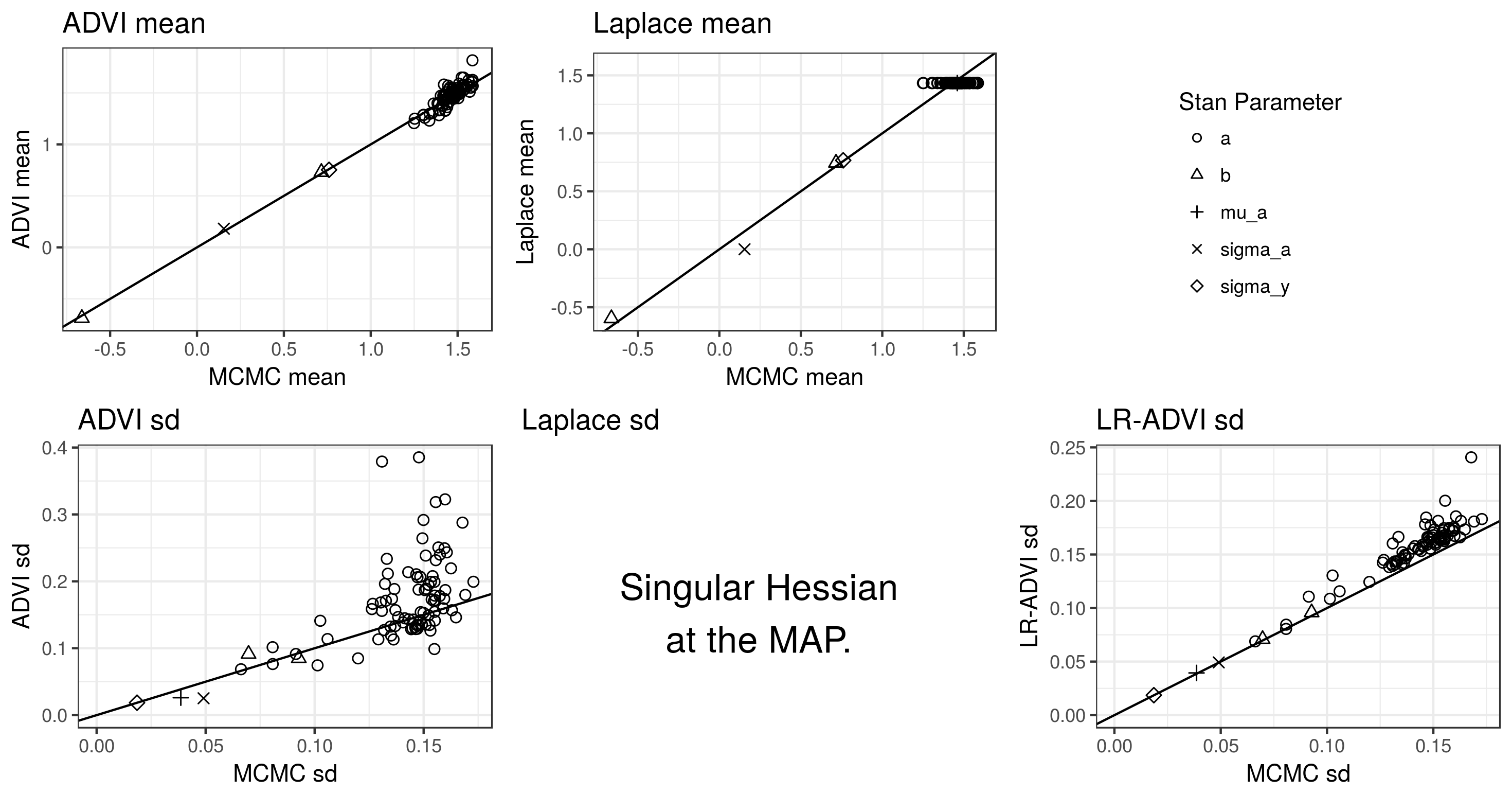} 

}

\caption[Radon model]{Radon model}\label{fig:ADVIRadonModel}
\end{figure}

\end{knitrout}

We now turn to \texttt{radon\_vary\_intercept\_floor}, a hierarchical model of
radon levels in Minnesota homes described in
\citet[Chapters 16 and 21]{gelman:2006:arm}.
This model is relatively simple, with univariate normal
observations and unknown variances. Nevertheless, the Laplace approximation
again produces a numerically singular covariance matrix.
The ADVI means are reasonably
accurate, but the standard deviations are not. Only LR-ADVI produces an
accurate approximation to the MCMC posterior standard deviations.

\subsubsection{Ecology Model Accuracy}

\begin{knitrout}
\definecolor{shadecolor}{rgb}{0.969, 0.969, 0.969}\color{fgcolor}\begin{figure}[t]

{\centering \includegraphics[width=0.98\linewidth,height=0.514\linewidth]{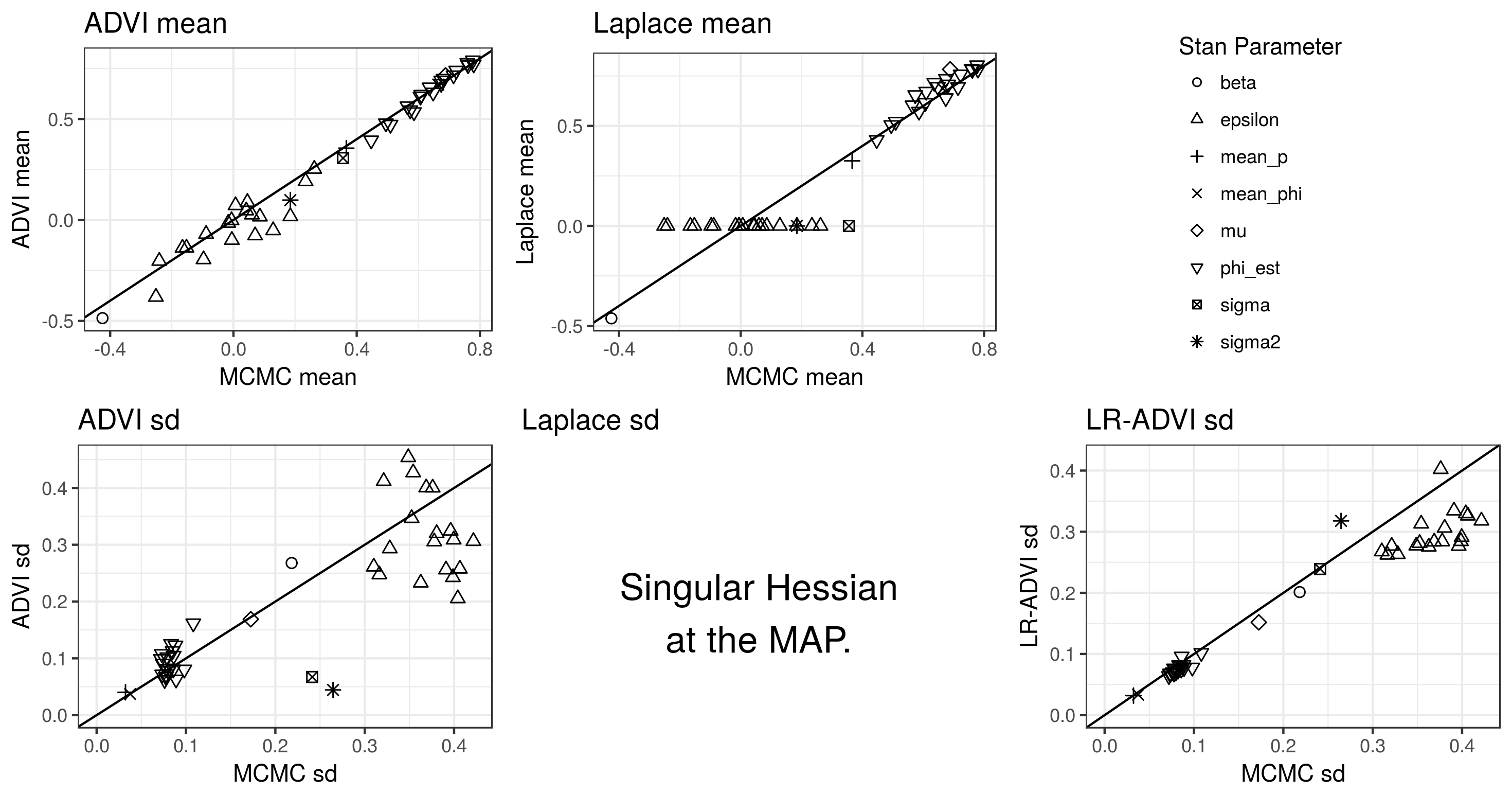} 

}

\caption[Ecology model]{Ecology model}\label{fig:ADVIEcologyModel}
\end{figure}

\end{knitrout}

Finally, we consider a more complicated mark-recapture model from
ecology known as the Cormack-Jolly-Seber (CJS) model. This model is
described in detail in \citet[Chapter 7]{kery:2011:bpa}, and discussion
of the Stan implementation can be found in \citet[Section 15.3]{stan-manual:2015}.

The Laplace approximation is again degenerate, and the ADVI standard
deviations again deviate considerably from MCMC. In this case, the
ADVI means are also somewhat inaccurate, and some of the LR-ADVI standard
deviations are mis-estimated in turn. However, LR-ADVI remains by
far the most accurate method for approximating the MCMC standard errors.

\subsubsection{Timing Results\label{subsec:advi_timing}}

\begin{knitrout}
\definecolor{shadecolor}{rgb}{0.969, 0.969, 0.969}\color{fgcolor}\begin{figure}[t]

{\centering \includegraphics[width=0.98\linewidth,height=0.514\linewidth]{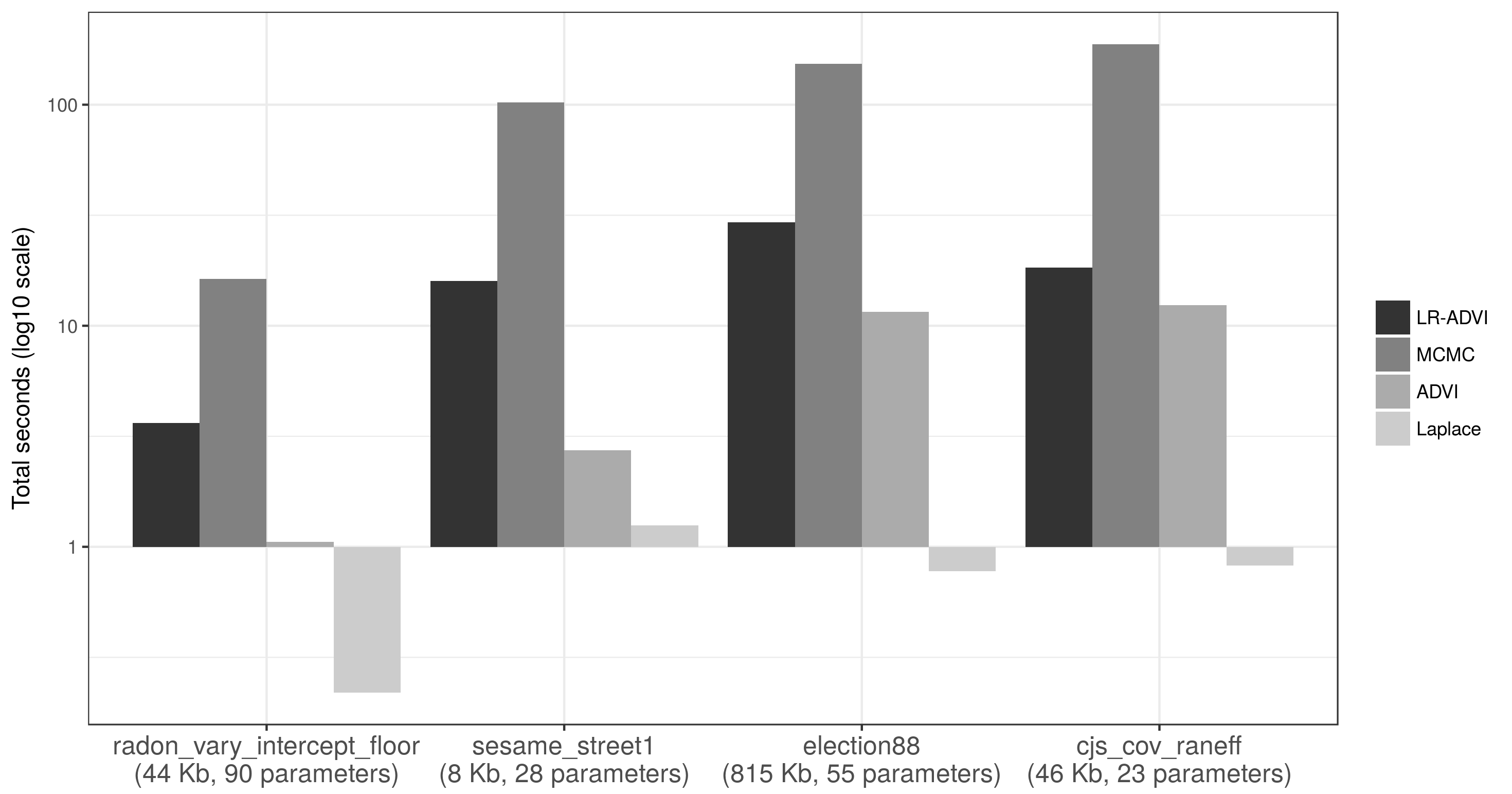} 

}

\caption[Comparision of timing in ADVI experiments]{Comparision of timing in ADVI experiments}\label{fig:ADVITimingGraph}
\end{figure}

\end{knitrout}

Detailed timing results for the ADVI experiments are shown in \fig{ADVITimingGraph}.
Both the Laplace approximation and ADVI alone are faster than LR-ADVI,
which in turn is about five times faster than MCMC. We achieved the best
results optimizing $\widehat{KL}\left(\eta\right)$ by using the conjugate
gradient Newton's trust region method (\texttt{trust-ncg} of \texttt{scipy.optimize}),
but the optimization procedure still accounted for an appreciable
proportion of the time needed for LR-ADVI.
 
\subsection{Criteo Dataset
\label{subsec:glmm_model}}

\newcommand{\glmmDimension}{5}
\newcommand{\glmmNumGroups}{5000}
\newcommand{\glmmNumObs}{61895}
\newcommand{\glmmHessDim}{10014}
\newcommand{\glmmInverseTime}{173}
\newcommand{\glmmHessianTime}{323}
\newcommand{\glmmVBTime}{57}
\newcommand{\glmmMCMCTime}{21066}
\newcommand{\glmmMCMCTimeMinutes}{351}
\newcommand{\glmmCGRowTime}{9.4}
\newcommand{\glmmCGRowIters}{81}
\newcommand{\glmmCGBetaTime}{46.9}
\newcommand{\glmmMAPTime}{12}
\newcommand{\glmmGLMERTime}{104}
\newcommand{\glmmVBRefitTime}{27.2}
\newcommand{\glmmSpeedup}{370}
\newcommand{\glmmNumMCMCDraws}{5000}
\newcommand{\glmmNumGHPoints}{4}
\newcommand{\glmmBetaInfoDiag}{0.100}
\newcommand{\glmmBetaLoc}{0.000}
\newcommand{\glmmMuLoc}{0.000}
\newcommand{\glmmMuInfo}{0.010}
\newcommand{\glmmTauAlpha}{3.000}
\newcommand{\glmmTauBeta}{3.000}

We now apply our methods to a real-world data set using a logistic
regression with random effects, which is an example of a generalized
linear mixed model (GLMM) \citep[Chapter 13]{agresti:2011:categorical}.
This data and model have several advantages as an illustration of
our methods: the data set is large, the model contains a large number
of imprecisely-estimated latent variables (the unknown random effects),
the model exhibits the sparsity of $\klhess$ that is typical in many
MFVB applications, and the results exhibit the same shortcomings of
the Laplace approximation seen above. For this model, we will evaluate
both posterior covariances and prior sensitivities.

\subsubsection{Data and Model}

We investigated a custom subsample of the 2014 Criteo Labs conversion
logs data set \citep{criteo:2014:dataset}, which contains an obfuscated
sample of advertising data collected by Criteo over a period of two
months. Each row of the data set corresponds to a single user click
on an online advertisement. For each click, the data set records a
binary outcome variable representing whether or not the user subsequently
``converted'' (i.e., performed a desired task, such as purchasing
a product or signing up for a mailing list). Each row contains two timestamps
(which we ignore), eight numerical covariates, and nine factor-valued
covariates. Of the eight numerical covariates, three contain 30\%
or more missing data, so we discarded them. We then applied a per-covariate
normalizing transform to the distinct values of those remaining. Among
the factor-valued covariates, we retained only the one with the largest
number of unique values and discarded the others. These data-cleaning
decisions were made for convenience. The goal of the present paper
is to demonstrate our inference methods, not to draw conclusions about
online advertising.

Although the meaning of the covariates has been obfuscated, for the
purpose of discussion we will imagine that the single retained factor-valued
covariate represents the identity of the advertiser, and the numeric
covariates represent salient features of the user and/or the advertiser
(e.g., how often the user has clicked or converted in the past, a
machine learning rating for the advertisement quality, etc.). As such,
it makes sense to model the probability of each row's binary outcome
(whether or not the user converted) as a function of the five numeric
covariates and the advertiser identity using a logistic GLMM. Specifically,
we observe binary conversion outcomes, $y_{it}$, for click $i$ on
advertiser $t$, with probabilities given by observed numerical explanatory
variables, $x_{it}$, each of which are vectors of length $K_{x}=\glmmDimension$.
Additionally, the outcomes within a given value of $t$ are correlated
through an unobserved random effect, $u_{t}$, which represents the
``quality'' of advertiser $t$, where the value of $t$ for each
observation is given by the factor-valued covariate. The random effects
$u_{t}$ are assumed to follow a normal distribution with unknown
mean and variance. Formally,
\begin{eqnarray*}
y_{it}\vert p_{it} & \sim &
    \textrm{Bernoulli}\left(p_{it}\right),\textrm{ for }t=1,...,T
    \textrm{ and }i=1,...,N_t\\
p_{it} & := & \frac{e^{\rho_{it}}}{1+e^{\rho_{it}}}
    \quad\textrm{where}\quad\rho_{it}:=x_{it}^{T}\beta+u_{t}\\
u_{t}\vert\mu,\tau & \sim & \mathcal{N}\left(\mu,\tau^{-1}\right).
\end{eqnarray*}
Consequently, the unknown parameters are
$\theta=\left(\beta^{\trans},\mu,\tau,u_{1},...,u_{T}\right)^{\trans}$.
We use the following priors:

\begin{eqnarray*}
\mu\vert\mu_{0},\tau_{\mu} & \sim & \mathcal{N}\left(\mu_{0},\tau_{\mu}^{-1}\right)\\
\tau\vert\alpha_{\tau},\beta_{\tau} & \sim & \textrm{Gamma}\left(\alpha_{\tau},\beta_{\tau}\right)\\
\beta\vert\beta_{0},\tau_{\beta},\gamma_{\beta} & \sim & \mathcal{N}\left(\left(\begin{array}{c}
\beta_{0}\\
\vdots\\
\beta_{0}
\end{array}\right),\left(\begin{array}{ccc}
\tau_{\beta} & \gamma_{\beta} & \gamma_{\beta}\\
\gamma_{\beta} & \ddots & \gamma_{\beta}\\
\gamma_{\beta} & \gamma_{\beta} & \tau_{\beta}
\end{array}\right)^{-1}\right).
\end{eqnarray*}
Note that we initially take $\gamma_{\beta}=0$ so that the prior
information matrix on $\beta$ is diagonal. Nevertheless, by retaining $\gamma_{\beta}$
as a hyperparameter we will be able to assess the sensitivity to the
assumption of a diagonal prior in \prettyref{subsec:glmm_sensitivity}.
The remaining prior values are given in \prettyref{app:glmm_details}.
It is reasonable to expect that a modeler would be interested both
in the effect of the numerical covariates and in the quality of individual
advertisers themselves, so we take the parameter of interest to be
$\gtheta=\left(\beta^{\trans},u_{1},...,u_{T}\right)^{\trans}$.

To produce a data set small enough to be amenable to MCMC but large
and sparse enough to demonstrate our methods, we subsampled the data
still further. We randomly chose $\glmmNumGroups$ distinct advertisers
to analyze, and then subsampled each selected advertiser to contain
no more than 20 rows each. The resulting data set had $N=\glmmNumObs$
total rows. If we had more observations per advertiser, the ``random
effects'' $u_{t}$ would have been estimated quite precisely, and
the nonlinear nature of the problem would not have been important; these changes would thus have
obscured the benefits of using MFVB versus the Laplace approximation.
In typical internet data sets a large amount of data comes from advertisers
with few observations each, so our subsample is representative of
practically interesting problems.

\subsubsection{Inference and Timing\label{subsec:glmm_inference}}

\begin{table}[h]
\begin{center}\begin{tabular}{lr}
  \hline
Method & Seconds \\ 
  \hline
MAP (optimum only) & 12 \\ 
  VB (optimum only) & 57 \\ 
  VB (including sensitivity for $\beta$) & 104 \\ 
  VB (including sensitivity for $\beta$ and $u$) & 553 \\ 
  MCMC (Stan) & 21066 \\ 
   \hline
\end{tabular}
\end{center}
\caption{Timing results\label{tab:timing}}
\end{table}

We estimated the expectation and covariance of $\gtheta$ using four
techniques: MCMC, the Laplace approximation, MFVB, and linear response
(LRVB) methods. For MCMC, we used Stan \citep{stan-manual:2015},
and to calculate the MFVB, Laplace, and LRVB estimates we used our
own Python code using \texttt{numpy}, \texttt{scipy}, and \texttt{autograd}
\citep{scipy,maclaurin:2015:autograd}. As described in \prettyref{subsec:glmm_Means-and-variances},
the MAP estimator did not estimate $\mbe_{\pzeropost}\left[\gtheta\right]$
very well, so we do not report standard deviations or sensitivity
measures for the Laplace approximations. The summary of the computation
time for all these methods is shown in \prettyref{tab:timing}, with
details below.

For the MCMC estimates, we used Stan to draw $\glmmNumMCMCDraws$
MCMC draws (not including warm-up), which took $\glmmMCMCTimeMinutes$
minutes. We estimated all the prior sensitivities of \prettyref{subsec:glmm_sensitivity}
using the Monte Carlo version of the covariance in \prettyref{eq:covariance_sensitivity}.

For the MFVB approximation, we use the following mean field exponential
family approximations:
\begin{align*}
q\left(\beta_{k}\right) & =\normal\left(\beta_{k};\eta_{\beta_{k}}\right),\textrm{ for }k=1,...,K_{x}\\
q\left(u_{t}\right) & =\normal\left(u_{t};\eta_{u_{t}}\right),\textrm{ for }t=1,...,T\\
q\left(\tau\right) & =\textrm{Gamma}\left(\tau;\eta_{\tau}\right)\\
q\left(\mu\right) & =\normal\left(\mu;\eta_{\mu}\right)\\
q\left(\theta\right) & =
    q\left(\tau\right)q\left(\mu\right)
    \prod_{k=1}^{K_{x}}q\left(\beta_{k}\right)\prod_{t=1}^{T}q\left(u_{t}\right).
\end{align*}
With these choices, evaluating the variational objective requires
the following intractable univariate variational expectation:
\begin{eqnarray*}
\mbe_{q\left(\theta;\eta\right)}\left[\log\left(1-p_{it}\right)\right] & = &
    \mbe_{q\left(\theta;\eta\right)}\left[\log\left(1
        -\frac{e^{\rho_{it}}}{1+e^{\rho_{it}}}\right)\right].
\end{eqnarray*}
We used the re-parameterization trick and four points of Gauss-Hermite
quadrature to estimate this integral for each observation. See
\prettyref{app:glmm_details} for more details.

We optimized the variational objective using the conjugate gradient
Newton's trust region method, \texttt{trust-ncg}, of \texttt{scipy.optimize}.
One advantage of \texttt{trust-ncg} is that it performs second-order
optimization but requires only Hessian-vector products, which can
be computed quickly by \texttt{autograd} without constructing the
full Hessian. The MFVB fit took $\glmmVBTime$ seconds, roughly $\glmmSpeedup$
times faster than MCMC with Stan.

With variational parameters for each random effect $u_{t}$, $\klhess$
is a $\glmmHessDim\times\glmmHessDim$ dimensional matrix. Consequently,
evaluating $\klhess$ directly as a dense matrix using \texttt{autograd}
would have been prohibitively time-consuming. Fortunately, our model
can be decomposed into global and local parameters, and the Hessian
term $\klhess$ in \prettyref{thm:lrvb_formula} is extremely sparse.
In the notation of \prettyref{subsec:lrvb_implementation}, take
$\theta_{glob}=\left(\beta^{\trans},\mu,\tau\right)^{\trans},$
take $\theta_{loc,t}=u_{t}$, and stack the variational parameters
as $\eta=\left(\eta_{glob}^{\trans},\eta_{loc,1},...,\eta_{loc,T}\right)^{\trans}$.
The cross terms in $\klhess$ between the local variables vanish:
\begin{align*}
\frac{\partial^{2}KL\left(q\left(\theta;\eta\right)||
    \pthetapost\left(\theta\right)\right)}{\partial\eta_{loc,t_{1}}
    \partial\eta_{loc,t_{2}}} & =0\textrm{ for all }t_{1}\ne t_{2}.
\end{align*}
Equivalently, note that the full likelihood in \prettyref{app:glmm_details},
\prettyref{eq:glmm_log_lik}, has no cross terms between $u_{t_{1}}$
and $u_{t_{2}}$ for $t_{1}\ne t_{2}$. As the dimension $T$ of the
data grows, so does the length of $\eta$. However, the dimension
of $\eta_{glob}$ remains constant, and $\klhess$ remains easy to
invert. We show an example of the sparsity pattern of the first few
rows and columns of $\klhess$ in \fig{LogitGLMMHessianSparsity}
.

\begin{figure}[h]
\centering{}\includegraphics[width=2in]{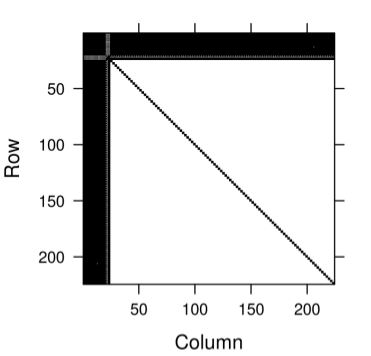}
\caption{Sparsity pattern of top-left sub-matrix of $\protect\klhess$
for the logit GLMM model.
The axis numbers represent indices within $\eta$,
and black indicates non-zero entries of $\protect\klhess$.}
\label{fig:LogitGLMMHessianSparsity}
\end{figure}

Taking advantage of this sparsity pattern, we used \texttt{autograd}
to calculate the Hessian of the KL divergence one group at a time
and assembled the results in a sparse matrix using the \texttt{scipy.sparse}
Python package. Even so, calculating the entire sparse Hessian took
$\glmmHessianTime$ seconds, and solving the system $\klhess^{-1}\eggrad^{\trans}$
using \texttt{scipy.sparse.linalg.spsolve} took an additional $\glmmInverseTime$
seconds. These results show that the evaluation and inversion of $\klhess$
was several times more costly than optimizing the variational objective
itself. (Of course, the whole procedure remains much faster than running
MCMC with Stan.)

We note, however, that instead of the direct approach to calculating
$\klhess^{-1}\eggrad^{\trans}$ one can use the conjugate gradient algorithm of
\texttt{sp.sparse.linalg.cg} \citep[Chapter 5]{nocedalwright:1999:numerical}
together with the fast Hessian-vector products of \texttt{autograd} to query one
column at a time of $\klhess^{-1}\eggrad^{\trans}$. On a typical column of
$\klhess^{-1}\eggrad^{\trans}$ in our experiment, calculating the conjugate
gradient took only $\glmmCGRowTime$ seconds (corresponding to $\glmmCGRowIters$
Hessian-vector products in the conjugate gradient algorithm). Thus, for example,
one could calculate the columns of $\klhess^{-1}\eggrad^{\trans}$ corresponding
to the expectations of the global variables $\beta$ in only
$\glmmCGRowTime\times K_{x}=\glmmCGBetaTime$ seconds, which is much less time
than it would take to compute the entire $\klhess^{-1}\eggrad^{\trans}$ for both
$\beta$ and every random effect in $u$.

For the Laplace approximation, we calculated the MAP estimator and
$\laphess$ using Python code similar to that used for the MFVB estimates.
We observe that the MFVB approximation to posterior means would
be expected to improve on the MAP estimator only in cases when there is
both substantial uncertainty in some parameters and when this uncertainty,
through nonlinear dependence between parameters, affects the values
of posterior means. These circumstances obtain in the logistic GLMM
model with sparse per-advertiser data since the random effects $u_{t}$
will be quite uncertain and the other posterior means depend on them
through the nonlinear logistic function.

\subsubsection{Posterior Approximation Results\label{subsec:glmm_Means-and-variances}}

In this section, we assess the accuracy of the MFVB, Laplace, and LRVB methods
as approximations to $\mbe_{\pzeropost}\left[\gtheta\right]$ and
$\pcov\left(\gtheta\right)$. We take the MCMC estimates as ground truth.
Although, as discussed in \prettyref{subsec:glmm_model}, we are principally
interested in the parameters
$\gtheta=\left(\beta^{\trans},u_{1},...,u_{T}\right)^{\trans}$, we will report
the results for all parameters for completeness. For readability, the tables and
graphs show results for a random selection of the components of the random
effects $u$.

\subsubsection{Posterior Means}

\begin{table}
\begin{center}\begin{tabular}{lrrrrr}
  \hline
Parameter & MCMC & MFVB & MAP & MCMC std. err. & Eff. \# of MCMC draws \\ 
  \hline
$\beta_{1}$ & 1.454 & 1.447 & 1.899 & 0.02067 & 33 \\ 
  $\beta_{2}$ & 0.031 & 0.033 & 0.198 & 0.00025 & 5000 \\ 
  $\beta_{3}$ & 0.110 & 0.110 & 0.103 & 0.00028 & 5000 \\ 
  $\beta_{4}$ & -0.172 & -0.173 & -0.173 & 0.00016 & 5000 \\ 
  $\beta_{5}$ & 0.273 & 0.273 & 0.280 & 0.00042 & 5000 \\ 
  $\mu$ & 2.041 & 2.041 & 3.701 & 0.04208 & 28 \\ 
  $\tau$ & 0.892 & 0.823 & 827.724 & 0.00051 & 1232 \\ 
  $u_{1431}$ & 1.752 & 1.757 & 3.700 & 0.00937 & 5000 \\ 
  $u_{4150}$ & 1.217 & 1.240 & 3.699 & 0.01022 & 5000 \\ 
  $u_{4575}$ & 2.427 & 2.413 & 3.702 & 0.00936 & 5000 \\ 
  $u_{4685}$ & 3.650 & 3.633 & 3.706 & 0.00862 & 5000 \\ 
   \hline
\end{tabular}
\end{center}
\caption{Results for the estimation of the posterior means\label{tab:mean_results}}
\end{table}

\begin{knitrout}
\definecolor{shadecolor}{rgb}{0.969, 0.969, 0.969}\color{fgcolor}\begin{figure}[t]

{\centering \includegraphics[width=0.98\linewidth,height=0.343\linewidth]{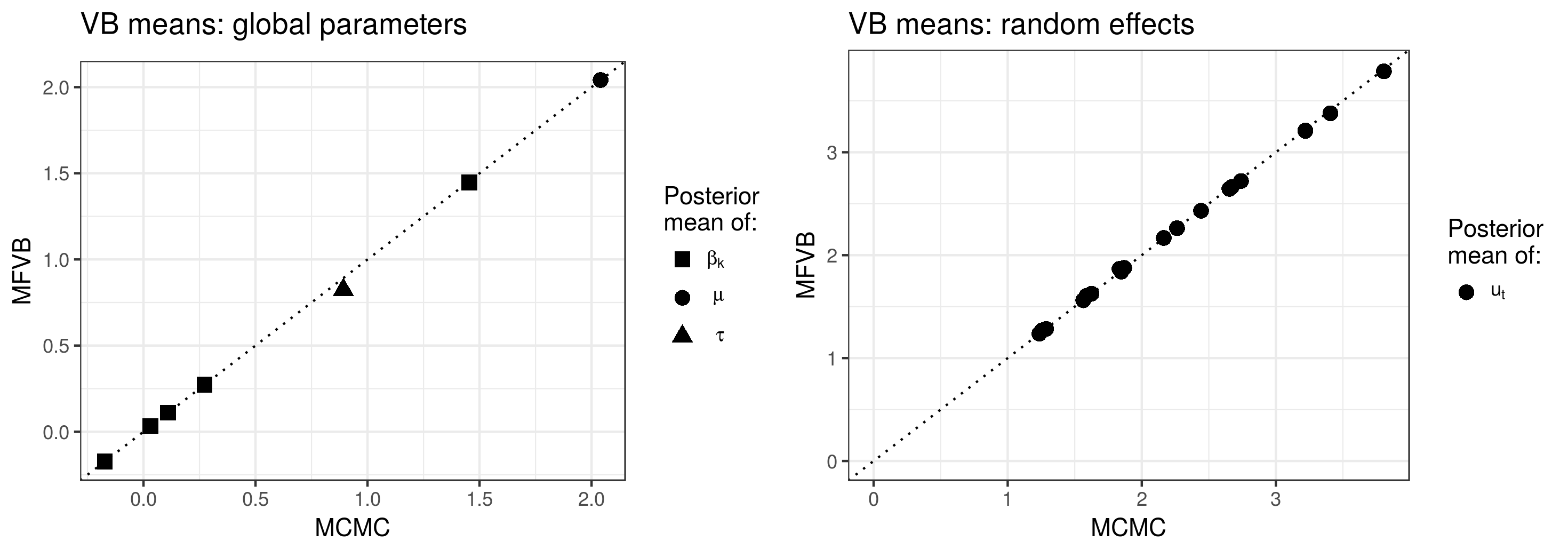} 

}

\caption[Comparison of MCMC and MFVB means]{Comparison of MCMC and MFVB means}\label{fig:LogitGLMMMCMCComparisonMeans}
\end{figure}

\end{knitrout}

\begin{knitrout}
\definecolor{shadecolor}{rgb}{0.969, 0.969, 0.969}\color{fgcolor}\begin{figure}[t]

{\centering \includegraphics[width=0.98\linewidth,height=0.343\linewidth]{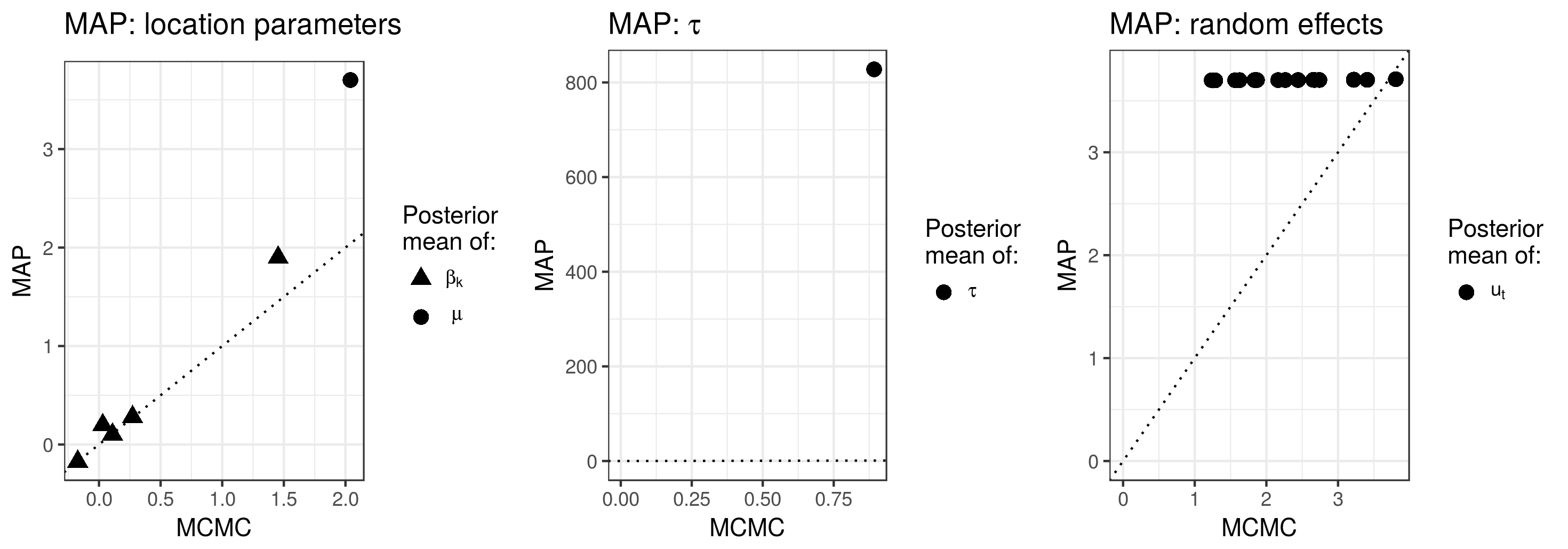} 

}

\caption[Comparison of MCMC and Laplace means]{Comparison of MCMC and Laplace means}\label{fig:LogitGLMMMapComparisonMeans}
\end{figure}

\end{knitrout}

We begin by comparing the posterior means in \prettyref{tab:mean_results},
\fig{LogitGLMMMCMCComparisonMeans}, and \fig{LogitGLMMMapComparisonMeans}.
We first note that, despite the long running time for MCMC, the $\beta_{1}$
and $\mu$ parameters did not mix well in the MCMC sample, as is reflected
in the MCMC standard error and effective number of draws columns of
\prettyref{tab:mean_results}. The $x_{it}$ data corresponding to
$\beta_{1}$ contained fewer distinct values than the other columns
of $x$, which perhaps led to some co-linearity between $\beta_{1}$
and $\mu$ in the posterior. This co-linearity could have caused both poor MCMC
mixing and, perhaps, excessive measured prior sensitivity, as discussed
below in \prettyref{subsec:glmm_sensitivity}. Although we will report
the results for both $\beta_{1}$ and $\mu$ without further comment,
the reader should bear in mind that the MCMC ``ground truth'' for
these two parameters is somewhat suspect.

The results in \prettyref{tab:mean_results} and \fig{LogitGLMMMCMCComparisonMeans}
show that MFVB does an excellent job of approximating the posterior
means in this particular case, even for the random effects $u$ and
the related parameters $\mu$ and $\tau$. In contrast, the MAP estimator
does reasonably well only for certain components of $\beta$ and does
extremely poorly for the random effects parameters. As can be seen
in \fig{LogitGLMMMapComparisonMeans}, the MAP estimate dramatically
overestimates the information $\tau$ of the random effect distribution
(that is, it underestimates the variance). As a consequence, it estimates
all the random effects to have essentially the same value, leading
to mis-estimation of some location parameters, including both $\mu$
and some components of $\beta$. Since the MAP estimator performed
so poorly at estimating the random effect means, we will not consider
it any further.

\subsubsection{Posterior Covariances}

\begin{table}
\begin{center}\begin{tabular}{lrrr}
  \hline
Parameter & MCMC & LRVB & Uncorrected MFVB \\ 
  \hline
$\beta_{1}$ & 0.118 & 0.103 & 0.005 \\ 
  $\beta_{2}$ & 0.018 & 0.018 & 0.004 \\ 
  $\beta_{3}$ & 0.020 & 0.020 & 0.004 \\ 
  $\beta_{4}$ & 0.012 & 0.012 & 0.004 \\ 
  $\beta_{5}$ & 0.029 & 0.030 & 0.004 \\ 
  $\mu$ & 0.223 & 0.192 & 0.016 \\ 
  $\tau$ & 0.018 & 0.033 & 0.016 \\ 
  $u_{1431}$ & 0.663 & 0.649 & 0.605 \\ 
  $u_{4150}$ & 0.723 & 0.707 & 0.662 \\ 
  $u_{4575}$ & 0.662 & 0.649 & 0.615 \\ 
  $u_{4685}$ & 0.610 & 0.607 & 0.579 \\ 
   \hline
\end{tabular}
\end{center}
\caption{Standard deviation results\label{tab:sd_results}}
\end{table}

\begin{knitrout}
\definecolor{shadecolor}{rgb}{0.969, 0.969, 0.969}\color{fgcolor}\begin{figure}[t]

{\centering \includegraphics[width=0.98\linewidth,height=0.514\linewidth]{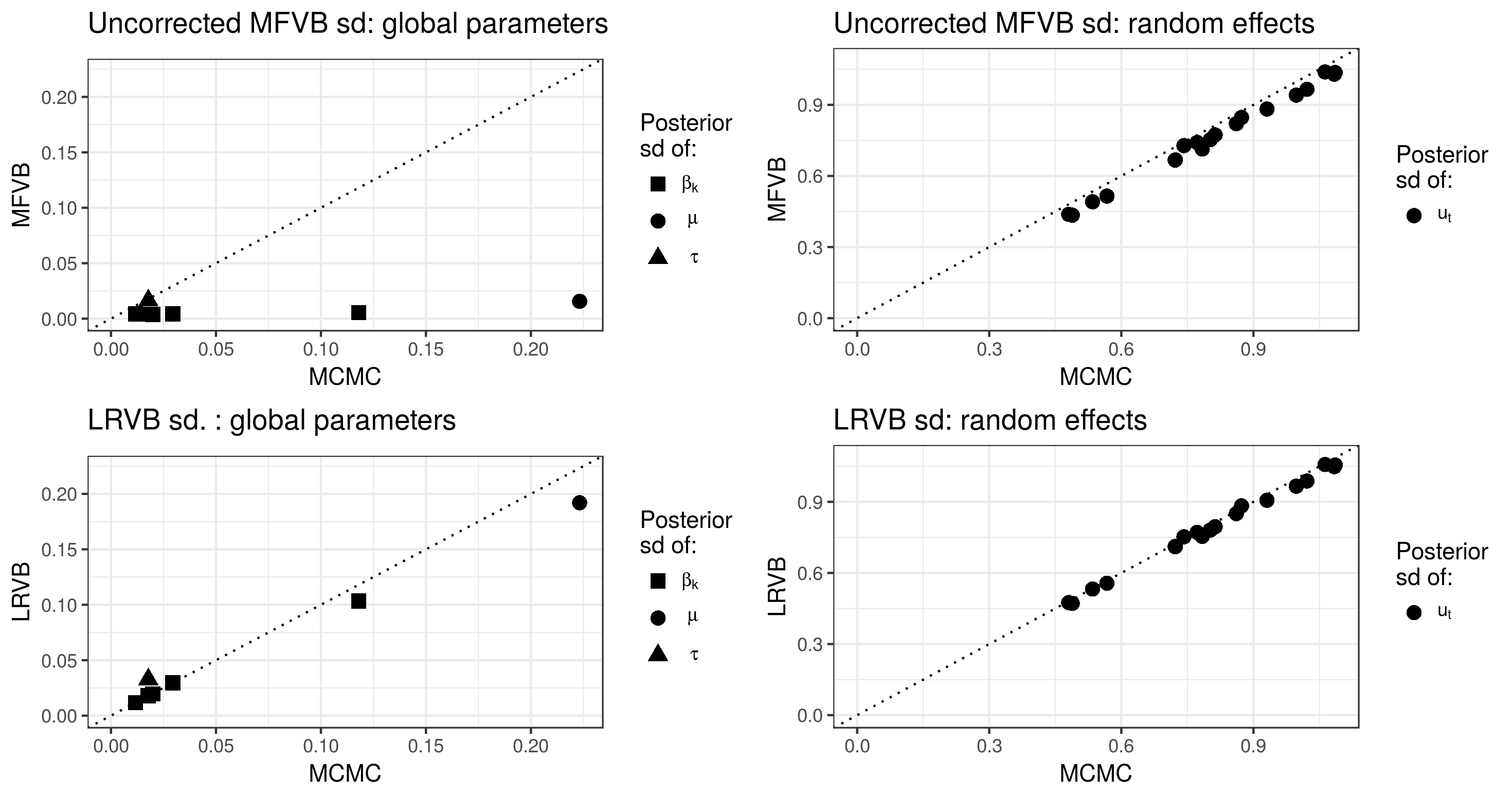} 

}

\caption[Comparison of MCMC, MFVB, and LRVB standard deviations]{Comparison of MCMC, MFVB, and LRVB standard deviations}\label{fig:LogitGLMMMCMCComparisonSds}
\end{figure}

\end{knitrout}

\begin{knitrout}
\definecolor{shadecolor}{rgb}{0.969, 0.969, 0.969}\color{fgcolor}\begin{figure}[t]

{\centering \includegraphics[width=0.98\linewidth,height=0.343\linewidth]{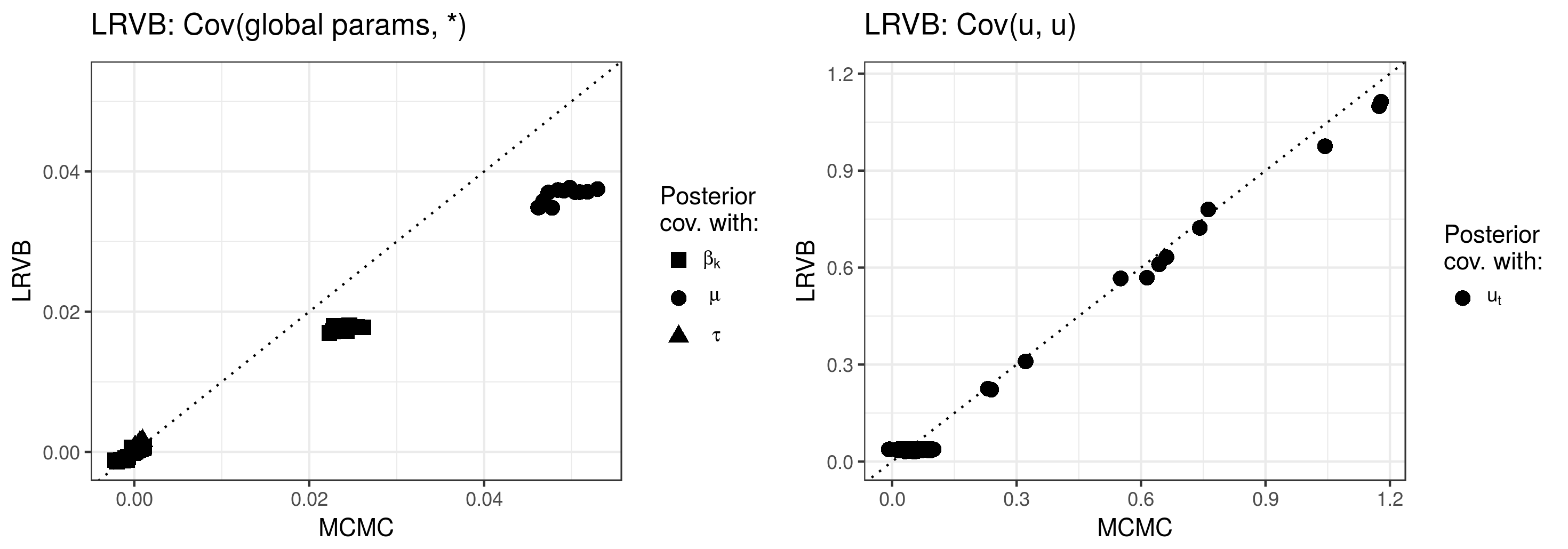} 

}

\caption[Comparison of MCMC and LRVB off-diagonal covariances]{Comparison of MCMC and LRVB off-diagonal covariances}\label{fig:LogitGLMMMCMCComparisonCovariances}
\end{figure}

\end{knitrout}
We now assess the accuracy of our estimates of $\pcov\left(\gtheta\right)$.
The results for the marginal standard deviations are shown in \prettyref{tab:sd_results}
and \fig{LogitGLMMMCMCComparisonSds}. We refer to the standard deviations
of $\qcov\left(\gtheta\right)$ as the ``uncorrected MFVB'' estimate,
and of $\lrvbcov\left(\gtheta\right)$ as the ``LRVB'' estimate.
The uncorrected MFVB variance estimates of $\beta$ are particularly
inaccurate, but the LRVB variances match the exact posterior closely.

In \fig{LogitGLMMMCMCComparisonCovariances}, we compare the off-diagonal
elements of $\lrvbcov\left(\gtheta\right)$ and $\pcov\left(\gtheta\right)$.
These covariances are zero, by definition, in the uncorrected MFVB
estimates $\qcov\left(\gtheta\right)$. The left panel of \fig{LogitGLMMMCMCComparisonCovariances}
shows the estimated covariances between the global parameters and
all other parameters, including the random effects, and the right
panel shows only the covariances amongst the random effects. The LRVB
covariances are quite accurate, particularly when we recall that the MCMC
draws of $\mu$ may be inaccurate due to poor mixing.

\subsubsection{Parametric Sensitivity Results\label{subsec:glmm_sensitivity}}

\begin{table}
\begin{center}\begin{tabular}{lrrrrrrr}
  \hline
  & $\beta_{0}$ & $\tau_{\beta}$ & $\gamma_{\beta}$ & $\mu_0$ & $\tau_{\mu}$ & $\alpha_{\tau}$ & $\beta{\tau}$ \\ 
  \hline
$\mu$ & 0.0094 & -0.1333 & -0.0510 & 0.0019 & -0.3920 & 0.0058 & -0.0048 \\ 
  $\tau$ & 0.0009 & -0.0086 & -0.0142 & 0.0003 & -0.0575 & 0.0398 & -0.0328 \\ 
  $\beta_{1}$ & 0.0089 & -0.1464 & -0.0095 & 0.0017 & -0.3503 & 0.0022 & -0.0018 \\ 
  $\beta_{2}$ & 0.0012 & -0.0143 & -0.0113 & 0.0003 & -0.0516 & 0.0062 & -0.0051 \\ 
  $\beta_{3}$ & -0.0035 & 0.0627 & -0.0081 & -0.0006 & 0.1218 & -0.0003 & 0.0002 \\ 
  $\beta_{4}$ & 0.0018 & -0.0037 & -0.0540 & 0.0004 & -0.0835 & 0.0002 & -0.0002 \\ 
  $\beta_{5}$ & 0.0002 & 0.0308 & -0.0695 & 0.0002 & -0.0383 & 0.0011 & -0.0009 \\ 
  $u_{1431}$ & 0.0028 & -0.0397 & -0.0159 & 0.0006 & -0.1169 & 0.0018 & -0.0015 \\ 
  $u_{4150}$ & 0.0026 & -0.0368 & -0.0146 & 0.0005 & -0.1083 & 0.0022 & -0.0018 \\ 
  $u_{4575}$ & 0.0028 & -0.0406 & -0.0138 & 0.0006 & -0.1153 & 0.0011 & -0.0009 \\ 
  $u_{4685}$ & 0.0028 & -0.0409 & -0.0142 & 0.0006 & -0.1163 & 0.0003 & -0.0002 \\ 
   \hline
\end{tabular}
\end{center}
\caption{MFVB normalized prior sensitivity results\label{tab:prior_sens}}
\end{table}

\begin{knitrout}
\definecolor{shadecolor}{rgb}{0.969, 0.969, 0.969}\color{fgcolor}\begin{figure}[t]

{\centering \includegraphics[width=0.98\linewidth,height=0.343\linewidth]{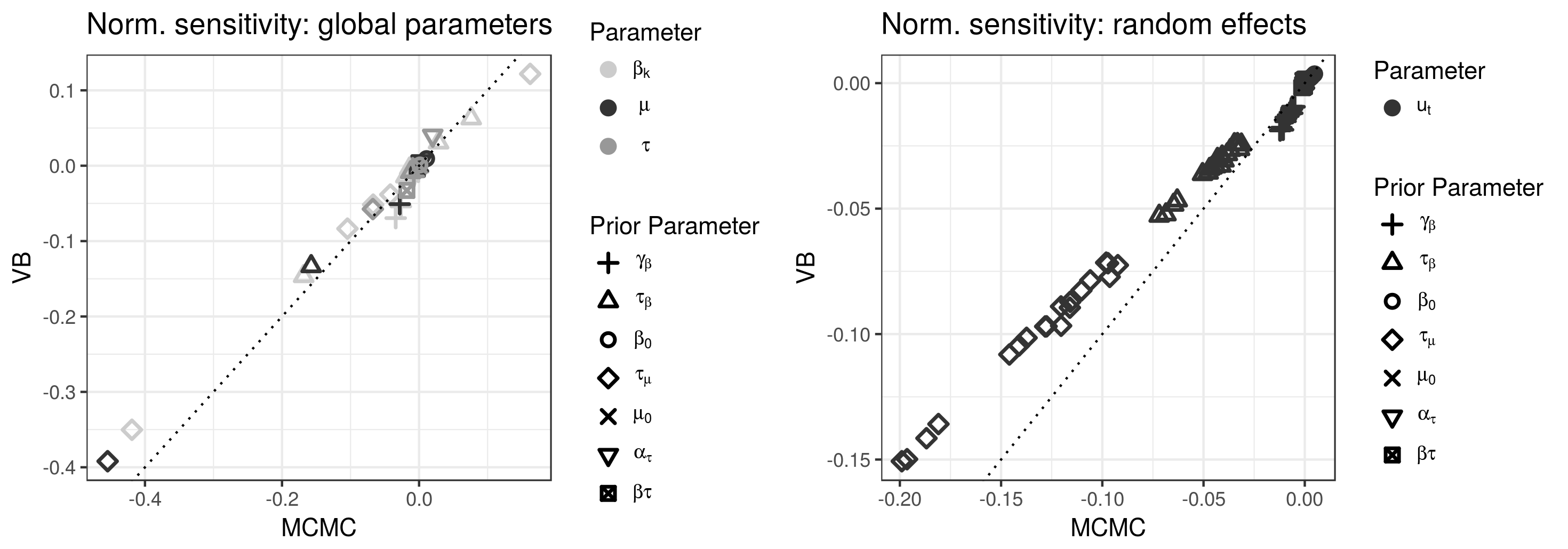} 

}

\caption[Comparison of MCMC and MFVB normalized parametric sensitivity results]{Comparison of MCMC and MFVB normalized parametric sensitivity results}\label{fig:LogitGLMMParametricRobustness}
\end{figure}

\end{knitrout}

\begin{knitrout}
\definecolor{shadecolor}{rgb}{0.969, 0.969, 0.969}\color{fgcolor}\begin{figure}[t]

{\centering \includegraphics[width=0.98\linewidth,height=0.343\linewidth]{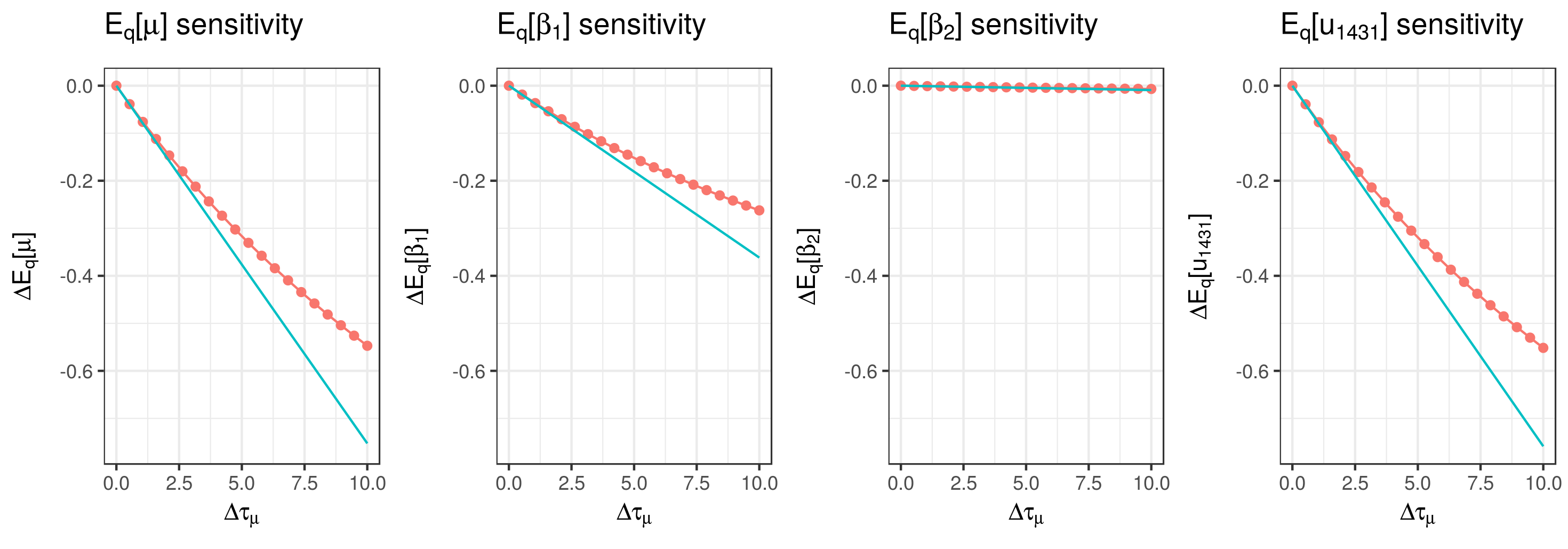} 

}

\caption[MFVB sensitivity as measured both by linear approximation (blue) and re-fitting (red)]{MFVB sensitivity as measured both by linear approximation (blue) and re-fitting (red)}\label{fig:LogitGLMMRefit}
\end{figure}

\end{knitrout}
Finally, we compare the MFVB prior sensitivity measures of \prettyref{subsec:lrvb_robustness}
to the covariance-based MCMC sensitivity measures of \prettyref{subsec:local_sensitivity}.
Since sensitivity is of practical interest only when it is of comparable
order to the posterior uncertainty, we report sensitivities normalized
by the appropriate standard deviation. That is, we report
$\psenshat/\sqrt{\textrm{diag}\left(\hat{\cov}_{\pzeropost}\left(\gtheta\right)\right)}$,
and $\qsens/\sqrt{\textrm{diag}\left(\lrvbcov\left(\gtheta\right)\right)}$,
etc., where $\textrm{diag}\left(\cdot\right)$ denotes the diagonal
vector of a matrix, and the division is element-wise. Note that we
use the sensitivity-based variance estimates $\lrvbcov$, not the
uncorrected MFVB estimates $\qcov$, to normalize the variational
sensitivities. We refer to a sensitivity divided by a standard deviation
as a ``normalized'' sensitivity.

The comparison between the MCMC and MFVB sensitivity measures is shown in
\fig{LogitGLMMParametricRobustness}. The MFVB and MCMC sensitivities correspond
very closely, though the MFVB means appear to be slightly more sensitive to the
prior parameters than the MCMC means. This close correspondence should not be
surprising. As shown in \prettyref{subsec:glmm_Means-and-variances}, the MFVB
and MCMC posterior means match quite closely. If we assume, reasonably, that
they continue to match to first order
in a neighborhood of our original prior parameters, then
\prettyref{cond:vb_accurate} will hold and we would expect
$\psenshat\approx\qsens$.

\prettyref{tab:prior_sens} shows the detailed MFVB normalized sensitivity
results. Each entry is the sensitivity of the MFVB mean of the row's
parameter to the column's prior parameter. One can see that several
parameters are quite sensitive to the information parameter prior
$\tau_{\mu}.$ In particular, $\mbe_{\pthetapost}\left[\mu\right]$
and $\mbe_{\pthetapost}\left[\beta_{1}\right]$ are expected to change
approximately $-0.39$ and $-0.35$ standard deviations, respectively,
for every unit change in $\tau_{\mu}$. This size of change could
be practically significant (assuming that such a change in $\tau_{\mu}$
is subjectively plausible). To investigate this sensitivity further,
we re-fit the MFVB model at a range of values of the prior parameter
$\tau_{\mu}$, assessing the accuracy of the linear approximation
to the sensitivity. The results are shown in \fig{LogitGLMMRefit}.
Even for very large changes in $\tau_{\mu}$---resulting
in changes to $\mbe_{\pthetapost}\left[\mu\right]$ and $\mbe_{\pthetapost}\left[\beta_{1}\right]$
far in excess of two standard deviations---the linear approximation
holds up reasonably well. \fig{LogitGLMMRefit} also shows a (randomly
selected) random effect to be quite sensitive, though not to a practically
important degree relative to its posterior standard deviation. The
insensitivity of $\mbe_{\pthetapost}\left[\beta_{2}\right]$ is also
confirmed. Of course, the accuracy of the linear approximation cannot
be guaranteed to hold as well in general as it does in this particular
case, and the quick and reliable evaluation of the linearity assumption
without re-fitting the model remains interesting future work.

Since we started the MFVB optimization close to the new, perturbed
optimum, each new MFVB fit took only $\glmmVBRefitTime$ seconds
on average. Re-estimating the MCMC posterior so many times would have
been extremely time-consuming. (Note that importance sampling would
be useless for prior parameter changes that moved the posterior so
far from the original draws.) The considerable sensitivity of this
model to a particular prior parameter, which is perhaps surprising
on such a large data set, illustrates the value of having fast, general
tools for discovering and evaluating prior sensitivity. Our framework
provides just such a set of tools.

\section{Conclusion
\label{sec:Conclusion}}

By calculating the sensitivity of MFVB posterior means to model perturbations,
we are able to provide two important practical tools for MFVB posterior
approximations: improved variance estimates and measures of prior
robustness. When MFVB models are implemented in software that supports
automatic differentiation, our methods are fast, scalable, and require
little additional coding beyond the MFVB objective itself. In our experiments,
we were able to calculate accurate posterior means, covariances, and
prior sensitivity measures orders of magnitude more quickly than MCMC.

\iftoggle{arxivformat}{
\section*{Acknowledgements}
We are grateful to the anonymous reviewers for their insightful comments and
suggestions.
Ryan Giordano's research was funded in part by the National Energy Research
Scientific Computing Center, a DOE Office of Science User Facility supported by
the Office of Science of the U.S. Department of Energy under Contract number
DE-AC02- 05CH11231, and in part by the Gordon and Betty Moore Foundation through
Grant GBMF3834 and by the Alfred P. Sloan Foundation through Grant 2013-10-27 to
the University of California, Berkeley.
Tamara Broderick's research was supported in part by an NSF CAREER Award, an ARO
YIP Award, and a Google Faculty Research Award.
This work was also supported by the DARPA program on Lifelong Learning Machines,
the Office of Naval Research under contract/grant number N00014-17-1-2072, and
the Army Research Office under grant number W911NF-17-1-0304.
}
{
\acks{
We are grateful to the anonymous reviewers for their insightful comments and
suggestions.
Ryan Giordano's research was funded in part by the National Energy Research
Scientific Computing Center, a DOE Office of Science User Facility supported by
the Office of Science of the U.S. Department of Energy under Contract number
DE-AC02- 05CH11231, and in part by the Gordon and Betty Moore Foundation through
Grant GBMF3834 and by the Alfred P. Sloan Foundation through Grant 2013-10-27 to
the University of California, Berkeley.
Tamara Broderick's research was supported in part by an NSF CAREER Award, an ARO
YIP Award, and a Google Faculty Research Award.
This work was also supported by the DARPA program on Lifelong Learning Machines,
the Office of Naval Research under contract/grant number N00014-17-1-2072, and
the Army Research Office under grant number W911NF-17-1-0304.
}
}
 
\clearpage{}

\appendix
\appendixpage

\section{\label{app:sens_and_cov}Proof of \prettyref{thm:sens_cov}}

In this section we prove \prettyref{thm:sens_cov}.
\begin{proof}
Under \prettyref{assu:exchange_order}, we can exchange differentiation
and integration in
$\frac{\partial}{\partial \alpha^{\trans}} \int \pzeropost\left(\theta\right)\exp\left(\covdens\right)
g\left(\theta\right)\lambda\left(d\theta\right)$
and $\frac{\partial}{\partial \alpha^{\trans}} \int \pzeropost\left(\theta\right)\exp\left(\covdens\right)\lambda\left(d\theta\right)$
by \citet[Chapter 5-11, Theorem 18]{fleming:1965:functions}, which
ultimately depends on the Lebesgue dominated convergence theorem.
By \prettyref{assu:exchange_order},
$\mbe_{\pthetapost}\left[g\left(\theta\right)\right]$
is well-defined for $\alpha\in\alphazeroball$ and
\begin{align*}
\frac{\partial \pzeropost\left(\theta\right)\exp\left(\covdens\right)}{\partial\alpha} &
    =\pzeropost\left(\theta\right)\exp\left(\covdens\right)
    \frac{\partial\covdens}{\partial\alpha}\quad\lambda\textrm{-almost everywhere}.
\end{align*}

Armed with these facts, we can directly compute
\begin{align*}
\left.\frac{d\mbe_{\pthetapost}\left[g\left(\theta\right)\right]}
    {d\alpha^{\trans}}\right|_{\alpha_{0}} &
    =\left.\frac{d}{d\alpha^{\trans}}
    \frac{\int g\left(\theta\right)
          \pzeropost\left(\theta\right)\exp\left(\covdens\right)
    \lambda\left(d\theta\right)}
    {\int \pzeropost\left(\theta\right)\exp\left(\covdens\right)
    \lambda\left(d\theta\right)}\right|_{\alpha_{0}}\\
 & =\frac{\left.\frac{\partial}{\partial\alpha^{\trans}}
    \int g\left(\theta\right)
    \pzeropost\left(\theta\right)\exp\left(\covdens\right)
    \lambda\left(d\theta\right)\right|_{\alpha_{0}}}
    {\int \pzeropost\left(\theta\right)
    \exp\left(\covdens[][\alpha_{0}]\right)\lambda\left(d\theta\right)}-
        \mbe_{\pzeropost}\left[g\left(\theta\right)\right]
        \frac{\left.\frac{\partial}{\partial\alpha^{\trans}}
        \int \pzeropost\left(\theta\right)\exp\left(\covdens\right)
        \lambda\left(d\theta\right)\right|_{\alpha_{0}}}
        {\int \pzeropost\left(\theta\right)\exp\left(\covdens[][\alpha_{0}]\right)
        \lambda\left(d\theta\right)}\\
 & =\frac{\int g\left(\theta\right)\pzeropost\left(\theta\right)
    \exp\left(\covdens\right)\left.\frac{\partial\covdens}{\partial\alpha}
    \right|_{\alpha_{0}}\lambda\left(d\theta\right)}
    {\int \pzeropost\left(\theta\right)\exp\left(\covdens[][\alpha_{0}]\right)
    \lambda\left(d\theta\right)}-\mbe_{\pzeropost}
    \left[g\left(\theta\right)\right]\mbe_{\pzeropost}
    \left[\left.\frac{\partial\covdens}{\partial\alpha}\right|_{\alpha_{0}}\right]\\
 & =\cov_{\pzeropost}\left(g\left(\theta\right),\left.
    \frac{\partial\covdens}{\partial\alpha}\right|_{\alpha_{0}}\right).
\end{align*}
\end{proof}

\section{Comparison With MCMC Importance Sampling\label{app:mcmc_importance_sampling}}

In this section, we show that using importance sampling with MCMC
samples to calculate the local sensitivity in \prettyref{eq:local_robustness}
is precisely equivalent to using the same MCMC samples to estimate
the covariance in \prettyref{eq:covariance_sensitivity_general} directly.
For this section, will suppose that \prettyref{assu:exchange_order}
holds. Further suppose, without loss of generality, we have samples
$\theta_{i}$ drawn IID from $\pzeropost\left(\theta\right)$:
\begin{align*}
\theta_{n} & \iid\pzeropost\left(\theta\right),\textrm{ for }n=1,...,N_s\\
\mbe_{\pzeropost}\left[g\left(\theta\right)\right] & \approx\frac{1}{N_s}\sum_{n=1}^{N_s}g\left(\theta_{n}\right).
\end{align*}
Typically we cannot compute the dependence of the normalizing constant
$\int p\left(\theta'\right)\exp\left(\covdens[\theta']\right)\lambda\left(d\theta'\right)$
on $\alpha$, so we use the following importance sampling estimate
for $\mbe_{\pthetapost}\left[g\left(\theta\right)\right]$ \citep[Chapter 9]{owen:2013:mcmcbook}:
\begin{align*}
w_{n} & =\exp\left(\covdens[\theta_{n}]-\covdens[\theta_{n}][\alpha_{0}]\right)\\
\tilde{w}_{n} & :=\frac{w_{n}}{\sum_{n'=1}^{N_s}w_{n'}}\\
\mbe_{\pthetapost}\left[g\left(\theta\right)\right] & \approx\sum_{n=1}^{N_s}\tilde{w}_{n}g\left(\theta_{n}\right).
\end{align*}
Note that $\left.\tilde{w}_{n}\right|_{\alpha_{0}}=\frac{1}{N_s}$,
so the importance sampling estimate recovers the ordinary sample mean
at $\alpha_{0}$. The derivatives of the weights are given by
\begin{align*}
\frac{\partial w_{n}}{\partial\alpha} & =
    w_{n}\frac{\partial\covdens[\theta_{n}]}{\partial\alpha}\\
\frac{\partial\tilde{w}_{n}}{\partial\alpha} & =
    \frac{\frac{\partial w_{n}}{\partial\alpha}}{\sum_{n'=1}^{N_s}w_{n'}}-
    \frac{w_{n}\sum_{n'=1}^{N_s}\frac{\partial w_{n'}}{\partial\alpha}}
    {\left(\sum_{n'=1}^{N_s}w_{n'}\right)^{2}}\\
 & =\frac{w_{n}}{\sum_{n'=1}^{N_s}w_{n'}}
    \frac{\partial\covdens[\theta_{n}]}{\partial\alpha} -
    \frac{w_{n}}{\sum_{n'=1}^{N_s}w_{n'}}\sum_{n'=1}^{N_s}
    \frac{w_{n}}{\sum_{n'=1}^{N_s}w_{n'}}
    \frac{\partial\covdens[\theta_{n'}]}{\partial\alpha}\\
 & =\tilde{w}_{n}\frac{\partial\covdens[\theta_{n}]}{\partial\alpha}-
    \tilde{w}_{n}\sum_{n'=1}^{N_s}\tilde{w}_{n'}
    \frac{\partial\covdens[\theta_{n'}]}{\partial\alpha}.
\end{align*}
It follows that
\begin{align*}
\left.\frac{\partial}{\partial\alpha}
    \sum_{n=1}^{N_s}\tilde{w}_{n}g\left(\theta_{n}\right)\right|_{\alpha_{0}} &
    =\sum_{n=1}^{N_s}\left.\left(\tilde{w}_{n}\frac{\partial\covdens[\theta_{n}]}
    {\partial\alpha}-
    \tilde{w}_{n}\sum_{n'=1}^{N_s}\tilde{w}_{n'}
    \frac{\partial\covdens[\theta_{n'}]}{\partial\alpha}\right)
    \right|_{\alpha_{0}}g\left(\theta_{n}\right)\\
 & =\frac{1}{N_s}\sum_{n=1}^{N_s}
    \left.\frac{\partial\covdens[\theta_{n}]}{\partial\alpha}\right|_{\alpha_{0}}
    g\left(\theta_{n}\right)
    - \left[ \frac{1}{N_s}\sum_{n=1}^{N_s}\left.
    \frac{\partial\covdens[\theta_{n}]}{\partial\alpha}\right|_{\alpha_{0}} \right]
    \left[ \frac{1}{N_s}\sum_{n=1}^{N_s}g\left(\theta_{n}\right) \right],
\end{align*}
which is precisely the sample version of the covariance in
\prettyref{thm:sens_cov}.

\section{Our Use of the Terms ``Sensitivity'' and ``Robustness''\label{app:sens_and_robustness}}

In this section we clarify our usage of the terms ``robustness''
and ``sensitivity.'' The quantity $\psens^{\trans}\left(\alpha-\alpha_{0}\right)$
measures the \emph{sensitivity} of $\epgtheta$ to perturbations in
the direction $\Delta\alpha$. Intuitively, as sensitivity increases,
robustness decreases, and, in this sense, sensitivity and robustness
are opposites of one another. However, we emphasize that sensitivity
is a clearly defined, measurable quantity and that robustness is a
subjective judgment informed by sensitivity, but also by many other
less objective considerations.

Suppose we have calculated $\psens$ from \prettyref{eq:local_robustness}
and found that it has a particular value. To determine whether our
model is robust, we must additionally decide
\begin{enumerate}
\item How large of a change in the prior, $||\alpha-\alpha_{0}||$, is plausible,
and
\item How large of a change in $\epgtheta$ is important.
\end{enumerate}
The set of plausible prior values necessarily remains a subjective
decision.\footnote{This decision can be cast in a formal decision theoretic
framework based on a partial ordering of subjective beliefs
\citep{insua:2012:robustaxioms}.} Whether or not a particular change in
$\epgtheta$ is important depends on the ultimate use of the posterior mean. For
example, the posterior standard deviation can be a guide: if the prior
sensitivity is swamped by the posterior uncertainty then it can be neglected
when reporting our subjective uncertainty about $g\left(\theta\right)$, and the
model is robust. Similarly, even if the prior sensitivity is much larger than
the posterior standard deviation but small enough that it would not affect any
actionable decision made on the basis of the value of $\epgtheta$, then the
model is robust. Intermediate values remain a matter of judgment. An illustration
of the relationship between sensitivity and robustness is shown in
\fig{robustness_vs_sensitivity}.

\begin{figure}
\centering{}\includegraphics[width=0.56\textwidth]{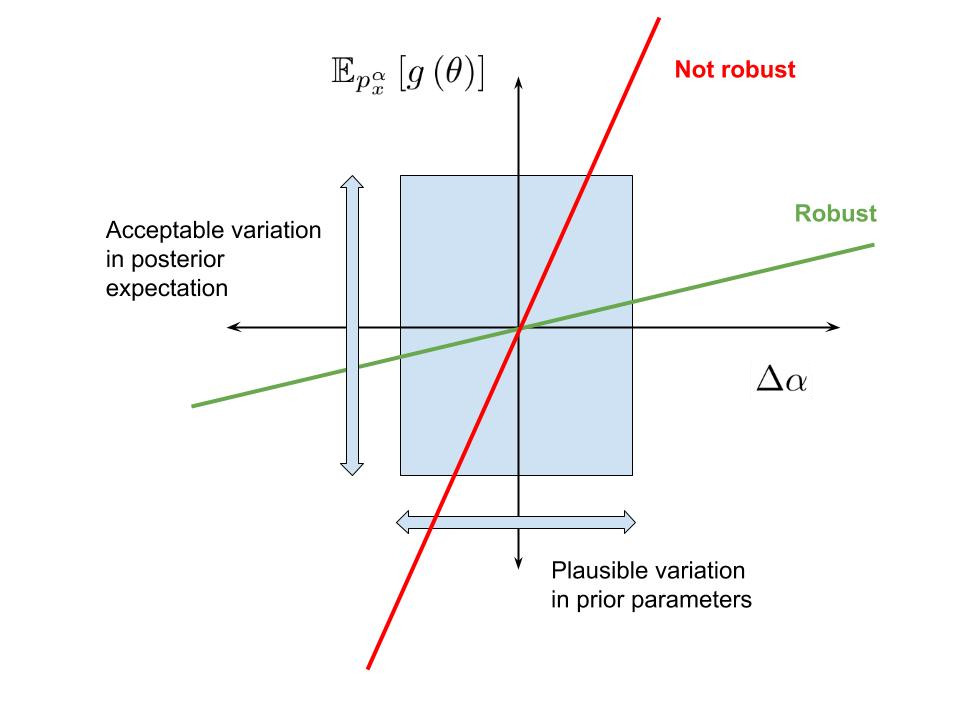}
\caption{The relationship between robustness and sensitivity\label{fig:robustness_vs_sensitivity} }
\end{figure}

Finally, we note that if $\mathcal{A}$ is small enough that $\epgtheta$
is roughly linear in $\alpha$ for $\alpha\in\mathcal{A}$, then calculating
\prettyref{eq:local_robustness} for all $\alpha\in\mathcal{A}$ and
finding the worst case can be thought of as a first-order approximation
to a global robustness estimate. Depending on the problem at hand,
this linearity assumption may not be plausible except for very small
$\mathcal{A}$. This weakness is inherent to the local robustness
approach. Nevertheless, even when the perturbations are valid only
for a small $\mathcal{A}$, these easily-calculable measures may still
provide valuable intuition about the potential modes of failure for
a model.

If $\gtheta$ is a scalar, it is natural to attempt to summarize the
high-dimensional vector $\psens$ in a single easily reported number
such as
\begin{align*}
\psens^{sup} & :=\sup_{\alpha:\left\Vert \alpha - \alpha_{0} \right\Vert \le1}
    \left|\psens^{\trans}\left(\alpha-\alpha_{0}\right)\right|.
\end{align*}
For example, the calculation of $\psens^{sup}$ is the principal ambition of
\citet{basu:1996:local}. The use of such summaries is also particularly common
in work that considers function-valued perturbations
\citep[e.g.,][]{gustafson:1996:localposterior,roos:2015:sensitivity}.
(Function-valued perturbations can be connected to the finite-dimensional
perturbations of the present work through the notion of the Gateaux derivative
\citep[Chapter 2.5]{huber:2011:robust}, the elaboration of which we leave to
future work.) Although the summary $\psens^{sup}$ has obvious merits, in the
present work we emphasize the calculation only of $\psens$ in the belief that
its interpretation is likely to vary from application to application and require
some critical thought and subjective judgment. For example, the unit ball
$\left\Vert \alpha - \alpha_{0}\right\Vert \le1$ (as in \citet{basu:1996:local}) may not make
sense as a subjective description of the range of plausible variability of
$p\left(\theta\vert\alpha\right)$. Consider, e.g.: why should the off-diagonal term of
a Wishart prior plausibly vary as widely as the mean of some other parameter,
when the two might not even have the same units? This problem is easily remedied
by choosing an appropriate scaling of the parameters and thereby making the unit
ball an appropriate range for the problem at hand, but the right scaling will
vary from problem to problem and necessarily be a somewhat subjective choice, so
we refrain from taking a stand on this decision. As another example, the
worst-case function-valued perturbations of
\citet{gustafson:1996:localmarginals,gustafson:1996:localposterior} require a
choice of a metric ball in function space whose meaning may not be intuitively
obvious, may provide worst-case perturbations that depend on the data to a
subjectively implausible degree, and may exhibit interesting but perhaps
counter-intuitive asymptotic behavior for different norms and perturbation
dimensions. Consequently, we do not attempt to prescribe a particular
one-size-fits-all summary measure. The local sensitivity $\psens$ is a
well-defined mathematical quantity. Its relationship to robustness must remain a
matter of judgment.

\section{\label{app:lrvb}Proof of \prettyref{thm:lrvb_formula}}

In this section we prove \prettyref{thm:lrvb_formula}.
\begin{proof}
For notational convenience, we will define
\begin{eqnarray*}
KL\left(\eta,\alpha\right) & := &
    KL\left(q\left(\theta;\eta\right)||\pthetapost\left(\theta\right)\right).
\end{eqnarray*}
By \prettyref{assu:opt_interior}, $\etatopt$ is both optimal and interior
for all $\alpha\in\alphazeroball$, and by \prettyref{assu:kl_differentiable},
$KL\left(\eta,\alpha\right)$ is continuously differentiable in $\eta$. Therefore, the
first-order conditions of the optimization problem in \prettyref{eq:kl_divergence} give:
\begin{eqnarray}
\left.\frac{\partial KL\left(\eta,\alpha\right)}{\partial\eta}\right|_{\eta=\etatopt}
    & = & 0\textrm{ for all }\alpha\in\alphazeroball.
    \label{eq:kl_first_order_condition}
\end{eqnarray}
$\left.\frac{\partial^{2}KL\left(\eta,\alpha\right)}{\partial\eta\partial\eta^{\trans}}\right|_{\alpha_{0}}$
is positive definite by the strict optimality of $\etaopt$ in
\prettyref{assu:opt_interior}, and
$\frac{\partial^{2}KL\left(\eta,\alpha\right)}{\partial\eta\partial\alpha^{\trans}}$
is continuous by \prettyref{assu:kl_differentiable}. It follows that $\etatopt$ is a
continuously differentiable function of $\alpha$ by application of the implicit
function theorem to the first-order condition in
\prettyref{eq:kl_first_order_condition} \citep[Chapter
4.6]{fleming:1965:functions}. So we can use the chain rule to take the total
derivative of \prettyref{eq:kl_first_order_condition} with respect to
$\alpha$.
\begin{eqnarray*}
\frac{d}{d\alpha}\left(\left.
    \frac{\partial KL\left(\eta,\alpha\right)}{\partial\eta}
    \right|_{\eta=\etatopt}\right) & = &
    0\textrm{ for all }\alpha\in\alphazeroball\Rightarrow\\
\left.\frac{\partial^{2}KL\left(\eta,\alpha\right)}
    {\partial\eta\partial\eta^{\trans}}\right|_{\eta=\etatopt}
    \frac{d\etatopt}{d\alpha^{\trans}}+\left.
    \frac{\partial^{2}KL\left(\eta,\alpha\right)}
        {\partial\eta\partial\alpha^{\trans}}\right|_{\eta=\etatopt} & = &
        0\textrm{ for all }\alpha\in\alphazeroball.
\end{eqnarray*}
The strict optimality of $KL\left(\eta,\alpha\right)$ at
$\etaopt\left(\alpha\right)$ in \prettyref{assu:opt_interior} requires that
$\left.\frac{\partial^{2}KL\left(\eta,t\right)}{\partial\eta\partial\eta^{T}}\right|_{\eta=\etatopt}$
be invertible. So we can evaluate at $\alpha=\alpha_{0}$ and solve to find that
\begin{eqnarray*}
\left.\frac{d\etatopt}{d\alpha^{\trans}}\right|_{\alpha_{0}} & = &
    -\left.\left(\frac{\partial^{2}KL\left(\eta,\alpha\right)}
    {\partial\eta\partial\eta^{\trans}}\right)^{-1}
    \frac{\partial^{2}KL\left(\eta,\alpha\right)}
    {\partial\eta\partial\alpha^{\trans}}\right|_{\eta=\etaoptzero,\alpha=\alpha_{0}}.
\end{eqnarray*}
$\mbe_{\qthetapost}\left[\gtheta\right]$ is a continuously differentiable
function of $\etatopt$ by \prettyref{assu:e_q_g_smooth}. So by the
chain rule and \prettyref{assu:kl_differentiable}, we have that
\begin{eqnarray*}
\left.\frac{d\mbe_{q\left(\theta;\eta\right)}
    \left[\gtheta\right]}{d\alpha^{\trans}}\right|_{\alpha_{0}} & = &
    \left.\frac{\partial\mbe_{q\left(\theta;\eta\right)}
    \left[g\left(\theta\right)\right]}{\partial\eta}\frac{d\etatopt}{d\alpha^{\trans}}
    \right|_{\eta=\etaoptzero,\alpha=\alpha_{0}}.
\end{eqnarray*}
Finally, we observe that
\begin{eqnarray*}
KL\left(\eta,\alpha\right) & = &
    \mbe_{q\left(\theta;\eta\right)}\left[
    \log q\left(\theta;\eta\right)-\log p\left(\theta\right) -
    \covdens\right]+\constant\Rightarrow\\
\left.\frac{\partial^{2}KL\left(\eta,\alpha\right)}
    {\partial\eta\partial\alpha^{\trans}}
    \right|_{\eta=\etaoptzero,\alpha=\alpha_{0}} & = &
-\left.\frac{\partial^{2}\mbe_{q\left(\theta;\eta\right)}
    \left[\covdens\right]}{\partial\eta\partial\alpha^{\trans}}
    \right|_{\eta=\etaoptzero,\alpha=\alpha_{0}}.
\end{eqnarray*}
Here, the term $\constant$ contains quantities that do not depend
on $\eta$. Plugging in gives the desired result.
\end{proof}

\section{Exactness of Multivariate Normal Posterior Means\label{app:mvn_exact}}

In this section, we show that the MFVB estimate of the posterior means
of a multivariate normal with known covariance is exact and that,
as an immediate consequence, the linear response covariance recovers
the exact posterior covariance, i.e., $\lrvbcov\left(\theta\right)=\pcov\left(\theta\right)$.

Suppose we are using MFVB to approximate a non-degenerate multivariate
normal posterior, i.e.,
\begin{align*}
\pzeropost\left(\theta\right) & =\mathcal{N}\left(\theta;\mu,\covmat\right)
\end{align*}
for full-rank $\covmat$. This posterior arises, for instance, given
a multivariate normal likelihood
$p\left(x\vert\mu\right)=\prod_{n=1:N}\mathcal{N}\left(x_{n}\vert\theta,\covmat_{x}\right)$
with known covariance $\covmat_{x}$ and a conjugate multivariate
normal prior on the unknown mean parameter $\theta\in\mathbb{R}^{K}$.
Additionally, even when the likelihood is non-normal or the prior
is not conjugate, the posterior may be closely approximated by a multivariate
normal distribution when a Bayesian central limit theorem can be applied
\citep[Chapter 8]{lecam:2012:asymptotics}.

We will consider an MFVB approximation to $\pzeropost\left(\theta\right)$.
Specifically, let the elements of the vector $\theta$ be given by
scalars $\theta_{k}$ for $k=1,...,K$, and take the MFVB normal approximation
with means $m_{k}$ and variances $v_{k}$:
\begin{align*}
\mathcal{Q} & =\left\{ q\left(\theta\right):
    q\left(\theta\right)=\prod_{k=1}^{K}
    \mathcal{N}\left(\theta_{k};m_{k},v_{k}\right)\right\}.
\end{align*}
In the notation of \prettyref{eq:q_mean_field_family}, we have $\eta_{k}=\left(m_{k},v_{k}\right)^{\trans}$
with $\Omega_{\eta}=\left\{ \eta:v_{k}>0,\forall k=1,...,K\right\} $.
The optimal variational parameters are given by $\eta_{k}^{*}=\left(m_{k}^{*},v_{k}^{*}\right)^{\trans}.$
\begin{lrvb_lemma}
    \label{lem:mvn_exact}Let
$\pzeropost\left(\theta\right)=\mathcal{N}\left(\theta;\mu,\covmat\right)$ for
full-rank $\covmat$ and let $\mathcal{Q}=\left\{
q\left(\theta\right):q\left(\theta\right)=\prod_{k=1}^{K}\mathcal{N}\left(\theta_{k};m_{k},v_{k}\right)\right\} $
be the mean field approximating family. Then there exists an
$\etaopt=\left(m^{*},v^{*}\right)$ that solves
\begin{align*}
\etaopt & =\argmin_{\eta:q\left(\theta;\eta\right)\in\mathcal{Q}}KL\left(q\left(\theta;\eta\right)||\pthetapost\left(\theta\right)\right)
\end{align*}
with $m^{*}=\mu$.
\end{lrvb_lemma}
\begin{proof}
    Let $\textrm{diag}\left(v\right)$ denote the $K\times K$ matrix
with the vector $v$ on the diagonal and zero elsewhere. Using the
fact that the entropy of a univariate normal distribution with variance
$v$ is $\frac{1}{2}\log v$ plus a constant, the variational objective in
\prettyref{eq:kl_divergence} is given by
\begin{align}
KL\left(q\left(\theta;\eta\right)||\pthetapost\left(\theta\right)\right) &
    =\mbe_{q\left(\theta;\eta\right)}\left[
        \frac{1}{2}\left(\theta-\mu\right)^{\trans}\covmat^{-1}
        \left(\theta-\mu\right)\right]-\frac{1}{2}\sum_{k}\log v_{k}+
        \constant\nonumber \\
 & =\frac{1}{2}\textrm{trace}\left(\covmat^{-1}
    \mbe_{q\left(\theta;\eta\right)}\left[
    \theta\theta^{\trans}\right]\right)-
    \mu^{\trans}\covmat^{-1}\mbe_{q\left(\theta;\eta\right)}
    \left[\theta\right]-\frac{1}{2}\sum_{k}\log v_{k}+\constant\nonumber \\
 & =\frac{1}{2}\textrm{trace}\left(\covmat^{-1}
    \left(mm^{\trans}+\textrm{diag}\left(v\right)\right)\right)-
    \mu^{\trans}\covmat^{-1}m-\frac{1}{2}\sum_{k}\log v_{k}+\constant\nonumber \\
 & =\frac{1}{2}\textrm{trace}\left(\covmat^{-1}
    \textrm{diag}\left(v\right)\right)+
    \frac{1}{2}m^{\trans}\covmat^{-1}m-
    \mu^{\trans}\covmat^{-1}m-\frac{1}{2}\sum_{k}\log v_{k}+
    \constant.\label{eq:mvn_kl_divergence}
\end{align}
The first-order condition for the optimal $m^{*}$ is then
\begin{align*}
\left.\frac{\partial KL\left(q\left(\theta;\eta\right)||
    \pzeropost\left(\theta\right)\right)}{\partial m}
    \right|_{m=m^{*},v=v^{*}} & =0\Rightarrow\\
\covmat^{-1}m^{*}-\covmat^{-1}\mu & =0\Rightarrow\\
m^{*} & =\mu.
\end{align*}
The optimal variances follow similarly:
\begin{align*}
\left.\frac{\partial KL\left(q\left(\theta;\eta\right)||
    \pzeropost\left(\theta\right)\right)}{\partial v_{k}}
    \right|_{m=m^{*},v=v^{*}} & =0\Rightarrow\\
\frac{1}{2}\left(\covmat^{-1}\right)_{kk}-\frac{1}{2}
    \frac{1}{v_{k}^{*}} & =0\Rightarrow\\
v_{k}^{*} & =\frac{1}{\left(\covmat^{-1}\right)_{kk}}.
\end{align*}
Since $v_{k}^{*}>0$, we have $\etaopt\in\Omega_{\eta}$.

\prettyref{lem:mvn_exact} can be also be derived via the variational
coordinate ascent updates (\citet[Section 10.1.2]{bishop:2006:pattern}
and \citet[Appendix B]{giordano:2015:lrvb}).
\end{proof}
Next, we show that \prettyref{lem:mvn_exact} holds for all perturbations
of the form $\covdens=\alpha^{\trans}\theta$ with $\alpha_{0}=0$
and that \vbassum are satisfied for all finite $\alpha$.
\begin{lrvb_lemma}
\label{lem:mvn_perturbation_exact}Under the conditions of
\prettyref{lem:mvn_exact}, let $\pthetapost\left(\theta\right)$ be defined from
\prettyref{eq:tilting_definition} with $\covdens=\alpha^{\trans}\theta$ and
$\alpha_{0}=0$. Take $\gtheta=\theta$. Then, for all finite $\alpha$, \vbassum
are satisfied, and \prettyref{cond:vb_accurate} is satisfied with equality.
\end{lrvb_lemma}
\begin{proof}
Up to a constant that does not depend on $\theta$, the log density
of $\pthetapost\left(\theta\right)$ is
\begin{align*}
\log\pthetapost\left(\theta\right) & =
    -\frac{1}{2}\left(\theta-\mu\right)^{\trans}\covmat^{-1}
    \left(\theta-\mu\right)+\alpha^{\trans}\theta+\constant\\
 & =-\frac{1}{2}\theta^{\trans}\Sigma^{-1}\theta-
    \frac{1}{2}\mu^{\trans}\Sigma^{-1}\mu+\left(\mu^{\trans}\Sigma^{-1}+
    \alpha^{\trans}\right)\theta+\constant.
\end{align*}
Since $\theta$ is a natural sufficient statistic of the multivariate normal
distribution and the corresponding natural parameter of
$\pthetapost\left(\theta\right),$ $\covmat^{-1}\mu+\alpha$, is interior when
$\covmat$ is full-rank, $\pthetapost\left(\theta\right)$ is multivariate normal
for any finite $\alpha$. \prettyref{assu:exchange_order} follows immediately.

By inspection of \prettyref{eq:mvn_kl_divergence},
\prettyref{assu:kl_differentiable} is satisfied. Because $\Omega_{\eta}$ is an
open set and $\Sigma$ is positive definite, \prettyref{assu:opt_interior} is
satisfied. Since $\mbe_{q\left(\theta;\eta\right)}\left[\gtheta\right]=m$,
\prettyref{assu:e_q_g_smooth} is satisfied. Finally, by
\prettyref{lem:mvn_exact},
$\mbe_{\qthetapost}\left[\theta\right]=\mbe_{\pthetapost}\left[\theta\right]$ ,
so \prettyref{cond:vb_accurate} is satisfied with equality.
\end{proof}
It now follows immediately from
\prettyref{def:lrvb_covariance} that the linear response variational
covariance exactly reproduces the exact posterior covariance for the
multivariate normal distribution.
\begin{lrvb_corollary}
\label{cor:mvn_cov_exact}
Under the conditions of \prettyref{lem:mvn_perturbation_exact},
$\lrvbcov\left(\theta\right)=\pcov\left(\theta\right)$.
\end{lrvb_corollary}

\section{ADVI Model Details}\label{app:advi}

This section reports the Stan code for the models used in \prettyref{subsec:advi}.
For details on how to interpret the models as well as the unconstraining
transforms, see the Stan manual \citep{stan-manual:2015}.  For the associated data, see
the Stan example models wiki \citep{stan-examples:2017}.

\subsection{Election Model (\texttt{election88.stan})}
\begin{lstlisting}[caption=\texttt{election88.stan}]
data {
  int<lower=0> N; 
  int<lower=0> n_state; 
  vector<lower=0,upper=1>[N] black;
  vector<lower=0,upper=1>[N] female;
  int<lower=1,upper=n_state> state[N];
  int<lower=0,upper=1> y[N];
} 
parameters {
  vector[n_state] a;
  vector[2] b;
  real<lower=0,upper=100> sigma_a;
  real mu_a;
}
transformed parameters {
  vector[N] y_hat;

  for (i in 1:N)
    y_hat[i] <- b[1] * black[i] + b[2] * female[i] + a[state[i]];
} 
model {
  mu_a ~ normal(0, 1);
  a ~ normal (mu_a, sigma_a);
  b ~ normal (0, 100);
  y ~ bernoulli_logit(y_hat);
}
\end{lstlisting}

\subsection{Sesame Street Model (\texttt{sesame\_street1})}
\begin{lstlisting}[caption=\texttt{sesame\_street1.stan}]
data {
  int<lower=0> J;
  int<lower=0> N;
  int<lower=1,upper=J> siteset[N];
  vector[2] yt[N];
  vector[N] z;
}
parameters {
  vector[2] ag[J];
  real b;
  real d;
  real<lower=-1,upper=1> rho_ag;
  real<lower=-1,upper=1> rho_yt;
  vector[2] mu_ag;
  real<lower=0,upper=100> sigma_a;
  real<lower=0,upper=100> sigma_g;
  real<lower=0,upper=100> sigma_t;
  real<lower=0,upper=100> sigma_y;
}
model {
  vector[J] a;
  vector[J] g;
  matrix[2,2] Sigma_ag;
  matrix[2,2] Sigma_yt;
  vector[2] yt_hat[N];

  //data level
  Sigma_yt[1,1] <- pow(sigma_y,2);
  Sigma_yt[2,2] <- pow(sigma_t,2);
  Sigma_yt[1,2] <- rho_yt*sigma_y*sigma_t;  
  Sigma_yt[2,1] <- Sigma_yt[1,2];
   
  // group level
  Sigma_ag[1,1] <- pow(sigma_a,2);
  Sigma_ag[2,2] <- pow(sigma_g,2);
  Sigma_ag[1,2] <- rho_ag*sigma_a*sigma_g;
  Sigma_ag[2,1] <- Sigma_ag[1,2];  

  for (j in 1:J) {
    a[j] <- ag[j,1];
    g[j] <- ag[j,2];
  }

  for (i in 1:N) {
    yt_hat[i,2] <- g[siteset[i]] + d * z[i];
    yt_hat[i,1] <- a[siteset[i]] + b * yt[i,2];
  }

  //data level
  sigma_y ~ uniform (0, 100);               
  sigma_t ~ uniform (0, 100);                  
  rho_yt ~ uniform(-1, 1);
  d ~ normal (0, 31.6);
  b ~ normal (0, 31.6);

  //group level
  sigma_a ~ uniform (0, 100);
  sigma_g ~ uniform (0, 100);
  rho_ag ~ uniform(-1, 1);
  mu_ag ~ normal (0, 31.6);

  for (j in 1:J)
    ag[j] ~ multi_normal(mu_ag,Sigma_ag);

  //data model
  for (i in 1:N)
    yt[i] ~ multi_normal(yt_hat[i],Sigma_yt);

}
\end{lstlisting}

\subsection{Radon Model (\texttt{radon\_vary\_intercept\_floor})}
\begin{lstlisting}[caption=\texttt{radon\_vary\_intercept\_floor.stan}]
data {
  int<lower=0> J; 
  int<lower=0> N; 
  int<lower=1,upper=J> county[N];
  vector[N] u;
  vector[N] x;
  vector[N] y;
} 
parameters {
  vector[J] a;
  vector[2] b;
  real mu_a;
  real<lower=0,upper=100> sigma_a;
  real<lower=0,upper=100> sigma_y;
} 
transformed parameters {
  vector[N] y_hat;

  for (i in 1:N)
    y_hat[i] <- a[county[i]] + u[i] * b[1] + x[i] * b[2];
}
model {
  mu_a ~ normal(0, 1);
  a ~ normal(mu_a, sigma_a);
  b ~ normal(0, 1);
  y ~ normal(y_hat, sigma_y);
}
\end{lstlisting}

\subsection{Ecology Model (\texttt{cjs\_cov\_randeff})}
\begin{lstlisting}[caption=\texttt{cjs\_cov\_randeff.stan}]
// This models is derived from section 12.3 of "Stan Modeling Language
// User's Guide and Reference Manual"

functions {
  int first_capture(int[] y_i) {
    for (k in 1:size(y_i))
      if (y_i[k])
        return k;
    return 0;
  }

  int last_capture(int[] y_i) {
    for (k_rev in 0:(size(y_i) - 1)) {
      // Compoud declaration was enabled in Stan 2.13
      int k = size(y_i) - k_rev;
      //      int k;
      //      k = size(y_i) - k_rev;
      if (y_i[k])
        return k;
    }
    return 0;
  }

  matrix prob_uncaptured(int nind, int n_occasions,
                         matrix p, matrix phi) {
    matrix[nind, n_occasions] chi;

    for (i in 1:nind) {
      chi[i, n_occasions] = 1.0;
      for (t in 1:(n_occasions - 1)) {
        // Compoud declaration was enabled in Stan 2.13
        int t_curr = n_occasions - t;
        int t_next = t_curr + 1;
        /*
        int t_curr;
        int t_next;

        t_curr = n_occasions - t;
        t_next = t_curr + 1;
        */
        chi[i, t_curr] = (1 - phi[i, t_curr])
                        + phi[i, t_curr] * (1 - p[i, t_next - 1]) * chi[i, t_next];
      }
    }
    return chi;
  }
}

data {
  int<lower=0> nind;            // Number of individuals
  int<lower=2> n_occasions;     // Number of capture occasions
  int<lower=0,upper=1> y[nind, n_occasions];    // Capture-history
  vector[n_occasions - 1] x;    // Covariate
}

transformed data {
  int n_occ_minus_1 = n_occasions - 1;
  //  int n_occ_minus_1;
  int<lower=0,upper=n_occasions> first[nind];
  int<lower=0,upper=n_occasions> last[nind];
  vector<lower=0,upper=nind>[n_occasions] n_captured;

  //  n_occ_minus_1 = n_occasions - 1;
  for (i in 1:nind)
    first[i] = first_capture(y[i]);
  for (i in 1:nind)
    last[i] = last_capture(y[i]);
  n_captured = rep_vector(0, n_occasions);
  for (t in 1:n_occasions)
    for (i in 1:nind)
      if (y[i, t])
        n_captured[t] = n_captured[t] + 1;
}

parameters {
  real beta;                          // Slope parameter
  real<lower=0,upper=1> mean_phi;     // Mean survival
  real<lower=0,upper=1> mean_p;       // Mean recapture
  vector[n_occ_minus_1] epsilon;
  real<lower=0,upper=10> sigma;
  // In case a weakly informative prior is used
  //  real<lower=0> sigma;
}

transformed parameters {
  matrix<lower=0,upper=1>[nind, n_occ_minus_1] phi;
  matrix<lower=0,upper=1>[nind, n_occ_minus_1] p;
  matrix<lower=0,upper=1>[nind, n_occasions] chi;
  // Compoud declaration was enabled in Stan 2.13
  real mu = logit(mean_phi);
  //  real mu;

  //  mu = logit(mean_phi);
  // Constraints
  for (i in 1:nind) {
    for (t in 1:(first[i] - 1)) {
      phi[i, t] = 0;
      p[i, t] = 0;
    }
    for (t in first[i]:n_occ_minus_1) {
      phi[i, t] = inv_logit(mu + beta * x[t] + epsilon[t]);
      p[i, t] = mean_p;
    }
  }

  chi = prob_uncaptured(nind, n_occasions, p, phi);
}

model {
  // Priors
  // Uniform priors are implicitly defined.
  //  mean_phi ~ uniform(0, 1);
  //  mean_p ~ uniform(0, 1);
  //  sigma ~ uniform(0, 10);
  // In case a weakly informative prior is used
  //  sigma ~ normal(5, 2.5);
  beta ~ normal(0, 100);
  epsilon ~ normal(0, sigma);

  for (i in 1:nind) {
    if (first[i] > 0) {
      for (t in (first[i] + 1):last[i]) {
        1 ~ bernoulli(phi[i, t - 1]);
        y[i, t] ~ bernoulli(p[i, t - 1]);
      }
      1 ~ bernoulli(chi[i, last[i]]);
    }
  }
}

generated quantities {
  real<lower=0> sigma2;
  vector<lower=0,upper=1>[n_occ_minus_1] phi_est;

  sigma2 = square(sigma);
  // inv_logit was vectorized in Stan 2.13
  phi_est = inv_logit(mu + beta * x + epsilon); // Yearly survival
  /*
  for (t in 1:n_occ_minus_1)
    phi_est[t] = inv_logit(mu + beta * x[t] + epsilon[t]); 
  */
}
\end{lstlisting}

\section{LKJ Priors for Covariance Matrices in Mean Field Variational Inference\label{app:lkj}}

In this section we briefly derive closed-form expressions for using
an LKJ prior with a Wishart variational approximation.
\begin{proposition}
Let $\covmat$ be a $K\times K$ positive definite covariance matrix.
Define the $K\times K$ matrix $\lkjdiagmat$ such that
\begin{align*}
\lkjdiagmat_{ij} & =\begin{cases}
\sqrt{\covmat_{ij}} & \textrm{if }i=j\\
0 & \textrm{otherwise}.
\end{cases}
\end{align*}
Define the correlation matrix $R$ as
\begin{align*}
\lkjcorrmat & =\lkjdiagmat^{-1}\covmat\lkjdiagmat^{-1}.
\end{align*}
Define the LKJ prior on $\lkjcorrmat$ with concentration parameter
$\xi$ \citep{lewandowski:2009:lkj}:
\begin{align*}
p_{\mathrm{LKJ}}\left(\lkjcorrmat\vert\xi\right) & \propto\left|\lkjcorrmat\right|^{\xi-1}.
\end{align*}
Let $q\left(\covmat\vert\lkjlocmat^{-1},\nu\right)$ be an inverse
Wishart distribution with matrix parameter $\lkjlocmat^{-1}$ and $\nu$
degrees of freedom. Then
\begin{align*}
\mathbb{E}_{q}\left[\log\left|\lkjcorrmat\right|\right] & =\log\left|\lkjlocmat^{-1}\right|-\psi_{K}\left(\frac{\nu}{2}\right)-\sum_{k=1}^{K}\log\left(\left(\lkjlocmat^{-1}\right)_{kk}\right)+K\psi\left(\frac{\nu-K+1}{2}\right)+\constant\\
\mbe_{q}\left[\log p_{\mathrm{LKJ}}\left(\lkjcorrmat\vert\xi\right)\right] & =\left(\xi-1\right)\mathbb{E}_{q}\left[\log\left|\lkjcorrmat\right|\right]+\constant,
\end{align*}
where $\constant$ does not depend on $\lkjlocmat$ or $\nu$. Here,
$\psi_{K}$ is the multivariate digamma function.
\end{proposition}
\begin{proof}
    First note that
\begin{align}
\nonumber
\log\left|\covmat\right| & =
    2\log\left|\lkjdiagmat\right|+\log\left|\lkjcorrmat\right|\\
\nonumber
 & =2\sum_{k=1}^{K}\log\sqrt{\covmat_{kk}}+
    \log\left|\lkjcorrmat\right|\\
\nonumber
 & =\sum_{k=1}^{K}\log\covmat_{kk}+
    \log\left|\lkjcorrmat\right|\Rightarrow\\
\label{eq:logdet_corr}
\log\left|\lkjcorrmat\right| & =\log\left|\covmat\right|-
    \sum_{k=1}^{K}\log\covmat_{kk}.
\end{align}
By Eq.\ B.81 in \citep{bishop:2006:pattern}, a property of the inverse Wishart distribution is the following relation.
\begin{align}
\label{eq:inv_wish_logdet}
E_{q}\left[\log\left|\covmat\right|\right] &
    =\log\left|\lkjlocmat^{-1}\right|-\psi_{K}\left(\frac{\nu}{2}\right)-K\log2,
\end{align}
where $\psi_{K}$ is the multivariate digamma function. By the marginalization
property of the inverse Wishart distribution,
\begin{align}
\nonumber
\covmat_{kk} & \sim\textrm{InverseWishart}\left(
    \left(\lkjlocmat^{-1}\right)_{kk},\nu-K+1\right)\Rightarrow\\
\label{eq:inv_wish_logdet_comps}
E_{q}\left[\log\covmat_{kk}\right] & =
    \log\left(\left(\lkjlocmat^{-1}\right)_{kk}\right)-
    \psi\left(\frac{\nu-K+1}{2}\right)-\log2.
\end{align}
Plugging \prettyref{eq:inv_wish_logdet} and \prettyref{eq:inv_wish_logdet_comps} into \prettyref{eq:logdet_corr} gives the desired result.
\end{proof}

\section{Logistic GLMM Model Details\label{app:glmm_details}}

In this section we include extra details about the model and analysis
of \prettyref{sec:experiments}. We will continue to use the notation
defined therein. We use $\constant$ to denote any constants that do not depend
on the prior parameters, parameters, or data. The log likelihood is
\begin{eqnarray}
\log p\left(y_{it}\vert u_{t},\beta\right) & = &
    y_{it}\log\left(\frac{p_{it}}{1-p_{it}}\right)+
    \log\left(1-p_{it}\right)\nonumber \\
 & = & y_{it}\rho+\log\left(1-p_{it}\right)+\constant\nonumber \\
\log p\left(u\vert\mu,\tau\right) & = &
    -\frac{1}{2}\tau\sum_{t=1}^{T}\left(u_{t}-\mu\right)^{2}-
    \frac{1}{2}T\log\tau\nonumber \\
 & = & -\frac{1}{2}\tau\sum_{t=1}^{T}\left(u_{t}^{2}-\mu u_{t}+
    \mu^{2}\right)-\frac{1}{2}T\log\tau+\constant\nonumber \\
\log p\left(\mu,\tau,\beta\right) & = &
    -\frac{1}{2}\sigma_{\mu}^{-2}\left(\mu^{2}+2\mu\mu_{0}\right)+\nonumber \\
 &  & \left(1-\alpha_{\tau}\right)\tau+\beta_{\tau}\log\tau+\nonumber \\
 &  & -\frac{1}{2}\left(\textrm{trace}\left(
    \covmat_{\beta}^{-1}\beta\beta^{T}\right)+2\textrm{trace}
    \left(\covmat_{\beta}^{-1}\beta_{0}\beta^{T}\right)\right).
    \label{eq:glmm_log_lik}
\end{eqnarray}

The prior parameters were taken to be
\begin{align*}
\mu_{0} & =\glmmMuLoc\\
\sigma_{\mu}^{-2} & =\glmmMuInfo\\
\beta_{0} & =\glmmBetaLoc\\
\sigma_{\beta}^{-2} & =\glmmBetaInfoDiag\\
\alpha_{\tau} & =\glmmTauAlpha\\
\beta_{\tau} & =\glmmTauBeta.
\end{align*}

Under the variational approximation, $\rho_{it}$ is normally distributed
given $x_{it}$, with
\begin{eqnarray*}
\rho_{it} & = & x_{it}^{T}\beta+u_{t}\\
\mbe_{q}\left[\rho_{it}\right] & = &
    x_{it}^{T}\mbe_{q}\left[\beta\right]+\mbe_{q}\left[u_{t}\right]\\
\textrm{Var}_{q}\left(\rho_{it}\right) & = &
    \mbe_{q}\left[\beta^{T}x_{it}x_{it}^{T}\beta\right]-
    \mbe_{q}\left[\beta\right]^{T}x_{it}x_{it}^{T}\mbe_{q}
    \left[\beta\right]+\textrm{Var}_{q}\left(u_{t}\right)\\
 & = & \mbe_{q}\left[\textrm{tr}\left(\beta^{T}x_{it}x_{it}^{T}
    \beta\right)\right]-\textrm{tr}\left(\mbe_{q}
    \left[\beta\right]^{T}x_{it}x_{it}^{T}\mbe_{q}\left[\beta\right]\right)+
    \textrm{Var}_{q}\left(u_{t}\right)\\
 & = & \textrm{tr}\left(x_{it}x_{it}^{T}\left(\mbe_{q}
    \left[\beta\beta^{T}\right]-\mbe_{q}\left[\beta\right]\mbe_{q}
    \left[\beta\right]^{T}\right)\right)+\textrm{Var}_{q}\left(u_{t}\right).
\end{eqnarray*}

We can thus use $n_{MC}=\glmmNumGHPoints$ points of Gauss-Hermite
quadrature to numerically estimate
$\mbe_{q}\left[\log\left(1-\frac{e^{\rho}}{1+e^{\rho}}\right)\right]$:
\begin{align*}
\rho_{it,s} & :=\sqrt{\textrm{Var}_{q}\left(\rho_{it}\right)}z_{s}+
    \mbe_{q}\left[\rho_{it}\right]\\
\mbe_{q}\left[\log\left(1-\frac{e^{\rho_{it}}}{1+e^{\rho_{it}}}\right)\right] &
    \approx\frac{1}{n_{MC}}\sum_{s=1}^{n_{MC}}
    \log\left(1-\frac{e^{\rho_{it,s}}}{1+e^{\rho_{it,s}}}\right)
\end{align*}
We found that increasing the number of points used for the quadrature
did not measurably change any of the results. The integration points
and weights were calculated using the \texttt{\textit{numpy.polynomial.hermite}}
module in Python \citep{scipy}.
 
\clearpage{}

\bibliography{variational_robustness}

\end{document}